\documentclass[11pt,uplatex]{article}
\usepackage{graphicx}
\usepackage{color}
\usepackage{amsmath, amssymb, amsthm, ascmac, caption, comment, enumerate, ascmac, bm}
\usepackage{latexsym, mathabx, mathrsfs, mathtools, psfrag, setspace, subcaption}
\usepackage{newtxtext}
\usepackage{natbib}
\usepackage[stable]{footmisc}
\usepackage{multirow}
\definecolor{mycolor}{cmyk}{0.80, 0.20, 0.25, 0}
\usepackage[colorlinks=true, linkcolor=mycolor, filecolor=mycolor, urlcolor=mycolor, citecolor=mycolor, driverfallback=dvipdfmx]{hyperref}
\usepackage[margin = 2.2cm]{geometry}

\DeclareMathOperator*{\argmin}{argmin}

\DeclareMathOperator*{\bbE}{\mathbb{E}}

\DeclareMathOperator*{\Cov}{\mathrm{Cov}}
\DeclareMathOperator*{\Var}{\mathrm{Var}}

\newcommand{\indep}{\mathrel{\text{\scalebox{1.07}{$\perp\mkern-10mu\perp$}}}}
\renewcommand{\hat}{\widehat}
\renewcommand{\tilde}{\widetilde}

\newcommand{\ol}{\overline}

\numberwithin{equation}{section}
 
\allowdisplaybreaks[2]

\theoremstyle{definition}
\newtheorem{theorem}{Theorem}[section]
\newtheorem{assumption}{Assumption}[section]
\newtheorem{definition}{Definition}[section]

\newtheorem{lemma}{Lemma}[section]
\newtheorem{proposition}{Proposition}[section]
\newtheorem{remark}{Remark}[section]

\title{Estimating a Continuous Treatment Model with Spillovers: A Control Function Approach}

\author{Tadao Hoshino\thanks{School of Political Science and Economics, Waseda University, 1-6-1 Nishi-waseda, Shinjuku-ku, Tokyo 169-8050, Japan. Email: \href{mailto:thoshino@waseda.jp}{thoshino@waseda.jp}.}}

\begin{document}

\maketitle

\begin{abstract}
    We study a continuous treatment effect model in the presence of treatment spillovers through social networks.
    We assume that one's outcome is affected not only by his/her own treatment but also by a (weighted) average of his/her neighbors' treatments, both of which are treated as endogenous variables.
    Using a control function approach with appropriate instrumental variables, we show that the conditional mean potential outcome can be nonparametrically identified.
    We also consider a more empirically tractable semiparametric model and develop a three-step estimation procedure for this model.
    As an empirical illustration, we investigate the causal effect of the regional unemployment rate on the crime rate.
\end{abstract}


\clearpage

\section{Introduction}\label{sec:intro}

Treatment evaluation under cross-unit interference is currently one of the most intensively studied topics in the causal inference literature (see, e.g., \citealp{halloran2016dependent, aronow2021spillover} for recent reviews). 
These previous studies have highlighted the importance of accounting for treatment spillovers from other units via empirical applications in many fields, including political science, epidemiology, education, economics, etc.
However, most of these studies have focused on the case of simple binary treatments, and estimating treatment spillover effects with continuous treatments has been rarely considered.
Even in the context of binary treatments, they often assume that the treatments are (conditionally) randomly assigned with full compliance (with some exceptions including \citealp{kang2016peer, imai2020causal, ditraglia2021identifying, hoshino2021causal, vazquez2021causal}).
To fill these gaps, this paper considers the estimation of a continuous treatment model with treatment spillovers while allowing the subjects to self-select their own treatment levels. 
 
For conventional continuous treatment models without spillovers, several approaches have been proposed to estimate the average \textit{dose-response} function -- the expected value of a potential outcome at particular treatment level.
Among others, \cite{hirano2004propensity} and \cite{imai2004causal} propose the generalized propensity score (GPS) approach under an unconfoundedness assumption.
A further extension of the GPS approach to the case of multiple correlated treatments is discussed in \cite{egger2013generalized}.
Then, it would be a natural idea to apply these methods to estimate continuous treatment and spillover effects jointly, as in \cite{del2020causal}, in which the authors use a GPS method to adjust for confounding variables for both own and others' treatments. 

However, the above-mentioned methods are not successful in the presence of unobservable confounding variables, which might lead to endogeneity issues not only with respect to own treatment but also to others' treatments.
For example, suppose the outcome is academic performance for high school students, and the treatment of interest is hours spent on in-home private tutoring.
Since close school mates are often taught by the same teachers and have similar socioeconomic backgrounds, they share a variety of unobservable attributes that affect the treatment and potential outcomes simultaneously.
In this situation, not only one's own treatment level but also the treatments of his/her friends should be treated as endogenous variables.
Other examples include the causal effect of minimum wage on the local employment rate, the effect of police force on the crime rate, and the effect of corporate tax subsidies on regional revenue.
These are cases where ``spatial'' treatment spillovers are of concern and are typical empirical situations that can fit into the framework of this paper.

To deal with this complex endogeneity problem without relying on strong parametric distributional restrictions, we assume a specific type of triangular model in this study, in which the treatment equation has a certain form of separability and others' treatments influence one's outcome in the form of peers' weighted average treatment, possibly with some monotonic transformation.
The latter assumption is similar to and more general than the \textit{mean interaction} in \cite{manski2013identification}, which assumes that the impacts from others can be summarized as the empirical mean of peers' treatments.
Since causal inference is generally impossible if no assumptions are imposed on the interference structure (cf. \citealp{imbens2015causal}), assumptions similar to the above are widely adopted in the literature.

With these assumptions, we address the endogeneity by employing a control function approach (see, e.g., \citealp{blundell2003endogeneity, florens2008identification, imbens2009identification}), which introduces auxiliary regressors, the \textit{control variables}, in the estimation of the outcome equation to eliminate the endogeneity bias. 
Under the availability of valid instrumental variables (IVs), the rank variable in the treatment equation can serve as the control variable for own treatment variable, which is a standard result in the literature.
Meanwhile, to control the endogeneity for the treatment spillovers, a straightforward choice for the control variables would be to use the peers' rank variables by analogy.
Although this approach is theoretically simple, it may not be practical unless the number of interacting partners is limited to one or two because of the curse of dimensionality.
As a novel finding of this paper, we demonstrate that the rank variable of a ``hypothetical'' individual who is on average equivalent to real peers can be used as a valid control variable by utilizing the additive structure and interaction structure of our model.
Since this rank variable is one dimensional, we can alleviate the dimensionality problem.

As is well-known, to achieve nonparametric identification of treatment parameters such as the average structural function based on a control function approach, we often require a strong support condition, such that the support of the control variable conditional on the treatment variable is equal to the marginal support of the control variable (see \citealp{imbens2009identification}).
This is true for our model as well, but such condition may be rarely satisfied in reality.
Thus, to improve the empirical tractability, we introduce additional functional form restrictions for estimation, similar to \cite{chernozhukov2020semiparametric} and \cite{newey2021control}.
The resulting model takes the form of a multiplicative functional-coefficient regression model, which can be seen as a special case of the so-called varying-coefficient additive models (see, e.g., \citealp{zhang2015varying,hu2019estimation}).
For this model, as a main parameter of interest, we focus on the estimation of the conditional average treatment response (CATR) -- the expectation of the potential outcome conditional on individual covariates.

For the estimation of the CATR parameter, we preliminarily need to estimate two control variables, one for own treatment effect and the other for the treatment spillover effect.
This preliminary step can be easily implemented with a composite quantile regression (CQR) method (see, e.g., \citealp{zou2008composite}).
Once the control variables are obtained, the CATR is estimated in three steps.
In the first step, we globally estimate the model without considering the multiplicative structure using a penalized series regression approach.
Next, we fine-tune the bias-correction function in the model by marginal integration to improve the efficiency.
Given this result, we finally re-estimate the model using a local linear kernel regression.
It is well-known that, 
to achieve pointwise asymptotic normality for a series estimator, we are often required to inefficiently undersmooth the estimator (cf. \citealp{huang2003local}).
For this reason, the idea of first obtaining a globally consistent estimate using a series method and then locally re-adjusting the estimate using a kernel method is widely adopted in the literature (see, e.g., \citealp{horowitz2004nonparametric, horowitz2005nonparametric, wang2007spline}).
With this multiple-step estimation procedure, we can show that the final CATR estimator is oracle, in that its asymptotic distribution is not affected by the estimation of the other nuisance functions.

To demonstrate the empirical usefulness of the proposed method, we conduct an empirical analysis on the causal effects of unemployment on crime.
In this empirical analysis, using Japanese city-level data, we estimate the CATR parameter, with the crime rate as the outcome of interest and each city's own unemployment rate and the average unemployment rate of its neighboring cities as the treatments.
To account for the endogeneity, we use the availability of childcare facilities as an instrument for the unemployment rate.
As a result, we find that the CATR tends to increase as the unemployment rate increases, which is a common finding in the literature.
Moreover, our model reveals new empirical evidence that the spillover effect of unemployment from the neighboring cities is important only for non-rural cities. 

To summarize, the main contributions of this study are fourfold.
\textit{First}, to our knowledge, this study is the first to address a continuous treatment model in which both treatment spillover and endogeneity are present.
\textit{Second}, we propose a novel control function approach to establish the nonparametric identification of this model under some functional form restrictions.
\textit{Third}, considering the empirical feasibility, we propose a semiparametric multiplicative potential outcome model and develop a three-step estimation procedure for it, which is of independent interest in semiparametric estimation theory. 
\textit{Fourth}, by applying the proposed method to Japanese city data, we provide new empirical evidence on the relationship between the local unemployment and crime rates.

\paragraph{Paper organization:}
In Section \ref{sec:model}, we present our model and discuss nonparametric identification of the CATR parameter.
In Section \ref{sec:semiparametric}, we introduce an empirically tractable semiparametric model and propose our three-step estimation procedure.
The asymptotic properties of the estimator are also presented in this section.
The empirical analysis on the Japanese crime data is presented in Section \ref{sec:empiric}, and Section \ref{sec:conclusion} concludes.
The proofs of technical results, Monte Carlo simulation results, and supplementary information on the empirical analysis are all summarized in the online supplementary material.

\paragraph{Notation:}
For natural numbers $a$ and $b$, $I_a$ denotes an $a \times a$ identity matrix, and $\mathbf{0}_{a \times b}$ denotes a matrix of zeros of dimension $a \times b$.
For a matrix $A$, $|| A ||$ denotes the Frobenius norm.
If $A$ is a square matrix, we use $\rho_{\max} (A)$ and $\rho_{\min} (A)$ to denote its largest and smallest eigenvalues, respectively. 
We use $\otimes$ and $\circ$ to denote the Kronecker and Hadamard (element-wise) product, respectively.
For a set $\mathcal{A}$, $|\mathcal{A}|$ denotes its cardinality. 
We use $c$ (possibly with a subscript) to denote a generic positive constant whose value may vary in different contexts.
For random variables $X$ and $Y$, $X \indep Y$ means that they are independent.
Lastly, for positive sequences $a_n$ and $b_n$, $a_n \asymp b_n$ means that there exist $c_1$ and $c_2$ such that $0 < c_1 \le a_n/b_n \le c_2 < \infty$ for sufficiently large $n$.

\section{Nonparametric Identification}\label{sec:model}

Suppose that we have a sample of $N$ agents that form social networks whose connections are represented by an $N \times N$ adjacency matrix $\bm{A}_N = (A_{i,j})_{1 \le i,j \le N}$.
These agents can be individuals, firms, or municipalities depending on the context.
The networks can be directed, that is, regardless of the value of $A_{j,i}$, we may observe $A_{i,j} = 1$ if $j$ affects $i$ and $A_{i,j} = 0$ otherwise.
The diagonal elements of $\bm{A}_N$ are all zero.
Throughout the paper, we treat $\bm{A}_N$ as non-random.
For each $i$, we denote $i$'s ``reference group'' (peers, colleagues, neighbors, etc) as $\mathcal{P}_i \coloneqq \{1 \le j \le N: A_{i,j} = 1\}$ and its size as $n_i = |\mathcal{P}_i|$.
To simplify the discussion, we assume that $n_i > 0$ for all $i$.
In addition, we write $\ol{\mathcal{P}}_i \coloneqq i \cup \mathcal{P}_i$.
For a general variable $Q$, we denote $\bm{Q}_{\mathcal{P}_i} \coloneqq \{Q_j\}_{j \in \mathcal{P}_i}$.
We define $\bm{Q}_{\ol{\mathcal{P}}_i}$ similarly.

Let $T \in \mathcal{T}$ denote the continuous treatment variable of interest.
We assume that the support $\mathcal{T}$ is common across all individuals and is a real closed interval of $\mathbb{R}$.
The outcome variable of interest is $Y \in \mathbb{R}$.
In this study, we explicitly allow that the peers' treatments $\bm{T}_{\mathcal{P}_i}$ influence on $i$'s own outcome $Y_i$ via a known real-valued function $S_i$: 
\begin{align*}
    S_i(\bm{T}_{\mathcal{P}_i}) \coloneqq m_i\left( \sum_{j \in \mathcal{P}_i} a_{i,j} T_j\right),
\end{align*}
where $m_i$ is a strictly increasing continuous function which may depend on $i$, and $a_{i,j}$'s are weight terms satisfying $a_{i,j} > 0$ if and only if $j \in \mathcal{P}_i$ and $\sum_{j \in \mathcal{P}_i} a_{i,j} = 1$.\footnote{
    Assuming that the functional form of $S_i$ is known a priori is a strong assumption.
    In the literature, there are several studies that investigate under what conditions we can make meaningful causal inference even when the ``exposure'' function (i.e., $S_i$) is unknown or mis-specified (e.g., \citealp{hoshino2021causal, savje2021average, leung2022causal, VAZQUEZBARE2022}).
    All these studies consider only binary treatment situations, and whether we can establish similar results for continuous treatment models would be an important open question.
}
For example, if $m_i$ is an identity function and $a_{i,j} = n_i^{-1}$, then $S_i$ simply returns the reference-group average.
This form of treatment spillover is the most commonly used in empirical studies on peer effects.
If $m_i(t) = n_i t$ instead, then $S_i$ in this case is the sum of peers' treatments.
For another example, in the spatial statistics literature, researchers often assume that $a_{i,j}$ is inversely proportional to the geographical distance between $i$ and $j$.
The support of $S_i$ is denoted as $\mathcal{S}_i$, which varies across individuals in general because of the transformation $m_i$ that may be specific to $i$.
As a special case, if $m_i$ is an identity function, then we have the same support for all $i$ as $\mathcal{S}_i = \mathcal{T}$.

Now, letting $X_i = (X_{1,i}, \ldots, X_{dx, i})^\top$ be a $dx \times 1$ vector of observed individual covariates, we suppose that the outcome $Y_i$ is determined in the following equation:
\begin{align}\label{eq:outcome_eq}
    Y_i = y(T_i, S_i, X_i, \epsilon_i) & \;\; \text{for} \;\; i = 1, \ldots, N
\end{align}
where $y$ is an unknown function, $S_i = S_i(\bm{T}_{\mathcal{P}_i})$, and $\epsilon_i$ is an unobservable determinant of $Y_i$.
Note that $\epsilon_i$ could be a vector of unknown dimension.
The non-separable setting is essential for a treatment effect model since it permits general interactions between $(T_i, S_i)$ and $\epsilon_i$, allowing for treatment heterogeneity among observationally identical individuals.
The potential outcome when $(T_i, S_i) = (t, s) \in \mathcal{T}\mathcal{S}_i$ is written as $Y_i(t, s) = y(t, s, X_i, \epsilon_i)$.
The target parameter of interest in this study is the conditional average treatment response (CATR) function:\footnote{
    Following the econometrics literature (e.g., \citealp{blundell2003endogeneity}), this parameter may be called the average structural function (conditional on $X_i = x$).
    In the causal inference literature, it is also known as the average dose response function.
    }
\begin{align*}
    \mathrm{CATR}(t, s, x) 
    & \coloneqq \bbE [Y_i(t, s) \mid X_i = x].
\end{align*}

Next, suppose that we have a $dz \times 1$ vector of IVs, $Z_i = (Z_{1,i}, \ldots, Z_{dz,i})^\top$, which determines the value of $T_i$ through the following equation:
\begin{align}\label{eq:treament_eq}
    T_i = \pi(X_i, Z_i) + \eta(U_i),
\end{align}
where $\pi$ and $\eta$ are unknown functions, and $U_i$ is a scalar unobservable variable.
We allow that not only $U_i$ but also $\bm{U}_{\mathcal{P}_i}$ are potentially correlated with $\epsilon_i$, which are the sources of endogeneity for $T_i$ and $S_i$, respectively.
To address the endogeneity issue, we make the following assumptions: 

\begin{assumption}\label{as:exclusion}
    \textbf{(i)} $\bm{Z}_{\ol{\mathcal{P}}_i} \indep (\bm{U}_{\ol{\mathcal{P}}_i}, \epsilon_i) \mid \bm{X}_{\ol{\mathcal{P}}_i}$;
    \textbf{(ii)} the function $u \mapsto \eta(u)$ is continuous and strictly increasing, and $U_i$ is distributed as $\mathrm{Uniform}[0,1]$ conditional on $\bm{X}_{\ol{\mathcal{P}}_i}$.
\end{assumption}

Assumption \ref{as:exclusion}(i) is the exogeneity assumption that own and peers' IVs are independent of the unobservables given their covariates.
This assumption would be particularly plausible in an experimental setup where $Z_i$ is a randomly assigned treatment inducement. 
Note that we do not preclude the cross-sectional correlation of the covariates, IVs, and unobservables, which should be an important feature of network data.
For Assumption \ref{as:exclusion}(ii), we emphasize that the functional form of $\eta$ is independent of $\bm{X}_{\ol{\mathcal{P}}_i}$.
Note that although this independence restriction is not strictly necessary for a consistent estimation of the treatment equation \eqref{eq:treament_eq},\footnote{
    When considering a more general treatment equation $T_i = \pi(X_i, Z_i) + \eta(U_i, \bm{X}_{\ol{\mathcal{P}}_i})$, the second part of Assumption \ref{as:exclusion}(ii) is always satisfied using $\tilde U_i$ and $\tilde \eta$, where $\tilde U_i = F_{U_i|\bm{X}_{\ol{\mathcal{P}}_i}}(U_i \mid \bm{X}_{\ol{\mathcal{P}}_i})$ and $\tilde \eta(\tilde U_i, \bm{X}_{\ol{\mathcal{P}}_i}) = \eta(F_{U_i|\bm{X}_{\ol{\mathcal{P}}_i}}^{-1}(\tilde U_i \mid \bm{X}_{\ol{\mathcal{P}}_i}))$.
    Under Assumption \ref{as:exclusion}(i), we have $\Pr(T_i \le \pi(X_i, Z_i) + \eta(u, \bm{X}_{\ol{\mathcal{P}}_i}) \mid Z_i, \bm{X}_{\ol{\mathcal{P}}_i}) = \Pr(U_i \le u \mid \bm{X}_{\ol{\mathcal{P}}_i}) = u$ for any $u \in (0,1)$, which implies that $\pi(X_i, Z_i) + \eta(u, \bm{X}_{\ol{\mathcal{P}}_i})$ can be estimated in a quantile regression framework.
    Furthermore, by checking whether the estimated conditional quantile function varies with the peers' covariates, we may be able to test the validity of assuming $\eta(U_i, \bm{X}_{\ol{\mathcal{P}}_i}) = \eta(U_i)$; however, directly implementing this would suffer from the curse of dimensionality. 
}
we introduce it for constructing an effective control variable for the spillover effect.

Recall that the disturbance term in the outcome equation is allowed to be a vector and to enter the model in a non-additive way.
Thus, we can see that the identification result heavily depends on the functional form of the treatment equation rather than that of the outcome equation, in line with the existing literature.

\bigskip

Below, we formally discuss the identification of the CATR parameter. 
Throughout this section, we use the term ``identification'' to indicate that the parameter of interest can be characterized through a moment of the observable random variables.
Note that this does not automatically imply the estimability of the parameters because network data naturally follow a nonidentical and dependent data distribution. \footnote{
    In this sense, the CATR parameter should be indexed by ``$\:i\:$,'' however, we suppress it for notational simplicity.
}
To estimate the CATR parameter, as discussed later, we require some additional stationarity assumptions.

As is well-known, we can deal with the endogeneity of $T_i$ by including $U_i$ as a control variable in the estimation of outcome equation \eqref{eq:outcome_eq} (e.g., \citealp{blundell2003endogeneity, florens2008identification, imbens2009identification}).
For the identification of $U_i$, in Lemma \ref{lem:UandV}, we prove that $\pi_i(X_i,Z_i)$ and $\eta(u)$ for any $u \in (0,1)$ are identifiable up to a location shift (see Remark \ref{rem:CQR} as well).
When these parameters are treated as known, $U_i$ can be obtained by solving $\min_{u \in [0,1]}|T_i - \pi(X_i,Z_i) - \eta(u)|$ under the monotonicity and continuity of $\eta$.

Similarly, to account for the endogeneity of $S_i$, we need to find a control variable(s) for it.
One obviously valid candidate is to use $\bm{U}_{\mathcal{P}_i}$.
Once the behaviors of $\bm{U}_{\mathcal{P}_i}$ are controlled, the variation of $S_i$ comes only from $(\bm{X}_{\mathcal{P}_i}, \bm{Z}_{\mathcal{P}_i})$, which are conditionally independent of the disturbances by Assumption \ref{as:exclusion}(i).
However, this approach is not practical unless $n_i$ is limited to one or two, because of the curse of dimensionality. 
In the next proposition, we show that under the additive separability in \eqref{eq:treament_eq} we can construct a one-dimensional control function for $S_i$.

\begin{proposition}\label{prop:S_control}
    Under Assumption \ref{as:exclusion}, $V_i \coloneqq \{v \in [0,1] \mid \sum_{j \in \mathcal{P}_i} a_{i,j} \eta(U_j) = \eta(v)\}$ is unique and is a valid control variable for $S_i$.
\end{proposition}
Proposition \ref{prop:S_control} implies that, for a general function $f$ of $\epsilon$, we have 
\begin{align}\label{eq:control}
    \begin{split}
    & \bbE[f(\epsilon_i) \mid T_i = t, S_i = s, U_i = u, V_i = v, \bm{X}_{\ol{\mathcal{P}}_i} = \bm{x}] \\
    & \quad = \bbE\left[ f(\epsilon_i) \mid \pi(x_i, Z_i) = c_1, \textstyle \sum_{j \in \mathcal{P}_i}a_{i,j} \pi(x_j, Z_j) = c_2, U_i = u, V_i = v , \bm{X}_{\ol{\mathcal{P}}_i} = \bm{x} \right] \\ 
    & \quad = \bbE[ f(\epsilon_i) \mid U_i = u, V_i = v, \bm{X}_{\ol{\mathcal{P}}_i} = \bm{x}],
    \end{split}
\end{align}
where $c_1 = t- \eta(u)$ and $c_2 = m_i^{-1}(s) - \eta(v)$, by Assumption \ref{as:exclusion}(i).
The identification of $V_i$ can be achieved by solving $\min_{v \in [0,1]}|\sum_{j \in \mathcal{P}_i} a_{i,j} (T_j - \pi(X_j, Z_j))  - \eta(v)|$.
Without the additive separability between $(X_i,Z_i)$ and $U_i$ as in \eqref{eq:treament_eq}, $V_i$ generally depends on the instruments $\bm{Z}_{\mathcal{P}_i}$, and conditioning on $V_i$ also restricts their behavior (cf. \citealp{kasy2011identification}), leading to a failure of establishing \eqref{eq:control}.

\begin{remark}[Interpretation of $V_i$]
    We can interpret the control variable $V_i$ as the rank variable for a hypothetical friend of $i$ who is observationally equal to the weighted average of $i$'s friends. 
    Note that this interpretation essentially comes from the assumption that $S_i$ is a scalar-valued function.
    That is, in our model, having multiple friends whose treatments are on average equal to $S_i$ is indistinguishable from having only one friend whose treatment level is exactly $S_i$.
    Note that simply using $\sum_{j \in \mathcal{P}_i} a_{i,j} \eta(U_j)$ as a control variable theoretically works, but in general, its support is unbounded.
    Unbounded control variables are less practical because, as presented below, our estimation procedure involves computing several integrals with respect to the control variables.
    In addition, a more delicate discussion is necessary to establish desirable convergence results for the functions of the control variables.
\end{remark}

Now, we define the marginal treatment response (MTR) function: 
\begin{align*}
    \mathrm{MTR}(t, s, u, v, \bm{x}) \coloneqq \bbE [Y_i(t, s) \mid U_i = u, V_i = v, \bm{X}_{\ol{\mathcal{P}}_i} = \bm{x}]. 
\end{align*}
From \eqref{eq:control}, we have
\begin{align}\label{eq:mtr_ident}
\begin{split}
    \bbE [Y_i \mid T_i = t, S_i = s, U_i = u, V_i = v, \bm{X}_{\ol{\mathcal{P}}_i} = \bm{x}] = \mathrm{MTR}(t, s, u, v, \bm{x}).
\end{split}
\end{align}
Hence, we can identify $\mathrm{MTR}(t, s, u, v, \bm{x})$ as the left-hand side of \eqref{eq:mtr_ident}.
Moreover, if $\mathrm{MTR}(t, s, u, v, \bm{x})$ is identifiable uniformly over $\mathrm{supp}( U_i, V_i, \bm{X}_{\mathcal{P}_i} \mid X_i = x)$, we can identify $\mathrm{CATR}(t, s, x)$ by
\begin{align}\label{eq:atr_ident}
    \mathrm{CATR}(t, s, x) 
    & = \bbE\left[ \int_0^1 \int_0^1 \mathrm{MTR}(t, s, u, v, \bm{X}_{\ol{\mathcal{P}}_i}) f_{U_i V_i |\bm{X}_{\ol{\mathcal{P}}_i}}(u,v \mid \bm{X}_{\ol{\mathcal{P}}_i})\mathrm{d}u\mathrm{d}v \; \bigg| \; X_i = x\right],
\end{align}
where the outer expectation is with respect to $\bm{X}_{\mathcal{P}_i}$ conditional on $X_i = x$, and $f_{U_i V_i|\bm{X}_{\ol{\mathcal{P}}_i}}$ is the conditional density function of $(U_i, V_i)$ given $\bm{X}_{\ol{\mathcal{P}}_i}$.
For $\mathrm{MTR}(t, s, u, v, \bm{x})$ to be well-defined on the entire $\mathrm{supp}( U_i, V_i, \bm{X}_{\mathcal{P}_i} \mid X_i = x)$, we need the following additional condition:
\begin{assumption}\label{as:full_support}
    $\mathrm{supp}( U_i, V_i, \bm{X}_{\mathcal{P}_i} \mid T_i = t, S_i = s, X_i = x) = \mathrm{supp}( U_i, V_i, \bm{X}_{\mathcal{P}_i} \mid X_i = x)$.
\end{assumption}
From \eqref{eq:mtr_ident}, we can see that without any support restriction, $\mathrm{MTR}(t, s, u, v, \bm{x})$ can be obtained only on $\mathrm{supp}(U_i, V_i, \bm{X}_{\mathcal{P}_i} \mid T_i = t, S_i = s, X_i = x)$.
Assumption \ref{as:full_support} requires that conditioning on $\{T_i = t, S_i = s\}$ should not affect the support of $(U_i, V_i, \bm{X}_{\mathcal{P}_i})$ conditional on $X_i = x$.
An important implication of this condition is that the IV must have very rich variation so that the values of $T_i$ and $S_i$ do not restrict the range of values that $U_i$ and $V_i$ can take.
We summarize the result obtained so far in the next theorem.

\begin{theorem}\label{thm:atr_ident}
    Suppose that Assumptions \ref{as:exclusion} and \ref{as:full_support} hold.
    Then, $\mathrm{CATR}(t, s, x)$ is identified through \eqref{eq:mtr_ident} -- \eqref{eq:atr_ident}.
\end{theorem}

In order to construct an estimator for $\mathrm{CATR}(t, s, x)$ based on Theorem \ref{thm:atr_ident} in practice, we need to cope with several obstacles.
The first issue is the full-support condition in Assumption \ref{as:full_support}.
Finding IVs ensuring such a support condition would be quite difficult in reality.
Even without this support condition, if there are informative upper and lower bounds of $Y_i$, it would be possible to partially identify $\mathrm{CATR}(t, s, x)$, as in Theorem 4 of \cite{imbens2009identification}.
For another more convenient approach to point-identify the CATR in the absence of Assumption \ref{as:full_support}, in the next section and thereafter, we consider introducing additional functional form restrictions on the outcome equation.

Another issue is that, since the size of the reference group and the weights $a_{\mathcal{P}_i}$ are not necessarily identical among individuals, the support $\mathcal{S}_i$ of $S_i$ and the distribution of $V_i$ should vary with individuals in general.
The dimension and distribution of $\bm{X}_{\mathcal{P}_i}$ are also obviously heterogeneous among individuals.
These heterogeneities are problematic in the estimation stage.
To circumvent this problem, we simply restrict our attention to a specific subsample $\mathcal{N}'$ such that those in this subsample have the same support $\mathcal{S}$ and are homogenous (in the sense of Assumption \ref{as:model}(ii) below).
The simplest but empirically the most typical example would be $\mathcal{N}' = \{1 \le i \le N: m_i = m, \: n_i = c, \: a_{i,j} = c^{-1}\mathbf{1}\{j \in \mathcal{P}_i\}\}$ for some monotonically increasing $m$ and an integer $c \ge 1$.
The size of this subsample is denoted as $n = |\mathcal{N}'|$.
In the following, without loss of generality, we re-label the $N$ agents so that the first $n$ individuals belong to this subsample, that is, $\mathcal{N}' = \{1, \ldots, n\}$.

\begin{remark}[Alternative identification strategies]
    For a conventional continuous treatment model without spillovers, \cite{d2015identification} and \cite{torgovitsky2015identification} prove that the outcome equation can be identified even when only discrete IVs are available, without introducing strong functional form restrictions.
    Their arguments are essentially based on the two assumptions that $\epsilon_i$ is one-dimensional and monotonically related to $Y_i$ and that the distributions of the treatment variable conditional on different values of the IV must have an intersection.
    Note that the former condition does not directly fit in our model \eqref{eq:outcome_eq}, where $\epsilon_i$ is allowed to be multi-dimensional.
    In addition, as pointed out by \cite{ishihara2021partial}, the latter condition on the conditional distributions of the treatment can be violated in some empirical settings.
\end{remark}

\section{Semiparametric Estimation and Asymptotic Properties}\label{sec:semiparametric}

In this section, we propose our estimation procedure for $\mathrm{CATR}(t, s, x)$ in a specific semiparametric model, and investigate its asymptotic properties as $n$ increases to infinity.
To increase $n$, we consider a sequence of networks $\{ \bm{A}_N \}$, where $\bm{A}_N$ is an $N \times N$ adjacency matrix.
Since each $\bm{A}_N$ is assumed to be non-stochastic, our analysis can be interpreted as being conditional on the realization of $\bm{A}_N$.
In this sense, the model and parameters presented in this study should be viewed as triangular arrays defined along the network sequence.
    
\subsection{A semiparametric model and estimation procedure}

In \cite{newey2021control}, they proposed the following type of potential outcome model as a baseline (in our notation): $Y_i(t,s) = g(t, s)^\top \epsilon_i$, where $g(t, s)$ is a vector of transformations of $(t,s)$; for example, $g(t,s) = (1,t,s)^\top$ (note that they do not consider models with two treatment variables).
A similar parametric assumption was adopted in \cite{chernozhukov2020semiparametric} as well.
In this study, we generalize their models such that $g(t,s)$ is a general nonparametric function.
In addition, we further extend their work by allowing the treatment effect to vary with the individual characteristics $X_i$.
At the same time, to preserve empirical tractability, we assume that there are no interaction effects among $(X_i, \epsilon_i)$, and that $\epsilon_i$ is one-dimensional.
Consequently, we focus on the following outcome model:
\begin{align}\label{eq:semiparametric_model}
\begin{split}
    Y_i = X_i^\top \beta(T_i, S_i) + g(T_i, S_i) \epsilon_i \;\; \text{for} \;\; i = 1, \ldots, n
\end{split}
\end{align}
where $\beta(t,s) = (\beta_1(t,s), \ldots, \beta_{dx}(t,s))^\top$ and $g(t,s)$ are unknown functions to be estimated.\footnote{
    If we add one more continuous treatment, then $\beta$'s and $g$ will become three-dimensional functions, which are practically difficult to estimate because of the curse of dimensionality.
    To estimate such a model, we would need to introduce additional functional form restrictions on \eqref{eq:semiparametric_model}.
}
For normalization, we assume that $X_i$ includes a constant term so that $\bbE[\epsilon_i] = 0$ holds.
We do not additionally restrict the treatment equation \eqref{eq:treament_eq}, but we assume that the control variables $U_i$ and $V_i$ can be consistently estimated at a certain convergence rate (see Remark \ref{rem:kappa} below).
For further simplification, we strengthen Assumption \ref{as:exclusion}(i) as follows:

\setcounter{assumption}{-1}
\begin{assumption}\label{as:model}
    \textbf{(i)} $(\bm{Z}_{\ol{\mathcal{P}}_i}, \bm{X}_{\ol{\mathcal{P}}_i}) \indep (\bm{U}_{\ol{\mathcal{P}}_i}, \epsilon_i)$; and
    \textbf{(ii)} $(U_i,V_i)$ are identically distributed and $\bbE [ \epsilon_i \mid U_i, V_i] = \lambda(U_i, V_i)$ for all $i \in \mathcal{N}'$.
\end{assumption}

Recalling the definition of $V_i$, the first part of Assumption \ref{as:model}(ii) essentially requires that the joint distribution of $\bm{U}_{\ol{\mathcal{P}}_i}$ is stationary across all $i \in \mathcal{N}'$ and they adopt the same weighting scheme.
For the second part, we do not require the treatment heterogeneity terms to be identically distributed.
Since Assumption \ref{as:model} is fundamental to derive our estimation procedure, it will be assumed implicitly throughout the whole subsequent discussion.

Now, the CATR for this model at $(t,s,x)$ is simply given by 
\begin{align*}
    \textrm{CATR}(t,s,x) = x^\top \beta(t, s),
\end{align*}
which in fact coincides with the ``unconditional'' expectation $\bbE[y(t,s,x,\epsilon_i)]$.
Thus, the task of estimating $\textrm{CATR}(t,s,x)$ is greatly simplified to the estimation of $\beta$, and the support condition in Assumption \ref{as:full_support} is not necessary to recover the CATR.
Further, an analogous argument to \eqref{eq:control} gives 
\begin{align*}
    \bbE [Y_i \mid T_i = t, S_i = s, U_i = u, V_i = v, X_i = x]
    & = x^\top \beta(t, s) + g(t, s) \bbE [ \epsilon_i \mid U_i = u, V_i = v] \\
    & = x^\top \beta(t, s) + g(t, s) \lambda(u, v).
\end{align*}
This yields the following semiparametric multiplicative regression model: 
\begin{align}\label{eq:regression_model}
Y_i = X_i^\top \beta(T_i, S_i) + g(T_i, S_i) \lambda(U_i, V_i) + \tilde\epsilon_i \;\; \text{for} \;\; i = 1, \ldots, n.
\end{align}
By construction, we have $\bbE[\tilde\epsilon_i \mid T_i, S_i, U_i, V_i, X_i] = 0$.
It is clear that because of the multiplicative structure in \eqref{eq:regression_model}, in order to identify $g$ and $\lambda$ separately, we need some functional form normalizations.
First, recall that the condition $\bbE[\epsilon_i] = 0$ must be maintained.
This implies the following location normalization:
\begin{align}\label{eq:loc_normal}
    \bbE[\epsilon_i] = \bbE[\lambda( U_i, V_i)] = 0.
\end{align}
We further need some scale normalization for identification.
To this end, we set
\begin{align}\label{eq:scale_normal}
    \int_{\mathcal{TS}}g(t,s)\mathrm{d}t\mathrm{d}s = 1.
\end{align}

Regression models with a multiplicative structure as in \eqref{eq:regression_model} can be found in the literature in different contexts (e.g., \citealp{linton1995kernel, zhang2015varying, hu2019estimation, chen2020estimation}).
A typical approach to estimating this type of model is to apply some marginal integration (to the estimated nonparametric function or to the data itself in advance of the estimation).
However, note that none of these prior studies considered a functional-coefficient specification as in ours.
Thus, we need to develop a new estimation procedure for our model, which should be of independent interest in the semiparametric estimation literature.
Specifically, we deal with the multiplicative structure by nonparametrically estimating the model ignoring the functional form restriction in the first step, and then splitting the estimated function into two multiplicative components by marginal integration in the second step in a similar manner to \cite{chen2020estimation}.\footnote{
    Although it is possible to estimate the model in a single step by estimating $g$ and $\lambda$ separately from the beginning, the resulting method requires solving a high-dimensional non-convex optimization problem, which is computationally challenging, as pointed out in \cite{zhang2020new}.
    To circumvent this issue, \cite{zhang2020new} proposed an iterative computational algorithm.
    }
The first stage estimation involves a four-dimensional nonparametric regression on $(T_i,S_i, U_i, V_i)$, which often produces unstable estimates under a moderate sample size.
Thus, we consider using a penalized regression in this step.

The whole estimation procedure is as follows.
Before estimating the regression model in \eqref{eq:regression_model}, we need to consistently estimate the function $\pi$ and $\eta$ to obtain consistent estimates of $(U_i, V_i)$.
Excluding this preliminary step, our procedure for estimating the CATR consists of three steps.
In the first step, we estimate the model in \eqref{eq:regression_model} ignoring the multiplicative structure using a penalized series regression approach.
In the next step, we obtain estimates of $g$ and $\lambda$ by marginally integrating the nonparametric function obtained in the previous step.
Finally, we re-estimate the coefficient functions $\beta$ using a local linear kernel regression.
With these two additional steps, we can achieve an oracle property for the estimation of the CATR parameter (see Theorem \ref{thm:catr}).
The proposed estimation procedure is new to the literature but may be viewed as a minor extension of a commonly used two-step series-kernel estimation method, in which \cite{chen2020estimation}'s marginal integration step is inserted between the series and kernel estimation steps.

\paragraph{Preliminary step: control variables estimation}

Under Assumptions \ref{as:exclusion}(ii) and \ref{as:model}(i), it holds that $\Pr(T_i \le \pi(X_i, Z_i) + \eta(u) \mid X_i, Z_i) = u$ for any $u \in (0,1)$.
This implies that we can estimate $\pi(X_i, Z_i)$ and $\eta(u)$ using the CQR method: for pre-specified $0 < u_1 < u_2 < \cdots < u_L < 1$,
\begin{align*}
    \hat \eta_n(u_1), \ldots, \hat \eta_n(u_L), \hat \pi_n \coloneqq \argmin_{b_1, \ldots, b_L, \varpi}\sum_{l = 1}^L \sum_{i = 1}^n q_{u_l}(T_i - \varpi(X_i, Z_i) - b_l)
\end{align*}
subject to some location normalization (e.g., $\eta(0.5) = 0$), where $q_u(x) \coloneqq x(u - \mathbf{1}\{x < 0\})$.
Then, we compute the residual, $res_i \coloneqq T_i - \hat \pi_n(X_i, Z_i)$, for each $i$. 
Note that $res_i$ is an estimator of $\eta(U_i)$, whose convergence rate is governed by that for $\hat \pi_n$.
The estimator for $U_i$ can be obtained by $\hat U_i \coloneqq \argmin_{u \in [0,1]}|res_i - \hat \eta_n(u)|$.
Similarly, we can estimate $V_i$ by $\hat V_i \coloneqq \argmin_{v \in [0,1]}| \sum_{j \in \mathcal{P}_i} a_{i,j} res_j - \hat \eta_n(v)|$.
In practice, these minimization problems are solved by grid search with a sufficiently large $L$.

\begin{remark}[Identification and estimation of the treatment equation]\label{rem:CQR}
    Here, we briefly comment on the identification and estimation of the treatment equation.
    If one considers estimating $\pi$ without explicit functional form assumptions, such a model has been investigated in \cite{kai2010local}, where they applied a local polynomial smoothing to the CQR problem.
    As shown in \cite{kai2010local}, we can estimate $\pi$ with the standard nonparametric convergence rate under mild conditions.
    However, as we demonstrate later (see Remark \ref{rem:kappa} as well), the full nonparametric estimation of $\pi$ is generally unacceptable for achieving the desirable asymptotic behavior for our CATR estimator.
    Thus, in the numerical studies in this paper, we use a linear model specification: $\pi(X_i, Z_i) = X_i^\top\gamma_x + Z_i^\top\gamma_z$, which corresponds to the model originally considered in \cite{zou2008composite}.
    They showed that the coefficients $(\gamma_x, \gamma_z)$ can be estimated at the $\sqrt{n}$ rate under the standard linear independence condition on $(X_i,Z_i)$.
    One may consider a semiparametric CQR model as an intermediate case of these two (e.g., \citealp{kai2011new}).
\end{remark}

\paragraph{First step: penalized series estimation}

Let $\{(p_{T,j}(t), p_{S,j}(s), p_{U,j}(u), p_{V,j}(v)): j = 1, 2, \ldots \}$ be the basis functions on $\mathcal{T}$, $\mathcal{S}$, $[0,1]$, and $[0,1]$, respectively.
We define $P_T(t) \coloneqq (p_{T,1}(t), \ldots , p_{T,K_T}(t))^\top$, $P_S(s) \coloneqq (p_{S,1}(s), \ldots , p_{S,K_S}(s))^\top$, and $P_{TS}(t, s) \coloneqq P_T(t) \otimes P_S(s)$.
Then, we consider series approximating $\beta_l(\cdot, \cdot)$ by $\beta_l(t, s) \approx P_{TS}(t, s)^\top \theta_{\beta_l}$ for some $K_{TS} \times 1$ coefficient vector $\theta_{\beta_l}$, for each $l = 1, \ldots, dx$, where $K_{TS} \coloneqq K_T K_S$.
Similarly, define $P_U(u) \coloneqq (p_{U,1}(u), \ldots , p_{U,K_U}(u))^\top$, $P_V(v) \coloneqq (p_{V,1}(v), \ldots , p_{V,K_V}(v))^\top$, $P_{UV}(u, v) \coloneqq P_U(u) \otimes P_V(v)$, and $K_{UV} \coloneqq K_U K_V$.
We assume that $K_{TS}$ and $K_{UV}$ increase as $n$ increases at the same speed so that there exists an increasing sequence $\{\kappa_n\}$ satisfying $(K_{TS}, K_{UV}) \asymp \kappa_n$.

For the estimation of $g(t,s)\lambda(u,v)$, to impose the location condition \eqref{eq:loc_normal} in the estimation, we would like to normalize the basis function such that $\ol{P}_{UV}(u, v) \coloneqq P_{UV}(u, v) - \bbE[P_{UV}(U, V)]$ (recall that we have been focusing on a subsample in which $\{(U_i, V_i)\}$ are identically distributed).
However, since both $\bbE[P_{UV}(U, V)]$ and $(U_i, V_i)$ are unknown, we instead use
\begin{align*}
    \hat P_{UV}(u, v) \coloneqq P_{UV}(u, v) - \frac{1}{n}\sum_{i=1}^n P_{UV}(\hat U_i, \hat V_i).
\end{align*}
Then, letting $\hat P(t, s, u, v) \coloneqq P_{TS}(t, s) \otimes \hat P_{UV}(u, v)$, we consider approximating $g(t,s)\lambda(u,v) \approx \hat P(t, s, u, v)^\top \theta_{g\lambda}$ with a $K_{TSUV} \times 1$ coefficient vector $\theta_{g\lambda}$, where $K_{TSUV} \coloneqq K_{TS} K_{UV}$.
Using these approximations, we have 
\begin{align}\label{eq:series_approx}
    Y_i \approx P_X(X_i, T_i, S_i)^\top \theta_\beta + \hat P(T_i, S_i, \hat U_i, \hat V_i)^\top \theta_{g\lambda} + \tilde\epsilon_i,  \;\; \text{for} \;\; i = 1, \ldots, n
\end{align}
where $P_X(X_i, T_i, S_i) \coloneqq X_i \otimes P_{TS}(T_i, S_i)$, and $\theta_\beta = (\theta_{\beta_1}^\top, \ldots, \theta_{\beta_{dx}}^\top)^\top$.
Based on this model approximation, we estimate $\theta = (\theta_\beta^\top, \theta_{g\lambda}^\top)^\top$ by the penalized least squares method.

Let $\bm{D}$ be a $(dx K_{TS} + K_{TSUV})$-dimensional positive semidefinite symmetric matrix such that $\rho_{\max}(\bm{D}) = O(1)$.
This matrix serves as a penalization matrix, in which typical elements of $\bm{D}$ are, for example, the integrated derivatives of the basis functions. 
Write $\hat \Pi_i \coloneqq (P_X(X_i, T_i, S_i)^\top, \hat P(T_i, S_i, \hat U_i, \hat V_i)^\top)^\top$,  $\hat{\bm{\Pi}}_n = (\hat \Pi_1, \ldots, \hat \Pi_n)^\top$, and $\bm{Y}_n = (Y_1, \ldots, Y_n)^\top$.
Then, for a sequence of tuning parameters $\{\tau_n\}$ tending to zero at a certain rate as $n \to \infty$, the penalized least squares estimator of $\theta$ is given by
\begin{align}\label{eq:series_pen}
    \hat \theta_n = (\hat \theta_{n,\beta}^\top, \hat \theta_{n,g\lambda}^\top)^\top
    & \coloneqq \left[\hat{\bm{\Pi}}_n^\top \hat{\bm{\Pi}}_n + \tau_n \bm{D} n \right]^{-1}\hat{\bm{\Pi}}_n^\top \bm{Y}_n.
\end{align}
The first-stage estimator of $\beta_l(t, s)$ can be obtained by $\hat \beta_{n,l}(t, s) \coloneqq P_{TS}(t, s)^\top \hat \theta_{n, \beta_l}$ for $l = 1 , \ldots, dx$.
Similarly, $g(t,s) \lambda(u, v)$ can be estimated as $\hat P(t, s, u, v)^\top \hat \theta_{n, g\lambda}$.

\paragraph{Second step: marginal integrations}

In the second step, we first recover the $\lambda$ function.
Recalling the scale normalization in \eqref{eq:scale_normal}, the estimator of $\lambda(u, v)$ can be naturally defined by
\begin{align}\label{eq:lambda_est}
    \hat \lambda_n(u, v) 
    \coloneqq \int_{\mathcal{TS}} \left( \hat P(t, s, u, v)^\top \hat \theta_{n, g\lambda} \right)\mathrm{d}t\mathrm{d}s
    = [P^*_{TS} \otimes \hat P_{UV}(u, v)]^\top \hat \theta_{n, g\lambda},
\end{align}
where $P^*_{TS} \coloneqq \int_{\mathcal{TS}}P_{TS}(t,s)\mathrm{d}t\mathrm{d}s$.

For the estimation of $g(t,s)$, note that simply dividing $\hat P(t, s, u, v)^\top \hat \theta_{n, g\lambda}$ by $\hat \lambda_{n}(u, v)$ results in an inefficient estimator and is not well-defined if $\hat \lambda_{n}(u, v)$ is close to zero.
To obtain a more efficient estimator utilizing the information at all $(u,v)$ values, we apply a least-squares principle to $\hat P(t, s, u, v)^\top \hat \theta_{n, g\lambda} \approx \hat \lambda_n(u, v)  g(t,s)$; that is, we define our estimator $\hat g_n(t,s)$ as the solution of $\min_{g^*} \int_0^1 \int_0^1 \left(\hat P(t,s,u,v)^\top \hat \theta_{n,g\lambda} - \hat \lambda_n(u,v) g^* \right)^2 \mathrm{d}u\mathrm{d}v$.
By simple calculation, we can find that the solution has a closed form expression:
\begin{align}\label{eq:g_est}
    \hat g_n(t,s) 
    \coloneqq \int_0^1\int_0^1 \left( \hat P(t, s, u, v)^\top \hat \theta_{n, g\lambda} \right)\hat \omega_n(u, v)\mathrm{d}u\mathrm{d}v
    = [P_{TS}(t,s) \otimes \hat P^*_{UV}]^\top \hat \theta_{n, g\lambda},
\end{align}
where $\hat P^*_{UV} \coloneqq \int_0^1\int_0^1 \hat P_{UV}(u, v) \hat \omega_n(u, v) \mathrm{d}u \mathrm{d}v$ and $\hat \omega_n(u, v) \coloneqq \hat\lambda_n(u, v)/\int_0^1\int_0^1 \hat \lambda_n(u', v')^2 \mathrm{d}u'\mathrm{d}v'$.

\paragraph{Third step: kernel smoothing}

In the final step, we re-estimate $\beta$ using a local linear kernel regression.
Let $W$ be a symmetric univariate kernel density function and $\{(h_T, h_S)\}$ be the sequence of bandwidths tending to zero as $n$ increases.
The bandwidths may depend on the evaluation point $(t,s,x)$ of the CATR parameter in general (see Remark \ref{rem:opt_band}).
Further, we define $X_i(t,s) \coloneqq (X_i^\top, X_i^\top(T_i - t)/h_T, X_i^\top(S_i - s)/h_S)^\top$, and $W_{i}(t, s) \coloneqq (h_T h_S)^{-1}W((T_i - t) / h_T)W( (S_i - s) /h_S)$ for a given $(t,s) \in \mathcal{TS}$.
Then, our final estimator $\tilde \beta_n(t, s)$ of $\beta(t,s)$ is obtained by
\begin{align}\label{eq:ll_reg}
    \tilde \beta_n(t, s) 
    & \coloneqq \mathbb{S}_{dx} \left[ \sum_{i=1}^n X_i(t,s) X_i(t,s)^\top W_{i}(t,s) \right]^{-1} \sum_{i=1}^n X_i(t,s) [Y_i - \hat g_n(T_i, S_i) \hat \lambda_n(\hat U_i, \hat V_i)] W_{i}(t,s),
\end{align}
where $\mathbb{S}_{dx} \coloneqq [I_{dx}, \: \bm{0}_{dx \times 2dx}]$, and the estimator of $\mathrm{CATR}(t,s,x)$ is given by $\hat{\mathrm{CATR}}_n(t,s,x) \coloneqq x^\top \tilde \beta_n(t, s)$.

\subsection{Asymptotic Properties}

To investigate the asymptotic properties of the proposed estimators, we first introduce assumptions on the dependence structure underlying the data.
Typical applications of treatment spillover models would be data from school classes, households, working places, neighborhoods, municipalities, etc., which are usually subject to some local dependence.
To account for such dependency, we assume that all $N$ agents are located in a (latent) $d$-dimensional space $\mathcal{D} \subseteq \mathbb{R}^d$ for some $d < \infty$.
For spatial data applications, the space $\mathcal{D}$ is typically defined by a geographical space with $d = 2$.
For non-spatial data, it is possible that $\mathcal{D}$ is a space of general social and demographic characteristics (or a mixture of geographic space and such spaces), and in this case we should view it as an embedding of individuals rather than their actual locations.
The set of observation locations, which we denote by $\mathcal{D}_N \subset \mathcal{D}$, may differ across different $N$.
With a little abuse of notation, for each $i$, we use the same $i$ to denote his/her location.
Let $\Delta(i,j)$ denote the Euclidean distance between $i$ and $j$.

\begin{assumption}\label{as:D_space}
    For all $i,j \in \mathcal{D}_N$ such that $i \neq j$, \textbf{(i)} $\Delta(i,j) \ge 1$ (without loss of generality); and
	\textbf{(ii)} there exists a threshold distance $\ol{\Delta} \in \mathbb{Z}$ satisfying $A_{i,j} = 0$ if $\Delta(i,j) > \ol{\Delta}$.
\end{assumption}

Assumption \ref{as:D_space}(i) rules out the infill asymptotics (\citealp{cressie1993statistics}).
Assumption \ref{as:D_space}(ii) is a ``homophily'' assumption in the space $\mathcal{D}$. 
This may not be too restrictive since the definition of $\mathcal{D}$ can be freely modified depending on the context.\footnote{
    However, note that assuming that the space $\mathcal{D}$ is a subset of Euclidean space is also a restriction. 
    For applications with a tree-like network structure (such as supply chain networks), $\mathcal{D}$ would more naturally fit in a hyperbolic space.
    For a comprehensive review of the implications of the choice of geometry in network data analysis, see \cite{smith2019geometry}.
    }
This framework would accommodate many empirically relevant situations.
For example, clustered sample data with finite cluster size can be seen as a special case of it.
An important implication from these two assumptions is that the size of reference group is uniformly bounded above by some constant of order $\ol{\Delta}^d$.
Thus, we do not allow for the existence of ``dominant units,'' which may have increasingly many interacting partners as the network grows.
Assuming $\ol{\Delta}$ to be an integer is only to simplify the proof.

\begin{assumption}\label{as:mixing}
	\textbf{(i)} $\{(X_i, Z_i, \epsilon_i, U_i): i \in \mathcal{D}_N, N \ge 1\}$ is an $\alpha$-mixing random field with mixing coefficient $\alpha(k,l,r) \le (k + l)^\vartheta \hat \alpha(r)$ for some constant $0 \le \vartheta < \infty$ and function $\hat \alpha(r)$ that satisfies $\sum_{r = 1}^\infty (r + 2\ol{\Delta})^{d-1} \hat\alpha(r) < \infty$; and
	\textbf{(ii)} $\sup_{i \in \mathcal{D}_N}||X_i|| < \infty$.
\end{assumption}

\begin{assumption}\label{as:error}
  $\sup_{i \in \mathcal{D}_N}\bbE[ \epsilon_i^4  \mid \mathcal{F}_N]  < \infty$, $\bbE[ \epsilon_i \mid \mathcal{F}_N] = \bbE[ \epsilon_i \mid U_i, V_i]$ for all $i \in \mathcal{D}_N$, and $\epsilon_i \indep \epsilon_j  \mid \mathcal{F}_N$ for all $i,j \in \mathcal{D}_N$ such that $i \neq j$, where $\mathcal{F}_N \coloneqq \{(X_i, Z_i, U_i): i \in \mathcal{D}_N\}$.
\end{assumption}

For the precise definition of the $\alpha$-mixing random field and the mixing coefficient in this context, see Definition \ref{def:mixing} in Appendix \ref{app:note_def} (see also \citealp{jenish2009central}).
Combined with Assumption \ref{as:D_space}(ii), Assumption \ref{as:mixing}(i) implies that $\{S_i\}$, $\{V_i\}$, and $\{\tilde \epsilon_i\}$ are also $\alpha$-mixing processes.
For Assumption \ref{as:error}, the second assumption restricts the dependence structure in the treatment heterogeneity, and the third requires that all potential correlations among $\epsilon_i$'s are through the correlations of the other individual characteristics.
Note that $\tilde \epsilon_i$ can be written as $\tilde \epsilon_i = g(T_i, S_i)[\epsilon_i - \lambda(U_i, V_i)]$.
Thus, with this assumption we impose that the $\tilde \epsilon_i$'s also have the fourth order moments and that they are independent of each other conditional on the individual characteristics.

\begin{assumption}\label{as:basis}
	\textbf{(i)} For all $j$, $(p_{T,j}(t), p_{S,j}(s), p_{U,j}(u), p_{V,j}(v))$ are uniformly bounded on $\mathcal{T}$, $\mathcal{S}$, $[0,1]$, and $[0,1]$, respectively;
    \textbf{(ii)} $P_U(u)$ and $P_V(v)$ are differentiable such that $\sup_{u \in [0,1]}||\partial P_U(u)/\partial u|| \le \zeta_U$ and $\sup_{v \in [0,1]}||\partial P_V(v)/\partial v|| \le \zeta_V$;
    \textbf{(iii)} $||P^*_{TS}|| = O(1)$ and $||\ol{P}^*_{UV}|| = O(1)$, where 
    \begin{align*}
        \ol{P}^*_{UV} \coloneqq \int_0^1 \int_0^1 \ol{P}_{UV}(u, v)\omega(u, v)\mathrm{d}u\mathrm{d}v
        \quad \text{with} \;\;
        \omega(u, v) \coloneqq \frac{\lambda(u, v)}{\int_0^1\int_0^1 \lambda(u', v')^2 \mathrm{d}u'\mathrm{d}v'};
    \end{align*} 
    \textbf{(iv)} there exist positive constants $(c, \xi_1)$ such that $||P_{UV}(u, v) - P_{UV}(u', v')|| \le c K_{UV}^{\xi_1} || (u, v) - (u', v') ||$ for any $(u, v), (u',v') \in [0,1]^2$; and
    \textbf{(v)} there exist positive constants $(c, \xi_2)$ such that $||P_{TS}(t,s) - P_{TS}(t', s')|| \le c K_{TS}^{\xi_2} || (t,s) - (t',s') ||$ for any $(t,s), (t', s') \in \mathcal{TS}$.
\end{assumption}

\begin{assumption}\label{as:eigen}
    Define $\bm{\Pi}_n$ analogously to $\hat{\bm{\Pi}}_n$ by replacing the estimates of $(U_i, V_i, \bbE[ P_{UV}(U, V)])$ with their true values.
    There exist $(c_1, c_2)$ such that $c_1 < \rho_{\min}(\bbE[\bm{\Pi}_n^\top \bm{\Pi}_n/n]) \le \rho_{\max}(\bbE[\bm{\Pi}_n^\top \bm{\Pi}_n/n]) < c_2$ uniformly in $(K_{TS}, K_{UV})$ for sufficiently large $n$.
\end{assumption}

Assumption \ref{as:basis}(i) is introduced for analytical simplicity, which imposes some restrictions on the choice of basis functions.
For example, B-spline basis and Fourier series satisfy this assumption.
The uniform boundedness implies that $\sup_{t \in \mathcal{T}}||P_T(t)|| = O(\sqrt{K_T})$, and the similar result applies to $P_S$, $P_U$, and $P_V$.
Assumptions \ref{as:basis}(iii)--(v) are technical conditions to derive the uniform convergence rate for the second-stage estimators.
Assumption \ref{as:eigen} is a standard non-singularity condition, which essentially serves as the identification condition for the CATR parameter.
Note that this can hold even when the support of the IVs is small.

\begin{assumption}\label{as:series_approximation}
    For all $l = 1, \ldots, dx$, $\beta_l$ is twice continuously differentiable, and there exists a vector $\theta_{\beta_l}^*$ and a positive constant $\mu_\beta$ such that $\sup_{(t, s) \in \mathcal{TS} } |\beta_l(t, s) - P_{TS}(t, s)^\top \theta_{\beta_l}^*| = O(K_{TS}^{-\mu_\beta})$.
    Similarly, $(g, \lambda)$ are continuously differentiable, and there exist vectors $(\theta_g^*, \theta_\lambda^*)$ and positive constants $(\mu_g, \mu_\lambda)$ such that $\sup_{(t, s) \in \mathcal{TS} } | g(t, s)  - P_{TS}(t, s)^\top \theta_g^*| = O(K_{TS}^{-\mu_g})$ and $\sup_{(u,v) \in [0,1]^2} | \lambda(u, v)  - \ol{P}_{UV}(u, v)^\top \theta_\lambda^*| = O(K_{UV}^{-\mu_\lambda})$.
\end{assumption}
The constants $(\mu_\beta, \mu_g, \mu_\lambda)$ generally depends on the choice of the basis function and the dimension and the smoothness of the function to be approximated.
For example, if the target function belongs to a $k$-dimensional H\"older class with smoothness $p$, it typically holds that $\mu = p/k$ for splines, wavelets, etc (\citealp{chen2007large}).
Here, define $\theta^*_{g\lambda} \coloneqq \theta_g^* \otimes \theta_\lambda^*$ so that $\ol{P}(t,s,u,v)^\top \theta^*_{g\lambda} =  (P_{TS}(t, s)^\top \theta_g^*) \cdot (\ol{P}_{UV}(u,v)^\top \theta_\lambda^*)$ holds, where $\ol{P}(t,s,u,v) \coloneqq  P_{TS}(t, s) \otimes \ol{P}_{UV}(u, v)$.
Then, we can easily see that
\begin{align*}
    & \sup_{(t, s, u, v)\in \mathcal{TS} [0,1]^2} | g(t, s)\lambda(u,v) - \ol{P}(t, s, u, v)^\top \theta_{g\lambda}^*| \\
    & = \sup_{(t, s, u, v)\in \mathcal{TS} [0,1]^2} | g(t, s)\lambda(u, v)  - (P_{TS}(t, s)^\top \theta_g^*) \cdot (\ol{P}_{UV}(u,v)^\top \theta_\lambda^*)| = O(K_{TS}^{-\mu_g} + K_{UV}^{-\mu_\lambda}).
\end{align*}

\begin{assumption}\label{as:1st_stage}
    There exists a constant $\nu \in (0,1/2]$ such that $(\sup_{i \in \mathcal{D}_N}|\hat U_i - U_i|,  \sup_{i \in \mathcal{D}_N}|\hat V_i - V_i|) = O_P(n^{-\nu})$.
\end{assumption}

\begin{assumption}\label{as:conv_rate}
    As $n \to \infty$, \textbf{(i)} $\kappa_n^4/n \to 0$ and $ \zeta_{\dagger} \sqrt{\kappa_n} n^{-\nu} \to 0$, where $\zeta_{\dagger} \coloneqq \zeta_U \sqrt{K_V} + \zeta_V \sqrt{K_U}$; and
    \textbf{(ii)} $(\kappa_n^4 \ln \kappa_n) /n \to 0$.
\end{assumption}

For Assumption \ref{as:1st_stage}, if we adopt a parametric model specification for the treatment equation in \eqref{eq:treament_eq}, then the assumption holds with $\nu = 1/2$.
As discussed in Remark \ref{rem:kappa}, we require $\nu$ to be at least larger than $1/3$ under optimal bandwidths.
Assumption \ref{as:conv_rate}(i) is used to establish a matrix law of large numbers.
Recalling that $(K_{TS}, K_{UV}) \asymp \kappa_n$, the first part of the assumption requires that $K_{TS}$ and $K_{UV}$ must grow slower than $n^{1/4}$.
It is easy to see that $\zeta_{\dagger}$ gives the order of $||\partial P_{UV}(u,v)/\partial u||$ and $||\partial P_{UV}(u,v)/\partial v||$.
For example, when one uses a tensor product B-splines, it can be shown that $\zeta_{\dagger} = O(\kappa_n)$ (see, e.g., \citealp{hoshino2021treatment}).
Condition (ii) implies the first part of (i).
This assumption can be relaxed if we can strengthen Assumption \ref{as:error} so that the error terms $\{\tilde \epsilon_i\}$ have the moments of order higher than four. 

Now, in the following theorem, we derive the convergence rate for the first-stage series estimator.
\begin{theorem}\label{thm:1st_stage}
    Suppose that Assumptions \ref{as:D_space}--\ref{as:error}, \ref{as:basis}(i),(ii), \ref{as:eigen}--\ref{as:1st_stage}, and \ref{as:conv_rate}(i) hold.
    Then, we have 
    \begin{description}
    \item[(i)] $|| \hat \theta_{n, \beta} - \theta_\beta^* || = O_P\left(\sqrt{\frac{\kappa_n}{n}} + b_\mu + \tau_n^* + n^{-\nu} \right)$
    \item[(ii)] $|| \hat \theta_{n, g\lambda} - \theta_{g\lambda}^* || = O_P\left(\frac{\kappa_n}{\sqrt{n}} + b_\mu + \tau_n^* + n^{-\nu} \right)$
    \end{description}
    where $b_\mu \coloneqq K_{TS}^{-\mu_\beta} + K_{TS}^{-\mu_g} + K_{UV}^{-\mu_\lambda}$, and $\tau^*_n \coloneqq \tau_n \sqrt{\theta^{*\top} \bm{D} \theta^*}$ with $\theta^* = (\theta_\beta^{*\top}, \theta^{*\top}_{g\lambda})^\top$.
\end{theorem}
The above results should be standard in the literature.
In particular, result (i) indicates that we can estimate $\mathrm{CATR}(t,s,x)$ consistently by $x^\top\hat \theta_{n, \beta}$, although less efficiently when compared with our final estimator.
The magnitude of $\tau^*_n$ depends not only on the penalty matrix $\bm{D}$ but also on the choice of the basis function.
Under $\rho_{\max}(\bm{D}) = O(1)$, it is clear that the upper bound for $\tau^*_n$ is $O(\tau_n\kappa_n)$.
If the coefficients are decaying in the order of series, $\tau^*_n = O(\tau_n)$ would hold, as in \cite{han2020nonparametric}.

For later use, we introduce the following miscellaneous assumptions.

\begin{assumption}\label{as:misc}
    \textbf{(i)} $\inf_{(u,v) \in [0,1]^2} f_{UV}(u,v) > 0$, where $f_{UV}$ is the joint density of $(U,V)$;
    \textbf{(ii)} there exists $c$ such that $\rho_{\max}(\bbE[\ol{P}_{UV}(U,V)\ol{P}_{UV}(U,V)^\top]) < c$ uniformly in $(K_U, K_V)$; and
    \textbf{(iii)} $\zeta_{\dagger}\kappa_n n^{-\nu} = O(1)$ and $\kappa_n^{3/2}(b_\mu + \tau^*_n + n^{-\nu}) = O(1)$.
\end{assumption}

In the next theorem, we establish the convergence rates for the second-stage estimators.
\begin{theorem}\label{thm:2nd_stage}
    Suppose that Assumptions \ref{as:D_space}--\ref{as:conv_rate} hold.
    Then, we have 
    \begin{description}
    \item[(i)] $\sup_{(u,v) \in [0,1]^2}|\hat \lambda_n(u,v) - \lambda(u, v)| = O_P\left(\sqrt{\frac{\kappa_n \ln \kappa_n}{n}} + \sqrt{\kappa_n}(b_\mu + \tau^*_n + n^{-\nu})\right)$.
    \item[(ii)] If Assumption \ref{as:misc} additionally holds, we have
    
    $\sup_{(t,s) \in \mathcal{TS}}|\hat g_n(t,s) - g(t,s)| = O_P\left(\sqrt{\frac{\kappa_n \ln \kappa_n}{n}} + \sqrt{\kappa_n}(b_\mu + \tau^*_n + n^{-\nu})\right)$.
    \end{description}
\end{theorem}
The above theorem clearly shows that the marginal integration procedure successfully resolves the slower convergence of the first-stage estimator.
Note however that both estimators $\hat \lambda_n$ and $\hat g_n$ do not attain the optimal uniform convergence rate of \cite{stone1982optimal}.\footnote{
    Whether our marginal integration-based estimation method can achieve the optimal rate is left as an open question.
    For example, for $\hat \lambda_n$, in order to achieve the optimal uniform convergence rate, the bias term should be of order just $O(K_{UV}^{-p/2})$, where $p$ is a smoothness parameter (see the discussion given after Assumption \ref{as:series_approximation}).
    For the attainability of the optimal uniform rate for standard linear series regression estimators, see \cite{huang2003local}, \cite{belloni2015some}, and \cite{chen2015optimal}, among others.
}

We next derive the limiting distribution of our final estimator for $\mathrm{CATR}(t,s,x)$, where $(t,s)$ is a given interior point of $\mathcal{TS}$.
Let $f_i(t,s)$ be the joint density of $(T_i, S_i)$, and define
\begin{align*}
    & \Omega_{1,i}(t,s) \coloneqq \bbE[X_i X_i^\top \mid T_i = t , S_i = s],
    \quad \ol{\Omega}_{1,n}(t,s) \coloneqq \frac{1}{n}\sum_{i = 1}^n\Omega_{1,i}(t,s)f_i(t,s) \\
    & \Omega_{2,i}(t,s) \coloneqq \bbE[X_i X_i^\top \tilde \epsilon_i^2 \mid T_i = t , S_i = s],
    \quad \ol{\Omega}_{2,n}(t,s) \coloneqq \frac{1}{n}\sum_{i = 1}^n\Omega_{2,i}(t,s)f_i(t,s). 
\end{align*}
\begin{assumption}\label{as:omega_f}
    \textbf{(i)} For all $i \in \mathcal{D}_N$, $f_i(t, s)$ and $\Omega_{1,i}(t, s)$ are continuously differentiable and $\Omega_{2,i}(t,s)$ is continuous on $\mathcal{TS}$; and
    \textbf{(ii)} $\ol{\Omega}_{1}(t,s) \coloneqq \lim_{n \to \infty} \ol{\Omega}_{1,n}(t,s)$ and $\ol{\Omega}_{2}(t,s) \coloneqq \lim_{n \to \infty} \ol{\Omega}_{2,n}(t,s)$ exist and are positive definite.
\end{assumption}
\begin{assumption}\label{as:kernel}
    The kernel $W$ is a probability density function that is symmetric and continuous on the support $[-c_W, c_W]$.
\end{assumption}

\begin{assumption}\label{as:band}
    \textbf{(i)} $(h_T, h_S) \asymp n^{-1/6}$; and
    \textbf{(ii)} $\zeta_\dagger n^{(1/6)-\nu} \to 0$ and $n^{1/3}(\tau^*_n + b_\mu + n^{-\nu}) \to 0$ as $n \to \infty$.
\end{assumption}

Assumption \ref{as:kernel} is fairly standard.
The compact support assumption is used just for simplicity, and it can be dropped at the cost of lengthier proof.
For Assumption \ref{as:band}, it will be later shown that condition (i) is the optimal rate for the bandwidths.
We use condition (ii) to ensure that the final CATR estimator becomes oracle efficient.

Let $\mathbf{1}_i(t,s) \coloneqq \mathbf{1}\left\{\frac{|T_i - t|}{h_T} \le c_W , \frac{|S_i - s|}{h_S} \le c_W \right\}$, $\bm{I}_n(t,s) \coloneqq \text{diag}\left(\frac{\mathbf{1}_1(t,s)}{h_T h_S}, \ldots, \frac{\mathbf{1}_n(t,s)}{h_T h_S}\right)$, $\bm{P}_{n,TS} \coloneqq (P_{TS}(T_1, S_1), \ldots, P_{TS}(T_n, S_n))^\top$, and $\bm{P}_{n,UV} \coloneqq (\ol{P}_{UV}( U_1, V_1), \ldots, \ol{P}_{UV}(U_n, V_n))^\top$.

\begin{assumption}\label{as:eigen2}
    There exist constants $(c_{TS}, c_{UV})$ such that $\rho_\text{max}(\bbE[ \bm{P}_{n,TS}^\top \bm{I}_n(t,s) \bm{P}_{n,TS}/n]) < c_{TS}$ and $\rho_\text{max}(\bbE[ \bm{P}_{n,UV}^\top \bm{I}_n(t,s) \bm{P}_{n,UV}/n]) < c_{UV}$ uniformly in $(K_{TS}, K_{UV}, h_T, h_S)$ for sufficiently large $n$.
\end{assumption}

\begin{assumption}\label{as:ij}
    \textbf{(i)} $\sup_{(i,j) \in \mathcal{D}_N}\sup_{(t_i,s_i,t_j,s_j)\in(\mathcal{TS})^2}f_{i,j}(t_i,s_i,t_j,s_j) < \infty$, where $f_{i,j}(t_i,s_i,t_j,s_j)$ is the joint density of $(T_i, S_i, T_j, S_j)$; and
    \textbf{(ii)} $\sup_{(i,j) \in \mathcal{D}_N}\bbE[|\tilde \epsilon_i \tilde \epsilon_j| \mid T_i, S_i, T_j, S_j] < \infty$.
\end{assumption}

\begin{assumption}\label{as:alpha}
    \textbf{(i)} $\sum_{r = 1}^\infty r^{e + d - 1} \hat \alpha(r)^{1/2} < \infty$ for some $e > d/2$;
    \textbf{(ii)} $\alpha(C \ol{\Delta}^d, \infty, r - 2\ol{\Delta}) = O(r^{-d'})$ for some $d' > d$ and a positive constant $C$ (see \eqref{eq:mixing2}); and
    \textbf{(iii)} $\sum_{r = 1}^\infty (r + 2\ol{\Delta})^{[d \ell/(\ell - 2)] - 1} \hat\alpha(r) < \infty$ for some $3 \le \ell < 4$.
\end{assumption}

Assumptions \ref{as:eigen2}, \ref{as:ij}, and \ref{as:alpha} are technical requirements to derive the asymptotic normality of our CATR estimator.
In particular, Assumption \ref{as:alpha}(ii) and (iii) are introduced to utilize the central limit theorem for mixing random fields developed by \cite{jenish2009central} in our context.
Now we are ready to state our main theorem:

\begin{theorem}\label{thm:catr}
    Suppose that Assumptions \ref{as:D_space}--\ref{as:alpha} hold.
    Then, for a given interior point $(t,s) \in \mathcal{TS}$ and a finite $x$ in the support of $X$, we have
    \begin{align*}
        & \sqrt{n h_T h_S}\left(\hat{\mathrm{CATR}}_n(t,s,x) - \mathrm{CATR}(t,s,x) - \frac{\varphi_2^1}{2} [ x^\top \ddot \beta_{TT}(t,s) h_T^2 + x^\top \ddot \beta_{SS}(t,s) h_S^2]\right) \\
        & \quad \overset{d}{\to} N\left(\mathbf{0}_{dx \times 1}, (\varphi_0^2)^2 x^\top[\ol{\Omega}_1(t,s)]^{-1} \ol{\Omega}_2(t,s) [\ol{\Omega}_1(t,s)]^{-1}x \right). 
    \end{align*}
    where $\varphi_j^k \coloneqq \int \phi^j W(\phi)^k \mathrm{d} \phi$, $\ddot \beta_{TT}(t,s) \coloneqq \partial^2 \beta (t,s)/ (\partial t)^2$, and $\ddot \beta_{SS}(t,s) \coloneqq \partial^2 \beta (t,s)/ (\partial s)^2$.
\end{theorem}

The proof of Theorem \ref{thm:catr} is straightforward from Lemma \ref{lem:beta_normal}, and thus is omitted.
In Lemma \ref{lem:gl_diff}, we show that the estimation error for our CATR estimator caused by the estimations of $g$ and $\lambda$ are of order $o_P((n h_T h_S)^{-1/2})$.
Thus, the asymptotic distribution presented in the theorem is in fact equivalent to that obtained when the estimators $\hat g_n$ and $\hat \lambda_n$ are replaced by their true counterparts; that is, our CATR estimator has an oracle property.

\begin{remark}[Covariance matrix estimation]
    For statistical inference, we need to consistently estimate the asymptotic covariance matrix.
    The matrix $\ol{\Omega}_1(t,s)$ can be easily estimated by the kernel method -- see Lemma \ref{lem:kernel_LLN}.
    Similarly, we can estimate $(\varphi_0^2)^2 \ol{\Omega}_2(t,s)$ using the kernel method with the error terms $\{\tilde \epsilon_i\}$ being replaced by the residuals.
    That is, letting $\hat \epsilon_i(t,s) \coloneqq Y_i - X_i^\top \tilde \beta_n(t, s) - \hat g_n (T_i,S_i) \hat \lambda_n(\hat U_i, \hat V_i)$ for $i = 1, \ldots, n$, we can show that $\hat \Omega_{2,n}(t,s) \coloneqq (h_T h_S /n)\sum_{i=1}^n X_i X_i^\top \hat \epsilon_i(t,s)^2 W_{i}(t,s)^2$ is consistent for $(\varphi_0^2)^2\ol{\Omega}_2(t,s)$ with an additional mild assumption -- see Lemma \ref{lem:cov_estimation}.
\end{remark}

\begin{remark}[Bandwidth selection]\label{rem:opt_band}
    As a result of Theorem \ref{thm:catr}, the asymptotic mean squared error (AMSE) of $\hat{\mathrm{CATR}}_n(t,s,x)$ is given by
    \begin{align*}
        \mathrm{AMSE}(t,s,x)
        & = \frac{(\varphi_2^1)^2}{4}\left[ x^\top \ddot \beta_{TT}(t,s) h_T^2 + x^\top \ddot \beta_{SS}(t,s) h_S^2\right]^2 + \frac{(\varphi_0^2)^2 x^\top[\ol{\Omega}_1(t,s)]^{-1} \ol{\Omega}_2(t,s) [\ol{\Omega}_1(t,s)]^{-1}x}{n h_T h_S}.
    \end{align*}
    From this, we can derive the optimal bandwidth parameters that minimize the AMSE.
    In particular, suppose that the bandwidths are given by $h_T = C(t,s,x) \sigma_n(T) n^{-1/6}$ and $h_S = C(t,s,x) \sigma_n(S) n^{-1/6}$ for some constant $C(t,s,x) > 0$, where $\sigma_n(T)$ and $\sigma_n(S)$ are the sample standard deviations of $T$ and $S$, respectively.
    Then, we can obtain $C(t,s,x)$ that minimizes the AMSE as follows:
    \begin{align}\label{eq:optimal_band}
        C(t,s,x) = \left( \frac{2 (\varphi_0^2)^2 x^\top[\ol{\Omega}_1(t,s)]^{-1} \ol{\Omega}_2(t,s) [\ol{\Omega}_1(t,s)]^{-1}x }{(\varphi_2^1)^2 \sigma_n(T) \sigma_n(S) [x^\top \ddot \beta_{TT}(t,s) \sigma^2_n(T) + x^\top \ddot \beta_{SS}(t,s) \sigma^2_n(S)]^2} \right)^{1/6}
    \end{align}
    Although the optimal bandwidths involve several unknown quantities, their approximate values are obtainable using the first-stage series estimates (with or without regularization).
\end{remark}

\begin{remark}\label{rem:kappa}
    The functional form of $\pi(X_i, Z_i)$ determines the possible range for $\nu$.
    In view of Assumption \ref{as:band}(ii), $\nu$ must satisfy $1/3 < \nu$.
    In addition, assuming $\zeta_{\dagger} = O(\kappa_n)$, the second part of Assumption \ref{as:conv_rate}(i) is reduced to $\kappa_n^{3/2} n^{-\nu} \to 0$.
    Thus, for example when $\kappa_n = O(n^{1/5})$ (in view of the first part of Assumption \ref{as:conv_rate}(i)), the condition is further reduced to $3/10 < \nu$, being consistent with Assumption \ref{as:band}(ii).
    These imply that the treatment model does not have to be fully parametrically specified in general (i.e., $\nu = 1/2$), but at the same time a full nonparametric specification may not be acceptable.
\end{remark}

\section{Causal Impacts of Unemployment on Crime}\label{sec:empiric}

It is often considered that unemployment and crime are endogenously related because of their simultaneity (e.g., \citealp{levitt2001alternative}).
In economic theory, criminal activity is typically characterized by the balance between the cost and benefit of committing illegal activities (\citealp{becker1968crime}).
Thus, poor local labor market conditions result in a relative increase in the benefits of crime, increasing the number of crime incidents.
On the other hand, crime drives away business owners and customers, exacerbating the working condition.
In addition, from the viewpoint of criminals, if their communities are economically deteriorating with high unemployment, they might commit crimes in more ``beneficial'' neighborhoods.
This would suggest the need to account for the spatial spillover effects of unemployment.

There are several prior studies that investigate the relationship between unemployment and crime using some IV-based methods (e.g., \citealp{raphael2001identifying, lin2008does, altindag2012crime}).
For example, using U.S. state-level data, \cite{raphael2001identifying} conducted two-stage least squares regression analysis with military spending and oil costs as the IVs for the unemployment.
Their estimation results suggest that unemployment is indeed an important determinant of property crime rates.
In this paper, we present an empirical analysis on the causal effect of local unemployment rate on crime based on Japanese city-level data.
Our empirical study aims at extending the earlier works in two ways.
Our model is based on a more flexible potential outcome framework and allows for the existence of spillover effects from neighboring regions.

The variables used and their definitions are summarized in Table \ref{table:vars}.
The outcome variable of interest is the city-level crime rate, which is based on the number of crimes recorded in each city in 2006.
The crime data cover all kinds of criminal offenses (the breakdown is unfortunately unavailable).
The treatment variable is the regional unemployment rate as of 2005.
As an IV for the unemployment rate, we employ the availability of child day-care facilities in 2006.
It would be legitimate to assume that the availability of childcare facilities does not directly affect the crime rate.
Meanwhile, in the labor economics literature, there is empirical evidence that expanding childcare services is effective in increasing female labor participation (e.g., \citealp{bauernschuster2015public}). 
Thus, the availability of childcare facilities can be viewed as an indicator of the local working environment, particularly for young females, and would contribute to reducing the total unemployment rate.
For other control variables, we include population density, annual retail sales per capita, the ratio of elderly people, and the ratio of single households.
The sales data are as of 2006, and the others are those in 2005.\footnote{
    All this information is freely available from \textbf{e-Stat} (a portal site for Japanese Government Statistics). 
    \url{https://www.e-stat.go.jp/en/regional-statistics/ssdsview/municipality}
    }

\begin{table}[!h]
    \begin{center}
    \caption{Variables used}
    \begin{tabular}{lll}
    \hline\hline
     & \multicolumn{1}{c}{Variables$^*$} & (Shorthand) \\
     \hline
    $Y$ & Crime rate: $100 \times \text{\# crimes / population}$ & \textit{crime} \\
    $T$ & Unemployment rate: $100 \times \text{unemployment / labor force population}$ & \textit{unempl} \\
    $X$ & Population density: $\ln[\text{daytime population / area}]$ & \textit{density} \\
     & Per-capita retail sales: $\ln[\text{retail sales / population}]$ & \textit{sales} \\
     & Elderly ratio: $\text{population over 65 / whole population}$ & \textit{elderly} \\
     & Ratio of single households: $\text{\# single households / \# total households}$ & \textit{single} \\
    $Z$ & Availability of childcare facilities: $1000 \times \text{\# childcare facilities / \# total households}$ & \textit{childcare} \\
     \hline
     & * The variables are all at city level. & \\
    \end{tabular}
    \label{table:vars}
    \end{center}
\end{table}

The network $\bm{A}_N$ is defined by whether the cities share a common boundary.
Specifically, assuming that there is a limit on the number of cities each city can interact with, we set $A_{i,j} = 1$ if city $j$ is adjacent to $i$ and in the $k$-nearest neighbors of $i$.
Below, we report the results when $k = 2$ for illustration.
We also have tried several different specifications for $\bm{A}_N$, and confirmed that the results are overall similar (for more details, see Appendix \ref{app:empir_robust}). 
For all cases, the treatment spillover variable $S_i$ is defined by $i$'s reference-group mean, say $\ol{\textit{unempl}}$.
Then, the model estimated is as follows:
\begin{align*}
    \textit{crime} 
    & = (1, \textit{density}, \textit{sales}, \textit{elderly}, \textit{single})^\top \beta(\textit{unempl}, \ol{\textit{unempl}}) + g(\textit{unempl}, \ol{\textit{unempl}})\epsilon \\
    \textit{unempl}
    & = (\textit{density}, \textit{sales}, \textit{elderly}, \textit{single}, \textit{childcare})^\top \gamma + \eta(U)
\end{align*}
Note that since the treatment equation does not involve any network interactions explicitly, the CQR estimation can be implemented using all data regardless of the specification of $\bm{A}_N$.
For the estimation of the CATR parameter, we need to select a subsample to maintain distributional homogeneity of $(U,V)$.
To this end we focus on the cities that have exactly two interacting partners (recalling that $k = 2$).
After excluding observations with missing data, the estimation of CATR was performed on a sample of size 1773.
The estimation procedure is the same as that in the Monte Carlo experiments in Appendix \ref{sec:MC}.
That is, we set the penalty parameter to $\tau_n = 5/n$ and use the B-spline basis with two internal knots.\footnote{
    We have confirmed that the results reported here have a certain robustness to other choices of penalty parameters and basis orders.
	However, we have also observed that when the number of internal knots is three or higher, it is better to employ some regularized estimator when computing the second derivatives appearing in \eqref{eq:optimal_band} to stabilize the estimates.
}
The descriptive statistics of the data are summarized in Table \ref{table:descstat} in Appendix \ref{app:empir_table}.

The estimation results for the treatment equation are provided in Table \ref{table:cqr} in Appendix \ref{app:empir_table}, where we can find that \textit{childcare} is significantly negatively related to the \textit{unempl} variable, as expected.
Figure \ref{fig:CATR} presents the estimated $\mathrm{CATR}(t,s,x)$ for different values of $(t,s,x)$.
In each panel, $t$ ranges over 0.1 to 0.9 empirical quantiles of $\{T_i\}$, and $s$ is either at 0.2 or 0.6 quantile of $\{S_i\}$.
Because of the correlation between $T$ and $S$, the CATR estimates at more extreme $t$ cannot be estimate reliably and thus they are not reported (see Figure \ref{fig:joint_density} in Appendix \ref{app:empir_table} for the joint density of $(T,S)$).
For the value of $x$, we evaluate at the empirical median in the left panel.
We consider two more cases for $x$: the median value for the cities in the bottom 20\% level of population density among all cities (middle panel), and that for the cities in the top 20\% level (right panel).
The former would be considered as a typical city in a rural area, while the latter as a typical city in an urban area.
In the figure, we also report the results obtained when $T$ and $S$ are treated as exogenous (this corresponds to $\hat{\mathrm{CATR}}_n^\text{ex}$ estimator given in Appendix \ref{sec:MC}).

Now, we report our main empirical findings.
First, we can observe that the CATR weakly increases in general as the unemployment rate increases.
That is, as the number of unemployed people increases, the crime rate tends to increase, which is consistent with the findings in the prior studies.
Second, in the left and right panels, the CATR with large $S$ tends to be significantly greater than the CATR with small $S$, indicating that the average unemployment rate of surrounding cities does affect the city's own crime rate in such cases.
However, interestingly, the spillovers are no longer prominent when the city is in a rural area.
One interpretation is that the cities in such areas are relatively independent from other cities and villages, and thus the spillover effects may be less impactful.
A more comprehensive figure that summarizes the CATR estimates when $S$ is at 0.2, 0.3, \ldots, 0.8 quantiles is presented in Figure \ref{fig:catr_all} in Appendix \ref{app:empir_robust}, where we can more clearly observe these tendencies. 
Lastly, we find that ignoring the potential endogeneity of $(\textit{unempl}, \ol{\textit{unempl}})$ leads to certain differences in the estimates, indicating the importance of accounting for the endogeneity.

\begin{figure}[!ht]
	\centering
	\includegraphics[bb = 0 0 1304 362, width = 17cm]{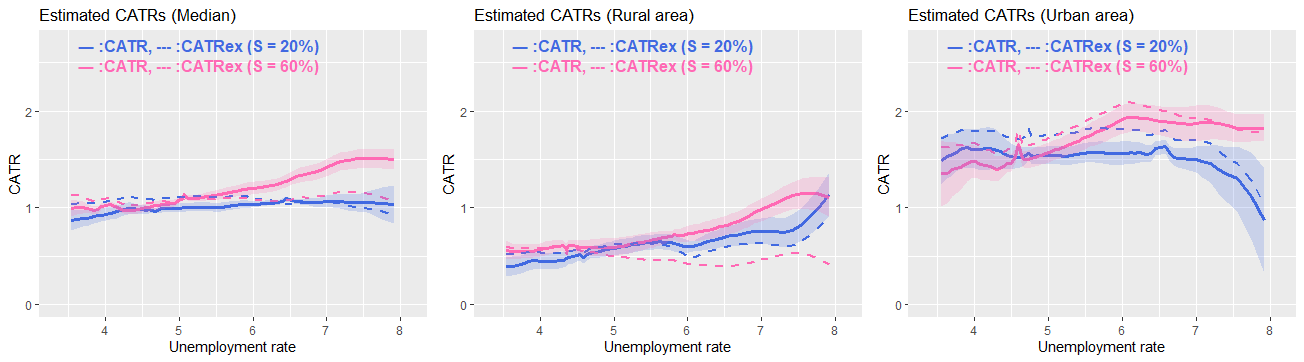}
	\caption{Estimated CATR.}
	\label{fig:CATR}
\end{figure}

\section{Conclusion}\label{sec:conclusion}

In this paper, we considered a continuous treatment effect model that admits potential treatment spillovers through social networks and the endogeneity for both one's own treatment and that of others.
We proved that the CATR, the conditional expectation of the potential outcome, can be nonparametrically identified under some functional form restrictions and the availability of appropriate IVs.
We also considered a more empirically tractable semiparametric treatment model and proposed a three-step procedure for estimating the CATR.
The consistency and asymptotic normality of the proposed estimator were established under certain regularity conditions.
As an empirical illustration, using Japanese city-level data, we investigated the causal effects of unemployment rate on the number of crime incidents.
As a result, we found that the unemployment rate indeed tends to increase the number of crimes and that, interestingly, the unemployment rates of neighboring cities matter more for non-rural cities than rural cities.
These results illustrate the usefulness of the proposed method.

\clearpage
\bibliography{references.bib}

\clearpage

\setcounter{page}{1}
\appendix

{\fontsize{10.5pt}{12pt}\selectfont

\begin{center}
	{\Large Supplementary Appendix for ``Estimating a Continuous Treatment Model with Spillovers: A Control Function Approach''} 
\end{center}

\begin{flushleft}
    {\large Tadao Hoshino}
    \bigskip
    
    School of Political Science and Economics, Waseda University.
    
    Email: \href{mailto:thoshino@waseda.jp}{thoshino@waseda.jp}
\end{flushleft}

\section{Appendix: Notations and Definitions}\label{app:note_def}

To facilitate the proofs of the theorems, we introduce the following notation, some of which have already been introduced in the main text.
\begin{itemize}
    \item $\bm{Y}_n = \mathbf{X}_n\bm{\beta}_n + \bm{G}_n \circ \bm{L}_n + \tilde{\mathcal{E}}_n$, where $\bm{Y}_n = (Y_1, \ldots, Y_n)^\top$, $\mathbf{X}_n\bm{\beta}_n = (X_1^\top \beta(T_1, S_1), \ldots, X_n^\top \beta(T_n, S_n))^\top$, $\bm{G}_n = (g(T_1, S_1), \ldots, g(T_n, S_n))^\top$, $\bm{L}_n = (\lambda(U_1, V_1), \ldots, \lambda(U_n, V_n))^\top$, and $\tilde{\mathcal{E}}_n = (\tilde \epsilon_1, \ldots , \tilde \epsilon_n)^\top$.
    \item Basis terms:
    \begin{align*}
       \begin{array}{lll}
        \multicolumn{2}{l}{\ol{P}_{UV}(u,v) = P_{UV}(u,v) - \bbE [P_{UV}(U,V)]}&\\
        \multicolumn{2}{l}{\hat{P}_{UV}(u,v) = P_{UV}(u,v) - \frac{1}{n}\sum_{i=1}^n P_{UV}(\hat U_i, \hat V_i)}&\\ 
        &&\\
        \Pi_{1,i} = P_X(X_i, T_i, S_i), & \Pi_{2,i} = \ol{P}(T_i, S_i, U_i, V_i), & \hat \Pi_{2,i} = \hat P(T_i, S_i, \hat U_i, \hat V_i) \\ 
       \Pi_i = (\Pi_{1,i}^\top,  \Pi_{2,i}^\top)^\top, & \hat \Pi_i = (\Pi_{1,i}^\top, \hat \Pi_{2,i}^\top)^\top & \\
        \bm{\Pi}_{1, n} = (\Pi_{1,1}, \ldots, \Pi_{1,n})^\top, & \bm{\Pi}_{2, n} = (\Pi_{2,1}, \ldots, \Pi_{2,n})^\top, & \hat{\bm{\Pi}}_{2, n} = (\hat \Pi_{2, 1}, \ldots, \hat \Pi_{2, n})^\top \\
        \bm{\Pi}_n = (\Pi_1, \ldots, \Pi_n)^\top, & \hat{\bm{\Pi}}_n = (\hat \Pi_1, \ldots, \hat \Pi_n)^\top & \\
        &&\\
        P^*_{TS} = \int_{\mathcal{TS}} P_{TS}(t,s)\mathrm{d}t\mathrm{d}s, & \ol{P}^*_{UV} = \int_0^1 \int_0^1 \ol{P}_{UV}(u,v) \omega(u,v) \mathrm{d}u\mathrm{d}v, & \hat{P}^*_{UV} = \int_0^1 \int_0^1 \hat{P}_{UV}(u,v) \hat \omega_n(u,v) \mathrm{d}u\mathrm{d}v
         \end{array}      
    \end{align*}
    \item The pseudo true coefficient vector: $\theta^* = (\theta_\beta^{*\top}, \theta^{*\top}_{g\lambda})^\top$ with $\theta_{g\lambda}^* = \theta_g^* \otimes \theta^*_\lambda$.
    \item An infeasible series coefficient estimator: $\tilde \theta_n = (\tilde \theta_{n,\beta}^\top,  \tilde \theta_{n,g\lambda}^\top)^\top \coloneqq \left[\bm{\Pi}_n^\top \bm{\Pi}_n + \tau_n\bm{D}n\right]^{-1}\bm{\Pi}_n^\top \bm{Y}_n$
    \item Series approximation bias: $b_i \coloneqq X_i^\top \beta(T_i, S_i) + g(T_i, S_i) \lambda(U_i, V_i) - \Pi_i^\top\theta^*$, and $\bm{B}_n = (b_1, \ldots, b_n)^\top$.
    \item Derivatives of $\beta(\cdot, \cdot)$:
    \begin{align*}
    \begin{array}{ll}
        \dot \beta_{l,T}(t,s) \coloneqq \partial \beta_l(t,s) / \partial t, & \dot \beta_T(t,s) = (\dot \beta_{1,T}(t,s), \ldots, \dot \beta_{dx,T}(t,s))^\top \\
        \dot \beta_{l,S}(t,s) \coloneqq \partial \beta_l(t,s) / \partial s, & \dot \beta_S(t,s) = (\dot \beta_{1,S}(t,s), \ldots, \dot \beta_{dx,S}(t,s))^\top \\
        \ddot \beta_{l,TT}(t,s) \coloneqq \partial^2 \beta_l(t,s) / (\partial t)^2, & \ddot \beta_{TT}(t,s) = (\ddot \beta_{1,TT}(t,s), \ldots, \ddot \beta_{dx,TT}(t,s))^\top \\
        \ddot \beta_{l,TS}(t,s) \coloneqq \partial^2 \beta_l(t,s) / (\partial t \partial s), & \ddot \beta_{TS}(t,s) = (\ddot \beta_{1,TS}(t,s), \ldots, \ddot \beta_{dx,TS}(t,s))^\top \\
        \ddot \beta_{l,SS}(t,s) \coloneqq \partial^2 \beta_l(t,s) / (\partial s)^2, & \ddot \beta_{SS}(t,s) = (\ddot \beta_{1,SS}(t,s), \ldots, \ddot \beta_{dx,SS}(t,s))^\top
    \end{array}      
    \end{align*}
    \item $\delta (t,s) \coloneqq (\beta(t,s)^\top, h_T \dot \beta_T(t,s)^\top, h_S \dot \beta_S(t,s)^\top)^\top$
    \item $\gamma_i(t,s) \coloneqq \beta(T_i, S_i) - \beta(t,s) - \dot \beta_T(t,s) \cdot (T_i - t) - \dot \beta_S(t,s) \cdot (S_i - s)$
\end{itemize}

\begin{definition}\label{def:mixing}
    Let $\mathcal{A}$ and $\mathcal{B}$ be two sub $\sigma$-algebras.
    The $\alpha$-mixing coefficient between  $\mathcal{A}$ and $\mathcal{B}$ is defined as
    \begin{align*}
        \alpha(\mathcal{A}, \mathcal{B}) \coloneqq \sup_{A \in \mathcal{A}, \: B \in \mathcal{B}}\left| \Pr(A \cap B) - \Pr(A)\Pr(B)\right|.
    \end{align*}
    For a random field $\{X_i : i \in \mathcal{D}_N, \: N \ge 1\}$ and a set $\mathcal{K} \subseteq \mathcal{D}_N$, let $\sigma_N(X, \mathcal{K})$ denote the $\sigma$-algebra generated by $\{X_i: i \in \mathcal{K}\}$.
    The $\alpha$-mixing coefficient of $\{X_i\}$ is defined as
    \begin{align*}
        \alpha(k,l,r) \coloneqq \sup_N \sup_{\mathcal{K}, \mathcal{L} \subseteq \mathcal{D}_N} \{\alpha(\sigma_N(X, \mathcal{K}), \sigma_N(X, \mathcal{L})): |\mathcal{K}| \le k, |\mathcal{L}| \le l, \Delta(\mathcal{K}, \mathcal{L}) \ge r \},
    \end{align*}
    where $\Delta(\mathcal{K}, \mathcal{L}) \coloneqq \inf\{\Delta(i,j): i \in \mathcal{K}, j \in \mathcal{L}\}$.
    The random field $\{X_i\}$ is said to be $\alpha$-mixing if $\alpha(k,l,r) \to 0$ as $r \to \infty$ for any $k, l \in \mathbb{N}$.
\end{definition}

Since $T_i$ is a measurable function of $(X_i, Z_i, U_i)$, $\{T_i\}$ is $\alpha$-mixing with the mixing coefficient $\alpha(k,l,r)$ by Assumption \ref{as:mixing}(i).
Recalling the definition $S_i = m_i( \sum_{j \in \mathcal{P}_i} a_{i,j} T_j)$ and $a_{i,j}$ is non-zero only if $j \in \mathcal{P}_i$, Assumption \ref{as:D_space}(ii) implies that
\begin{align*}
    \sigma_N(S, \mathcal{K}) \subseteq \sigma_N(T, \{j \in \ol{\mathcal{P}}_i, i \in \mathcal{K}\}) \subseteq \sigma_N(T, \{j : \Delta(i,j) \le \ol{\Delta}, i \in \mathcal{K}\}).
\end{align*} 
It follows from Lemma A.1(ii) of \cite{jenish2009central} that $|\{j : \Delta(i,j) \le \ol{\Delta}, i \in \mathcal{K}\}| = |\mathcal{K}| O(\ol{\Delta}^d)$.
Thus, we can find that $\{S_i\}$ is $\alpha$-mixing with the mixing coefficient $\alpha^\dagger(k,l,r)$:
\begin{align}\label{eq:mixing2}
        \alpha^\dagger(k,l,r) =
        \begin{cases}
        1 & \text{for} \;\; r \le 2 \ol{\Delta} \\
        \alpha(C k \ol{\Delta}^d,  C l \ol{\Delta}^d, r - 2 \ol{\Delta}) & \text{for} \;\; r > 2 \ol{\Delta}.
        \end{cases}
\end{align}
for some positive constant $C > 0$.
From the same discussion as above, $\{V_i\}$ and $\{\tilde \epsilon_i\}$ are also $\alpha$-mixing with the mixing coefficient $\alpha^\dagger(k,l,r)$.


\section{Appendix: Proofs}\label{app:proofs}

\begin{lemma}\label{lem:UandV}
    Under Assumption \ref{as:exclusion}(ii), $\pi(X_i,Z_i)$ and $\eta(u)$ can be identified up to a location normalization for any $u \in (0,1)$.
\end{lemma}

\begin{proof}
    Without loss of generality, suppose that $\eta(0.5) = 0$.
    For any $u \in (0,1)$, by Assumption \ref{as:exclusion}, we have
    \begin{align*}
        \Pr(T_i \le \pi(X_i,Z_i) + \eta(u) \mid \bm{X}_{\ol{\mathcal{P}}_i}) 
        & = \Pr(\eta(U_i) \le \eta(u) \mid \bm{X}_{\ol{\mathcal{P}}_i}) \\
        & = \Pr(U_i \le u \mid \bm{X}_{\ol{\mathcal{P}}_i}) = u.    
    \end{align*}
    Thus, by solving the moment equation $\Pr(T_i \le a \mid \bm{X}_{\ol{\mathcal{P}}_i}) = u$ with respect to $a$, we can identify $\pi(X_i,Z_i) + \eta(u)$.
    In particular, when $u = 0.5$, $\pi(X_i,Z_i)$ is identified, which automatically implies the identification of $\eta(u)$ for other $u$'s.
\end{proof}

\begin{flushleft}
    \textbf{Proof of Proposition \ref{prop:S_control}}
\end{flushleft}

We first show the uniqueness of $V_i$.
Define $\tilde T_i(v) \coloneqq \pi(X_i,Z_i) + \eta(v)$ and $\tilde S_i(v) \coloneqq m_i(\sum_{j \in \mathcal{P}_i} a_{i,j} \tilde T_j(v))$.
By assumption, $\tilde T_i(v)$ is strictly increasing.
Hence, $\tilde S_i(v)$ is also strictly increasing since $m_i$ is strictly increasing and $a_{i,j}$'s are all positive.
This implies that $S_i \in [\tilde S_i(0), \tilde S_i(1)]$.
In addition, by the intermediate value theorem, there exists $v^* \in [0,1]$ such that $\tilde S_i(v^*) = S_i$ holds, and such $v^*$ is unique.
Specifically, recalling that $\sum_{j \in \mathcal{P}_i} a_{i,j} = 1$, the $\tilde S^{-1}_i$ function can be implicitly written as
\begin{align*}
    \tilde S^{-1}_i(s) = \Biggl\{v \in [0,1] \; \Bigl| \; m_i^{-1}(s) - \sum_{j \in \mathcal{P}_i} a_{i,j} \pi(X_j, Z_j) = \eta(v) \Biggr\}.
\end{align*}
In particular, we can see that $V_i  = \tilde S^{-1}_i(S_i)$ holds.
As such, the value of $V_i$ is uniquely determined for each $S_i$ as a function of $\bm{U}_{\mathcal{P}_i}$.

This $V_i$ is a valid control variable to control for the endogeneity between $\epsilon_i$ and $S_i$:
\begin{align*}
    \bbE[\epsilon_i \mid S_i = s, V_i = v, \bm{X}_{\ol{\mathcal{P}}_i} = \bm{x}] 
    & = \bbE\left[ \epsilon_i \: \biggl| \: \sum_{j \in \mathcal{P}_i}a_{i,j} \pi(x_j, Z_j) = m_i^{-1}(s) - \eta(v), V_i = v , \bm{X}_{\ol{\mathcal{P}}_i} = \bm{x} \right] \\ 
    & = \bbE[ \epsilon_i \mid V_i = v, \bm{X}_{\ol{\mathcal{P}}_i} = \bm{x}].
\end{align*}
This completes the proof. \qed

\begin{flushleft}
    \textbf{Proof of Theorem \ref{thm:atr_ident}}
\end{flushleft}

By Lemma \ref{lem:UandV}, we can identify $\pi(X_i, Z_i)$ and $\eta(u)$ for any $u \in (0,1)$.
As discussed in the main manuscript, this leads to the identification of $U_i$ and $V_i$.
Then, by Proposition \ref{prop:S_control}, we have \eqref{eq:mtr_ident}, which implies
\begin{align*}
    \mathrm{CATR}(t, s, x) 
    & = \int_{\mathrm{supp}( U_i, V_i, \bm{X}_{\mathcal{P}_i} \mid X_i = x)} \mathrm{MTR}(t, s, u, v, \bm{x}_{-i}, x) f_{U_i V_i \bm{X}_{\mathcal{P}_i}|X_i}(u,v,\bm{x}_{-i} \mid X_i = x)\mathrm{d}u\mathrm{d}v\mathrm{d}\bm{x}_{-i},
\end{align*}
where $f_{U_i V_i \bm{X}_{\mathcal{P}_i}|X_i}$ is the conditional density function of $(U_i, V_i, \bm{X}_{\mathcal{P}_i})$ given $X_i$.
In order for the above integral to be well-defined, $\mathrm{MTR}(t, s, u, v, \bm{x}_{-i}, x)$ must be well-defined for all $(u,v,\bm{x}_{-i}) \in \mathrm{supp}( U_i, V_i, \bm{X}_{\mathcal{P}_i} \mid X_i = x)$.
This eventually requires that conditioning on $\{T_i = t, S_i = s\}$ does not limit the conditional support $\mathrm{supp}( U_i, V_i, \bm{X}_{\mathcal{P}_i} \mid X_i = x)$, i.e., Assumption \ref{as:full_support}.
\qed

\begin{lemma}\label{lem:matrix_LLN}
    Suppose that Assumptions \ref{as:D_space}, \ref{as:mixing},  \ref{as:basis}(i),(ii), \ref{as:eigen}, \ref{as:1st_stage}, and \ref{as:conv_rate}(i) hold. 
    Then, 
\begin{description}
    \item[(i)] $\left\| \bm{\Pi}_n^\top \bm{\Pi}_n/n - \bbE[\bm{\Pi}_n^\top \bm{\Pi}_n/n] \right\| = O_P\left(K_{TSUV} / \sqrt{n}\right) = o_P(1)$; 
    \item[(ii)] $\left\| \hat{\bm{\Pi}}_n^\top \hat{\bm{\Pi}}_n/n - \bbE[\bm{\Pi}_n^\top \bm{\Pi}_n/n]  \right\| = O_P\left(K_{TSUV} / \sqrt{n}\right) + O_P\left(\zeta_{\dagger} \sqrt{K_{TS}}  n^{-\nu}\right) = o_P(1)$; 
    \item[(iii)] $\left\| [\bm{\Pi}_n^\top \bm{\Pi}_n/n]^{-1} - \bbE[\bm{\Pi}_n^\top \bm{\Pi}_n/n]^{-1}  \right\| = O_P\left(K_{TSUV} / \sqrt{n}\right) = o_P(1)$; and
    \item[(iv)] $\left\|[ \hat{\bm{\Pi}}_n^\top \hat{\bm{\Pi}}_n/n]^{-1} - \bbE[\bm{\Pi}_n^\top \bm{\Pi}_n/n]^{-1} \right\| = O_P\left(K_{TSUV} / \sqrt{n}\right) + O_P\left(\zeta_{\dagger} \sqrt{K_{TS}}  n^{-\nu}\right) = o_P(1)$.
\end{description}
\end{lemma}

\begin{proof}
    (i) By definition,
    \begin{align*}
         \bm{\Pi}_n^\top \bm{\Pi}_n/n - \bbE[\bm{\Pi}_n^\top \bm{\Pi}_n/n]
        = \left( 
        \begin{array}{cc}
        \frac{1}{n}\sum_{i=1}^n \left( \Pi_{1,i} \Pi_{1,i}^\top - \bbE[ \Pi_{1,i} \Pi_{1,i}^\top] \right) & \frac{1}{n}\sum_{i=1}^n \left( \Pi_{1,i} \Pi_{2,i}^\top - \bbE[ \Pi_{1,i} \Pi_{2,i}^\top] \right) \\
        \frac{1}{n}\sum_{i=1}^n \left( \Pi_{2,i} \Pi_{1,i}^\top - \bbE[ \Pi_{2,i} \Pi_{1,i}^\top] \right) & \frac{1}{n}\sum_{i=1}^n \left( \Pi_{2,i} \Pi_{2,i}^\top - \bbE[ \Pi_{2,i} \Pi_{2,i}^\top] \right) 
        \end{array}
        \right).
    \end{align*}
    We first derive the convergence rate for the upper-left block of the above matrix.
    Let $\Pi_{1,i}^{(k)}$ denote the $k$-th element of $\Pi_{1,i}$, whose exact form is $\Pi_{1,i}^{(k)} = X_{a,i} p_{T,b}(T_i) p_{S,c}(S_i)$ for some $a \in \{1, \ldots, dx\}$, $b \in \{1, \ldots, K_T\}$, and $c \in \{1, \ldots, K_S\}$.
    Thus, since $\{\Pi_{1,i}^{(k)}\}$ is a measurable function of $(X_i, \{T_j\}_{j \in \ol{\mathcal{P}}_i})$, it is $\alpha$-mixing with the mixing coefficient $\alpha^\dagger$ as given in \eqref{eq:mixing2}.
    
    Observe that
    \begin{align*}
    \begin{split}
        & \bbE \left\| \frac{1}{n}\sum_{i=1}^n \left( \Pi_{1,i} \Pi_{1,i}^\top - \bbE[ \Pi_{1,i} \Pi_{1,i}^\top] \right)\right\|^2 \\
    	& = \frac{1}{n^2} \sum_{k_2 =1}^{dx K_{TS}} \sum_{k_1 = 1}^{dx K_{TS}}	\bbE\left\{\sum_{i=1}^n \left( \Pi_{1,i}^{(k_1)} \Pi_{1,i}^{(k_2)} - \bbE[ \Pi_{1,i}^{(k_1)} \Pi_{1,i}^{(k_2)}] \right) \right\}^2 \\
		& = \frac{1}{n^2} \sum_{k_2 =1}^{dx K_{TS}} \sum_{k_1 = 1}^{dx K_{TS}} \sum_{i=1}^n \mathrm{Var}\left( \Pi_{1,i}^{(k_1)} \Pi_{1,i}^{(k_2)} \right) + \frac{1}{n^2} \sum_{k_2 =1}^{dx K_{TS}}\sum_{k_1 = 1}^{dx K_{TS}}	\sum_{i=1}^n \sum_{j \neq i}^n \mathrm{Cov}\left(\Pi_{1,i}^{(k_1)} \Pi_{1,i}^{(k_2)}, \Pi_{1,j}^{(k_1)} \Pi_{1,j}^{(k_2)} \right) \\
		& \le  \frac{1}{n^2} \sum_{i=1}^n  \bbE \left[  \Pi_{1,i}^\top \Pi_{1,i} \Pi_{1,i}^\top \Pi_{1,i} \right] + \frac{1}{n^2} \sum_{k_2 =1}^{dx K_{TS}}\sum_{k_1 = 1}^{dx K_{TS}}	\sum_{i=1}^n \sum_{j \neq i}^n \mathrm{Cov}\left(\Pi_{1,i}^{(k_1)} \Pi_{1,i}^{(k_2)}, \Pi_{1,j}^{(k_1)} \Pi_{1,j}^{(k_2)} \right).
	\end{split}
    \end{align*}
    By Assumptions \ref{as:mixing}(ii) and \ref{as:basis}(i), it holds that $\Pi_{1,i}^\top \Pi_{1,i} \le c d_x K_{TS}$.
    Thus, the first term on the right-hand side is $O(K_{TS}^2 /n)$.
    Next, for the second term, using Billingsley's covariance inequality,
    \begin{align}\label{eq:covineq}
    \begin{split}
        \sum_{j \neq i}^n \left| \mathrm{Cov}\left(\Pi_i^{(k_1)} \Pi_i^{(k_2)}, \Pi_j^{(k_1)} \Pi_j^{(k_2)} \right) \right|
        & = \sum_{r = 1}^\infty \sum_{1 \le j \le n: \: \Delta(i, j) \in [r, r + 1)}\left|  \mathrm{Cov}\left(\Pi_i^{(k_1)} \Pi_i^{(k_2)}, \Pi_j^{(k_1)} \Pi_j^{(k_2)} \right)\right| \\
        & \le c_1 \sum_{r = 1}^\infty r^{d-1} \alpha^\dagger (1,1,r) \\
        & = c_1 \sum_{r = 1}^{2 \ol{\Delta}} r^{d-1} + c_1 \sum_{r = 2 \ol{\Delta} + 1}^\infty r^{d-1} \alpha(C \ol{\Delta}^d, C \ol{\Delta}^d, r - 2\ol{\Delta}) \\
        & \le c_1 \sum_{r = 1}^{2 \ol{\Delta}} r^{d-1} + c_2 \sum_{r = 1}^\infty (r + 2\ol{\Delta})^{d-1} \hat\alpha(r) = O(1),
    \end{split}
    \end{align}
    where the first equality is due to Assumption \ref{as:D_space}(i), and we have used Lemma A.1(iii) of \cite{jenish2009central} to derive the second line.
    Thus, the second-term is also $O(K_{TS}^2 /n)$.
    Hence, by Markov's inequality, $|| n^{-1}\sum_{i=1}^n ( \Pi_{1,i} \Pi_{1,i}^\top - \bbE[ \Pi_{1,i} \Pi_{1,i}^\top]) ||^2 = O_P(K_{TS}^2 /n)$.
    
    In the same way as above, we can show that $|| n^{-1}\sum_{i=1}^n ( \Pi_{1,i} \Pi_{2,i}^\top - \bbE[ \Pi_{1,i} \Pi_{2,i}^\top]) ||^2 = O_P(K_{TS} K_{TSUV} /n)$ and $|| n^{-1}\sum_{i=1}^n ( \Pi_{2,i} \Pi_{2,i}^\top - \bbE[ \Pi_{2,i} \Pi_{2,i}^\top]) ||^2 = O_P(K_{TSUV}^2 /n)$.
    This completes the proof.
    \bigskip
    
    (ii) By the triangle inequality,
	\begin{align*}
		\left\| \hat{\bm{\Pi}}_n^\top \hat{\bm{\Pi}}_n/n - \bbE[\bm{\Pi}_n^\top \bm{\Pi}_n/n]   \right\|
		\le \left\| \hat{\bm{\Pi}}_n^\top \hat{\bm{\Pi}}_n/n - \bm{\Pi}_n^\top \bm{\Pi}_n/n \right\| + \left\|\bm{\Pi}_n^\top \bm{\Pi}_n/n - \bbE[\bm{\Pi}_n^\top \bm{\Pi}_n/n] \right\|.
	\end{align*}
	The second term is $O_P(K_{TSUV}/ \sqrt{n} )$ by (i).
	For the first term, 
	\begin{align}\label{eq:Pihat_decomp}
		\left\| \hat{\bm{\Pi}}_n^\top \hat{\bm{\Pi}}_n/n - \bm{\Pi}_n^\top \bm{\Pi}_n/n \right\|
		\le \left\| \left( \hat{\bm{\Pi}}_n^\top - \bm{\Pi}_n^\top \right) \left( \hat{\bm{\Pi}}_n - \bm{\Pi}_n \right) /n \right\| + 2 \left\| \bm{\Pi}_n^\top \left( \hat{\bm{\Pi}}_n - \bm{\Pi}_n \right) /n \right\|.
	\end{align}
    With Assumption \ref{as:1st_stage}, the mean value expansion gives
    \begin{align*}
        & \hat P_{UV}(\hat U_i, \hat V_i) - \ol{P}_{UV}(U_i, V_i) \\
        & = P_U(\hat U_i) \otimes P_V(\hat V_i) - P_U(U_i) \otimes P_V(V_i) + \bbE[P_{UV}(U, V)] - \frac{1}{n}\sum_{i=1}^n P_{UV}(\hat U_i, \hat V_i) \\
        & = [P_U(\hat U_i) - P_U(U_i)]\otimes P_V(\hat V_i) + P_U(U_i)\otimes [ P_V(\hat V_i) - P_V(V_i)] + \bbE[P_{UV}(U, V)] - \frac{1}{n}\sum_{i=1}^n P_{UV}(\hat U_i, \hat V_i) \\
        & = [\partial P_U(\bar U_i)/\partial u \otimes P_V(\hat V_i) + P_U(U_i)\otimes \partial P_V(\bar V_i)/\partial v ] \cdot O_P( n^{-\nu}) \\
        & \quad - \frac{1}{n}\sum_{i=1}^n [\partial P_U(\bar U_i)/\partial u \otimes P_V(\hat V_i) + P_U(U_i)\otimes \partial P_V(\bar V_i)/\partial v ] \cdot O_P( n^{-\nu}) + \bbE[P_{UV}(U, V)] - \frac{1}{n}\sum_{i=1}^n P_{UV}(U_i, V_i),
    \end{align*}
    where $\bar U_i \in [U_i, \hat U_i]$ and $\bar V_i \in [V_i, \hat V_i]$.
    For the last term of the right-hand side, we have 
    \begin{align*}
        \left\| \bbE[P_{UV}(U, V)] - \frac{1}{n}\sum_{i=1}^n P_{UV}(U_i, V_i) \right\| = O_P\left(\sqrt{K_{UV}/n}\right)
    \end{align*}
    by the same argument as in the proof of (i).
    Thus, 
    \begin{align*}
       \left\| \hat P_{UV}(\hat U_i, \hat V_i) - \ol{P}_{UV}(U_i, V_i) \right\| 
       = O_P(\zeta_{\dagger} n^{-\nu}) + O_P\left(\sqrt{K_{UV}/n}\right),
    \end{align*}
    for all $i$, where recall that $\zeta_{\dagger} = \zeta_U \sqrt{K_V} + \zeta_V \sqrt{K_U}$.
    Hence, since
	\begin{align*}
		\hat \Pi_i - \Pi_i 
		= \left[
		\begin{array}{c}
		\bm{0}_{dx K_{TS} \times 1} \\
		\hat \Pi_{2,i} - \Pi_{2,i}
		\end{array}
		\right]
		& = \left[
		\begin{array}{c}
		\bm{0}_{dx K_{TS} \times 1} \\
		P_{TS}(T_i, S_i) \otimes [\hat P_{UV}(\hat U_i, \hat V_i) - \ol{P}_{UV}(U_i, V_i)]
		\end{array}
		\right],
	\end{align*}
    we have $|| \hat \Pi_i - \Pi_i || = O_P\left(\zeta_{\dagger}\sqrt{K_{TS}} n^{-\nu}\right) + O_P\left(\sqrt{K_{TSUV}/n}\right)$ by Assumption \ref{as:basis}(i).
	Thus, 
	\begin{align*}
	\begin{split}
		\left\| \left( \hat{\bm{\Pi}}_n - \bm{\Pi}_n \right)/\sqrt{n} \right\| 
		& = O_P\left(\zeta_{\dagger}\sqrt{K_{TS}} n^{-\nu}\right) + O_P\left(\sqrt{K_{TSUV}/n}\right).
	\end{split}
	\end{align*}
	Hence, the first term of \eqref{eq:Pihat_decomp} satisfies 
    \begin{align*}
        \left\| \left( \hat{\bm{\Pi}}_n^\top - \bm{\Pi}_n^\top \right) \left( \hat{\bm{\Pi}}_n - \bm{\Pi}_n \right) /n \right\| \le \left\| \left( \hat{\bm{\Pi}}_n - \bm{\Pi}_n \right)/\sqrt{n} \right\|^2 = O_P\left(\zeta_{\dagger}^2 K_{TS} n^{-2\nu}\right) + O_P\left(K_{TSUV}/n\right).
    \end{align*}
	For the second term of \eqref{eq:Pihat_decomp}, we have
	\begin{align*}
	\begin{split}
		\left\| \bm{\Pi}_n^\top \left( \hat{\bm{\Pi}}_n - \bm{\Pi}_n \right) /n \right\|^2
		& = \mathrm{tr} \left\{ \left( \hat{\bm{\Pi}}_n - \bm{\Pi}_n \right)^\top \bm{\Pi}_n \bm{\Pi}_n^\top \left( \hat{\bm{\Pi}}_n - \bm{\Pi}_n \right) /n^2 \right\} \\
		& \leq [c + o_P(1)] \cdot \mathrm{tr} \left\{  \left( \hat{\bm{\Pi}}_n - \bm{\Pi}_n \right)^\top   \left( \hat{\bm{\Pi}}_n - \bm{\Pi}_n \right) /n \right\} \\
		& = [c + o_P(1)] \cdot \left\| \left( \hat{\bm{\Pi}}_n - \bm{\Pi}_n \right)/\sqrt{n} \right\|^2
		= O_P\left(\zeta_{\dagger}^2 K_{TS} n^{-2\nu}\right) + O_P\left(K_{TSUV}/n\right),
	\end{split}
	\end{align*}
	since result (i) and Assumption \ref{as:eigen} imply that $\rho_{\max}(\bm{\Pi}_n \bm{\Pi}_n^\top/n) \le c + o_P(1)$.
	Combining these results gives the desired result.
	\bigskip
	
	(iii), (iv) Noting the equality, $A^{-1} - B^{-1} = A^{-1}(B - A)B^{-1}$ for nonsingular matrices $A$ and $B$, the results follow from (i) and (ii) with Assumption \ref{as:eigen}.
\end{proof}


\begin{lemma}\label{lem:coef_conv}
    Suppose that Assumptions \ref{as:D_space}--\ref{as:error}, \ref{as:basis}(i), \ref{as:eigen}, \ref{as:series_approximation}, and \ref{as:conv_rate}(i) hold. 
    Then,
    \begin{description}
    \item[(i)] $|| \tilde \theta_{n, \beta} - \theta_\beta^* || = O_P\left(\sqrt{\frac{\kappa_n}{n}} + b_\mu + \tau_n^* \right)$; and
    \item[(ii)] $|| \tilde \theta_{n, g\lambda} - \theta_{g\lambda}^* || = O_P\left(\frac{\kappa_n}{\sqrt{n}} + b_\mu + \tau_n^*\right)$.
    \end{description}
\end{lemma}

\begin{proof}

(i) Let $\mathbb{S}_\beta \coloneqq (I_{dxK_{TS}}, \: \bm{0}_{dxK_{TS} \times K_{TSUV}})$ so that $\tilde \theta_{n, \beta} = \mathbb{S}_\beta \tilde \theta_n$.
Observe that
\begin{align*}
\begin{split}
    \tilde \theta_{n, \beta} - \theta_\beta^*
    & = \mathbb{S}_\beta \left[\bm{\Pi}_n^\top \bm{\Pi}_n + \tau_n\bm{D}n\right]^{-1}\bm{\Pi}_n^\top(\mathbf{X}_n\bm{\beta}_n + \bm{G}_n \circ \bm{L}_n + \tilde{\mathcal{E}}_n ) - \theta_\beta^* \\
    & = \mathbb{S}_\beta \left[\bm{\Pi}_n^\top \bm{\Pi}_n + \tau_n\bm{D}n\right]^{-1}\bm{\Pi}_n^\top( \bm{\Pi}_n\theta^* +  \bm{B}_n + \tilde{\mathcal{E}}_n ) - \theta_\beta^*  \\
    & = \mathbb{S}_\beta \left[\bm{\Pi}_n^\top \bm{\Pi}_n + \tau_n\bm{D}n\right]^{-1}(\bm{\Pi}_n^\top\bm{\Pi}_n + \tau_n\bm{D}n - \tau_n\bm{D}n) \theta^* + \mathbb{S}_\beta \left[\bm{\Pi}_n^\top \bm{\Pi}_n + \tau_n\bm{D}n\right]^{-1}\bm{\Pi}_n^\top( \bm{B}_n + \tilde{\mathcal{E}}_n ) - \theta_\beta^*  \\
    & = -\tau_n\mathbb{S}_\beta \left[\bm{\Pi}_n^\top \bm{\Pi}_n + \tau_n\bm{D}n\right]^{-1} \bm{D}n\theta^* + \mathbb{S}_\beta \left[\bm{\Pi}_n^\top \bm{\Pi}_n + \tau_n\bm{D}n\right]^{-1}\bm{\Pi}_n^\top \bm{B}_n  + \mathbb{S}_\beta \left[\bm{\Pi}_n^\top \bm{\Pi}_n + \tau_n\bm{D}n\right]^{-1}\bm{\Pi}_n^\top \tilde{\mathcal{E}}_n \\
     & = \bm{C}_{1,n} + \bm{C}_{2,n} + \bm{C}_{3,n}, \;\; \text{say}.
\end{split}
\end{align*}
For $\bm{C}_{1,n}$, noting that $\rho_{\max}([\bm{\Pi}_n^\top \bm{\Pi}_n/n + \tau_n\bm{D}]^{-1}) = O_P(1)$ by Lemma \ref{lem:matrix_LLN}(i), Assumption \ref{as:eigen}, and the positive semidefiniteness of $\bm{D}$, we have
\begin{align}\label{eq:tau_bias}
\begin{split}
    \left\| \mathbb{S}_\beta \left[\bm{\Pi}_n^\top \bm{\Pi}_n + \tau_n\bm{D}n\right]^{-1} \bm{D}n\theta^* \right\|^2 
    & = \theta^{*\top}\bm{D}\left[\bm{\Pi}_n^\top \bm{\Pi}_n/n + \tau_n\bm{D}\right]^{-1}\mathbb{S}_\beta^\top \mathbb{S}_\beta \left[\bm{\Pi}_n^\top \bm{\Pi}_n/n + \tau_n\bm{D}\right]^{-1} \bm{D}\theta^* \\
    & \le O_P(1) \cdot \theta^{*\top}\bm{D}\theta^*,
\end{split}
\end{align}
where recall that we have assumed $\rho_{\max}(\bm{D}) = O(1)$.
Hence, $||\bm{C}_{1,n}|| = O_P(\tau_n^*)$.

Next, noting that $||\bm{B}_n||^2 = O(n b_\mu^2)$ by Assumptions \ref{as:mixing}(ii) and \ref{as:series_approximation}, then it is easy to see that $|| \bm{C}_{2,n}|| = O_P(b_\mu)$.

Finally, for $\bm{C}_{3,n}$, since $\bbE[\tilde \epsilon_i \mid \mathcal{F}_N] = 0$ and $\tilde \epsilon_i \indep \tilde \epsilon_j \mid \mathcal{F}_N$ by Assumption \ref{as:error}, we have
\begin{align}\label{eq:markov1}
    \begin{split}
    \bbE \left[ \left\| \bm{C}_{3,n} \right\|^2 \: \bigl| \: \mathcal{F}_N \right] 
    & =  \mathrm{tr}\left\{\mathbb{S}_\beta \left[\bm{\Pi}_n^\top \bm{\Pi}_n + \tau_n\bm{D}n\right]^{-1}\bm{\Pi}_n^\top \bbE\left[\tilde{\mathcal{E}}_n \tilde{\mathcal{E}}_n^\top \mid \mathcal{F}_N\right]\bm{\Pi}_n\left[\bm{\Pi}_n^\top \bm{\Pi}_n + \tau_n\bm{D}n\right]^{-1}\mathbb{S}_\beta^\top \right\}\\
    & \le O(1) \cdot \mathrm{tr}\left\{\mathbb{S}_\beta \left[\bm{\Pi}_n^\top \bm{\Pi}_n + \tau_n\bm{D}n\right]^{-1}\bm{\Pi}_n^\top \bm{\Pi}_n\left[\bm{\Pi}_n^\top \bm{\Pi}_n + \tau_n\bm{D}n\right]^{-1}\mathbb{S}_\beta^\top \right\}\\
    & \le O_P(1) \cdot \mathrm{tr}\left\{\mathbb{S}_\beta \mathbb{S}_\beta^\top /n\right\} = O_P(K_{TS}/n). 
    \end{split}
\end{align}
Thus, Markov's inequality implies that $||\bm{C}_{3,n}|| = O_P(\sqrt{K_{TS}/n})$.
This completes the proof.
\bigskip

(ii) The proof of (ii) is analogous and is omitted.
\end{proof}

\begin{flushleft}
    \textbf{Proof of Theorem \ref{thm:1st_stage}}
\end{flushleft}

    (i) Let $\hat b_i \coloneqq X_i^\top \beta(T_i, S_i) + g(T_i, S_i) \lambda(\hat U_i, \hat V_i) - - \Pi_{1,i}^\top \theta_\beta^* - \hat \Pi_{2,i}^\top\theta^*_{g\lambda}$.
    It is easy to see that 
    \begin{align*}
        \hat \Pi_{2,i}^\top \theta^*_{g\lambda} 
        & = (P_{TS}(T_i,S_i)^\top\theta_g^*) \cdot (\hat P_{UV}(\hat U_i, \hat V_i)^\top \theta_\lambda^*) \\
        & = (P_{TS}(T_i,S_i)^\top\theta_g^*) \cdot (\ol{P}_{UV}(\hat U_i, \hat V_i)^\top \theta_\lambda^*) + (P_{TS}(T_i,S_i)^\top\theta_g^*) \cdot ([\hat P_{UV}(\hat U_i, \hat V_i) - \ol{P}_{UV}(\hat U_i, \hat V_i)]^\top \theta_\lambda^*).
    \end{align*}
    Furthermore, noting that $\bbE[\lambda(U, V)] = 0$, for all $(u,v)$,
    \begin{align}\label{eq:P_diff_theta}
        \begin{split}
        [\hat P_{UV}(u, v) - \ol{P}_{UV}(u, v)]^\top \theta_\lambda^*
        & = \bbE[P_{UV}(U, V)^\top\theta_\lambda^*] - \frac{1}{n}\sum_{i=1}^n P_{UV}(\hat U_i, \hat V_i)^\top \theta_\lambda^*\\
        & = \bbE[\lambda(U, V)] - \frac{1}{n}\sum_{i=1}^n \biggl( \underbrace{P_{UV}(\hat U_i, \hat V_i) - \bbE[P_{UV}(U_i, V_i)]}_{= \: \ol{P}_{UV}(\hat U_i, \hat V_i)}\biggr)^\top  \theta_\lambda^* \\
        & = \bbE[\lambda(U, V)] - \frac{1}{n}\sum_{i=1}^n \lambda(\hat U_i, \hat V_i) + \frac{1}{n}\sum_{i=1}^n \left( \lambda(\hat U_i, \hat V_i) - \ol{P}_{UV}(\hat U_i, \hat V_i)^\top  \theta_\lambda^*\right) \\
        & =  \bbE[\lambda(U, V)] - \frac{1}{n}\sum_{i=1}^n \lambda( U_i, V_i) + \frac{1}{n}\sum_{i=1}^n \left[\lambda( U_i, V_i) - \lambda(\hat U_i, \hat V_i)\right] + O(K_{UV}^{-\mu_\lambda}) \\
        & = O_P(n^{-1/2}) + O_P(n^{-\nu}) + O(K_{UV}^{-\mu_\lambda}) 
        \end{split}
    \end{align}
    by the law of large numbers and Assumptions \ref{as:series_approximation} and \ref{as:1st_stage}.
    Hence, we have 
    \begin{align}\label{eq:bias_eval}
        \max_{1 \le i \le n} |\hat b_i| = O_P(b_\mu + n^{-\nu}).
    \end{align}
    Let $\hat{\bm{L}}_n = (\lambda(\hat U_1, \hat V_1), \ldots, \lambda(\hat U_n, \hat V_n))^\top$, and $\hat{\bm{B}}_n = \mathbf{X}_n\bm{\beta}_n + \bm{G}_n \circ \hat{\bm{L}}_n - \hat{\bm{\Pi}}_n\theta^*$.
    Noting that $\rho_{\max}([\hat{\bm{\Pi}}_n^\top \hat{\bm{\Pi}}_n/n + \tau_n\bm{D}]^{-1}) = O_P(1)$ by Lemma \ref{lem:matrix_LLN}(ii), Assumption \ref{as:eigen}, and the positive semidefiniteness of $\bm{D}$, we have
    \begin{align*}
        \hat \theta_{n, \beta} - \theta_\beta^*
        & = \mathbb{S}_\beta \left[\hat{\bm{\Pi}}_n^\top \hat{\bm{\Pi}}_n + \tau_n\bm{D}n\right]^{-1}\hat{\bm{\Pi}}_n^\top(\mathbf{X}_n\bm{\beta}_n + \bm{G}_n \circ \bm{L}_n + \tilde{\mathcal{E}}_n ) - \theta_\beta^* \\
        & = \mathbb{S}_\beta \left[\hat{\bm{\Pi}}_n^\top \hat{\bm{\Pi}}_n + \tau_n\bm{D}n\right]^{-1}\hat{\bm{\Pi}}_n^\top(\hat{\bm{\Pi}}_n \theta^* + \hat{\bm{B}}_n + \bm{G}_n \circ (\bm{L}_n - \hat{\bm{L}}_n)  + \tilde{\mathcal{E}}_n ) - \theta_\beta^* \\
        & = -\tau_n\mathbb{S}_\beta \left[\hat{\bm{\Pi}}_n^\top \hat{\bm{\Pi}}_n + \tau_n\bm{D}n\right]^{-1} \bm{D}n\theta^* + \mathbb{S}_\beta \left[\hat{\bm{\Pi}}_n^\top \hat{\bm{\Pi}}_n + \tau_n\bm{D}n\right]^{-1}\hat{\bm{\Pi}}_n^\top \hat{\bm{B}}_n \\
        & \quad + \mathbb{S}_\beta \left[\hat{\bm{\Pi}}_n^\top \hat{\bm{\Pi}}_n + \tau_n\bm{D}n\right]^{-1}\hat{\bm{\Pi}}_n^\top[\bm{G}_n \circ (\bm{L}_n - \hat{\bm{L}}_n)] + \mathbb{S}_\beta \left[\hat{\bm{\Pi}}_n^\top \hat{\bm{\Pi}}_n + \tau_n\bm{D}n\right]^{-1}\hat{\bm{\Pi}}_n^\top \tilde{\mathcal{E}}_n \\
         & = \mathbb{S}_\beta \left[\hat{\bm{\Pi}}_n^\top \hat{\bm{\Pi}}_n + \tau_n\bm{D}n\right]^{-1}\hat{\bm{\Pi}}_n^\top[\bm{G}_n \circ (\bm{L}_n - \hat{\bm{L}}_n)] + \mathbb{S}_\beta \left[\hat{\bm{\Pi}}_n^\top \hat{\bm{\Pi}}_n + \tau_n\bm{D}n\right]^{-1}\hat{\bm{\Pi}}_n^\top \tilde{\mathcal{E}}_n \\
         & \quad + O_P(\tau_n^*) + O_P(b_\mu + n^{-\nu}),
    \end{align*}
    where the last equality is from the same calculation as in \eqref{eq:tau_bias} and \eqref{eq:bias_eval}.
    Note that each $i$-th element of $\bm{G}_n \circ (\bm{L}_n - \hat{\bm{L}}_n)$ is bounded by $|g(T_i, S_i)[\lambda(U_i, V_i) - \lambda(\hat U_i, \hat V_i)]| = O_P(n^{-\nu})$.
    Thus, the first term of the right-hand side is $O_P(n^{-\nu})$.
    For the second term, by the same argument as in \eqref{eq:markov1}, we can see that it is of order $O_P(\sqrt{K_{TS}/n})$.
    \bigskip
    
    (ii) Let $\mathbb{S}_{g\lambda} \coloneqq (\mathbf{0}_{K_{TSUV} \times dxK_{TS}}, \: I_{K_{TSUV}})$ so that $\hat \theta_{n, g\lambda} = \mathbb{S}_{g\lambda} \hat \theta_n$.
    By the same discussion as above,
    \begin{align}\label{eq:theta_decomp1}
    \begin{split}
        \hat \theta_{n, g\lambda} - \theta_{g\lambda}^* 
        & = \mathbb{S}_{g\lambda}\left[\hat{\bm{\Pi}}_n^\top \hat{\bm{\Pi}}_n + \tau_n\bm{D}n\right]^{-1}\hat{\bm{\Pi}}_n^\top(\hat{\bm{\Pi}}_n \theta^* + \hat{\bm{B}}_n + \bm{G}_n \circ (\bm{L}_n - \hat{\bm{L}}_n)  + \tilde{\mathcal{E}}_n  ) - \theta_{g\lambda}^* \\
        & =  O_P(\tau_n^*) + O_P(b_\mu + n^{-\nu}) + O_P(n^{-\nu}) + \mathbb{S}_{g\lambda}\left[\hat{\bm{\Pi}}_n^\top \hat{\bm{\Pi}}_n + \tau_n\bm{D}n\right]^{-1}\hat{\bm{\Pi}}_n^\top \tilde{\mathcal{E}}_n.
    \end{split}
    \end{align}
    Although it is possible to derive the order of the remaining term by directly applying Markov's inequality, for later use, we further decompose it into three parts:
     \begin{align}\label{eq:theta_decomp2}
    \begin{split}
        \mathbb{S}_{g\lambda}\left[\hat{\bm{\Pi}}_n^\top \hat{\bm{\Pi}}_n + \tau_n\bm{D}n\right]^{-1}\hat{\bm{\Pi}}_n^\top \tilde{\mathcal{E}}_n
        & = \mathbb{S}_{g\lambda}\left[ \bm{\Pi}_n^\top \bm{\Pi}_n + \tau_n\bm{D}n\right]^{-1}\bm{\Pi}_n^\top \tilde{\mathcal{E}}_n \\
        & \quad + \mathbb{S}_{g\lambda}\left\{ \left[ \hat{\bm{\Pi}}_n^\top \hat{\bm{\Pi}}_n + \tau_n\bm{D}n\right]^{-1} - \left[ \bm{\Pi}_n^\top \bm{\Pi}_n + \tau_n\bm{D}n\right]^{-1} \right\}\bm{\Pi}_n^\top \tilde{\mathcal{E}}_n\\
        & \quad +\mathbb{S}_{g\lambda}\left[ \hat{\bm{\Pi}}_n^\top \hat{\bm{\Pi}}_n + \tau_n\bm{D}n\right]^{-1}(\hat{\bm{\Pi}}_n -\bm{\Pi}_n)^\top  \tilde{\mathcal{E}}_n\\
        & = \mathbb{S}_{g\lambda}\left[ \bm{\Pi}_n^\top \bm{\Pi}_n + \tau_n\bm{D}n\right]^{-1}\bm{\Pi}_n^\top \tilde{\mathcal{E}}_n + \bm{F}_{1,n} + \bm{F}_{2,n}, \;\; \text{say}.
    \end{split}
    \end{align}
    Clearly, the first term is $O_P\left(\sqrt{K_{TSUV}/n}\right)$ by Markov's inequality.

    Recalling that the left-hand side of \eqref{eq:Pihat_decomp} is $O_P\left(\zeta_{\dagger} \sqrt{K_{TS}}  n^{-\nu}\right) + O_P\left(\sqrt{K_{TSUV}/n}\right)$, we can see that the maximum eigenvalue of  $[ \hat{\bm{\Pi}}_n^\top \hat{\bm{\Pi}}_n/n + \tau_n\bm{D} ]^{-1} - [ \bm{\Pi}_n^\top \bm{\Pi}_n/n + \tau_n\bm{D}]^{-1}$ is bounded by $O_P\left(\zeta_{\dagger} \sqrt{K_{TS}}  n^{-\nu}\right) + O_P\left(\sqrt{K_{TSUV}/n}\right)$.
    Then using Markov's inequality in a similar way to \eqref{eq:markov1}, we obtain 
    \begin{align*}
        ||\bm{F}_{1,n}|| = O_P\left(\sqrt{ K_{TSUV} K_{TS}} \zeta_{\dagger} n^{-(\nu + 1/2)} \right) + O_P\left(K_{TSUV}/n\right) = o_P\left(\sqrt{K_{TSUV}/n}\right).
    \end{align*}
    Similarly, since the maximum eigenvalue of $( \hat{\bm{\Pi}}_n - \bm{\Pi}_n)^\top ( \hat{\bm{\Pi}}_n - \bm{\Pi}_n ) /n $ is $O_P\left(\zeta_{\dagger}^2 K_{TS} n^{-2\nu}\right) + O_P\left(K_{TSUV}/n\right)$, we have $||\bm{F}_{2,n}|| = O_P\left(\sqrt{ K_{TSUV} K_{TS}} \zeta_{\dagger} n^{-(\nu + 1/2)} \right) + O_P\left(K_{TSUV}/n\right) = o_P\left(\sqrt{K_{TSUV}/n}\right)$.
    \qed

\begin{lemma}\label{lem:IAE_lambda}
    Suppose that Assumptions \ref{as:D_space}--\ref{as:basis}(iii), \ref{as:eigen}--\ref{as:conv_rate}(i), and \ref{as:misc}(i)--(ii) are satisfied.
    Then, we have
    \begin{align*}
        \int_0^1 \int_0^1 | \hat \lambda_n(u, v) - \lambda(u, v)| \mathrm{d}u\mathrm{d}v = O_P\left(\sqrt{\kappa_n/n} + \tau_n^* + b_\mu + n^{-\nu}\right).
    \end{align*}
\end{lemma}

\begin{proof}
    Observe that, by \eqref{eq:scale_normal},
    \begin{align*}
        \hat \lambda_n(u, v) - \lambda(u, v)
        & = P^*_{TS} \otimes \hat P_{UV}(u, v) - \int_{\mathcal{TS}} g(t,s)\lambda(u, v) \mathrm{d}t\mathrm{d}s \\
        & = [P^*_{TS} \otimes  \ol{P}_{UV}(u,v)]^\top [\hat \theta_{n,g\lambda} - \theta_{g\lambda}^*] + [P^*_{TS} \otimes (\hat P_{UV}(u, v) - \ol{P}_{UV}(u, v))]^\top [\hat \theta_{n, g\lambda} - \theta_{g\lambda}^*]\\
        & \quad  + (P^{*\top}_{TS} \theta_g^*) \cdot [\hat P_{UV}(u, v) - \ol{P}_{UV}(u, v)]^\top \theta_\lambda^* + \int_{\mathcal{TS}}\left( \ol{P}(t,s,u,v)^\top\theta_{g\lambda}^* - g(t,s)\lambda(u, v) \right)\mathrm{d}t\mathrm{d}s\\ 
        & = [P^*_{TS} \otimes  \ol{P}_{UV}(u,v)]^\top [\hat \theta_{n,g\lambda} - \theta_{g\lambda}^*] \\
        & \quad + O_P(\zeta_{\dagger}n^{-\nu} + \sqrt{K_{UV}/n} ) \cdot || \hat \theta_{n, g\lambda} - \theta_{g\lambda}^* || + O_P(n^{-\nu} + K_{TS}^{-\mu_g} + K_{UV}^{-\mu_\lambda})
    \end{align*}
    uniformly in $(u,v)$, by Assumption \ref{as:basis}(iii) and \eqref{eq:P_diff_theta}.
    
    Here, we define a $K_{UV} \times K_{TSUV}$ block matrix $\Gamma_{TS}$ as follows:
    \begin{align*}
        \Gamma_{TS} \coloneqq \left(
        \begin{array}{ccccccccccccc}
            P_{TS}^{*(1)} & 0 & \cdots & 0 & P_{TS}^{*(2)} & 0 & \cdots & 0 & & P_{TS}^{*(K_{TS})} & 0 & \cdots & 0 \\
            0 & P_{TS}^{*(1)} &        & \vdots  & 0 & P_{TS}^{*(2)} &        & \vdots & \cdots &  0 & P_{TS}^{*(K_{TS})} &        & \vdots \\
            \vdots &      & \ddots & \vdots  & \vdots &      & \ddots & \vdots & \cdots & \vdots &      & \ddots & \vdots \\
            0 & \cdots & \cdots &  P_{TS}^{*(1)} & 0 & \cdots & \cdots &  P_{TS}^{*(2)} & & 0 & \cdots & \cdots & P_{TS}^{*(K_{TS})} 
        \end{array}
        \right),
    \end{align*}
    where $P_{TS}^{*(k)}$ is the $k$-th element of $P^*_{TS}$.
    Then, notice that we can re-write $[P^*_{TS} \otimes  \ol{P}_{UV}(u,v)]^\top[\hat\theta_{n, g\lambda} - \theta_{g\lambda}^*] =  \ol{P}_{UV}(u,v)^\top \Gamma_{TS} [\hat\theta_{n, g\lambda} - \theta_{g\lambda}^*]$.
    Note that Assumption \ref{as:basis}(iii) implies that $\rho_{\max}( \Gamma_{TS}^\top  \Gamma_{TS}) = O(1)$.
    By \eqref{eq:theta_decomp1} and \eqref{eq:theta_decomp2}, we have
    \begin{align*}
        \ol{P}_{UV}(u,v)^\top \Gamma_{TS} [ \hat\theta_{n, g\lambda} - \theta_{g\lambda}^*] 
        & = \ol{P}_{UV}(u,v)^\top R_n\\
        & \quad + \ol{P}_{UV}(u,v)^\top \Gamma_{TS}\bm{F}_{1,n} + \ol{P}_{UV}(u,v)^\top \Gamma_{TS}\bm{F}_{2,n} + \ol{P}_{UV}(u,v)^\top \Gamma_{TS}\bm{F}_{3,n},
    \end{align*}
    uniformly in $(u, v)$, where $\bm{F}_{1,n}$ and $\bm{F}_{2,n}$ are as defined in \eqref{eq:theta_decomp1}, $||\bm{F}_{3,n}|| = O_P(\tau^*_n + b_\mu + n^{-\nu})$, and 
    \begin{align*}
        R_n = \Gamma_{TS}\mathbb{S}_{g\lambda}\left[ \bm{\Pi}_n^\top \bm{\Pi}_n/n + \tau_n\bm{D}\right]^{-1}\bm{\Pi}_n^\top \tilde{\mathcal{E}}_n/n.
    \end{align*}
    Similarly as in the proof of Theorem \ref{thm:1st_stage}(ii), we can easily find that $||\ol{P}_{UV}(u,v)^\top \Gamma_{TS}\bm{F}_{1,n}|| = O_P\left(\zeta_{\dagger} \sqrt{K_{TSUV}} n^{-(\nu + 1/2)} \right) + O_P\left(\sqrt{K_{TS}} K_{UV}/n\right)$ and  $||\ol{P}_{UV}(u,v)^\top \Gamma_{TS}\bm{F}_{2,n}|| = O_P\left(\zeta_{\dagger} \sqrt{K_{TSUV}} n^{-(\nu + 1/2)} \right) + O_P\left(\sqrt{K_{TS}} K_{UV}/n\right)$.

    Summarizing the results so far,
    \begin{align*}
        \left| \hat \lambda_n(u, v) - \lambda(u, v) \right|
        & \le O_P\left(n^{-\nu} + K_{TS}^{-\mu_g} + K_{UV}^{-\mu_\lambda}\right) + O_P\left(\zeta_{\dagger} \sqrt{K_{TSUV}} n^{-(\nu + 1/2)} \right) + O_P\left(\sqrt{K_{TS}} K_{UV}/n\right) \\
        & \quad + O_P\left(\zeta_{\dagger}n^{-\nu} + \sqrt{K_{UV}/n}\right) \cdot || \hat \theta_{n,g\lambda} - \theta^*_{g\lambda} ||  + \left| \ol{P}_{UV}(u,v)^\top R_n \right| + \left|\ol{P}_{UV}(u,v)^\top \Gamma_{TS}\bm{F}_{3,n} \right|
    \end{align*}
    uniformly in $(u,v)$.
    By H\"older's inequality and Assumption \ref{as:misc}(i),
    \begin{align*}
        \int_0^1 \int_0^1 \left| \hat \lambda_n(u, v) - \lambda(u, v) \right|\mathrm{d}u\mathrm{d}v 
        & \le c \int_0^1 \int_0^1 \left| \hat \lambda_n(u, v) - \lambda(u, v) \right| f_{UV}(u,v)\mathrm{d}u\mathrm{d}v \\
        & \le c \left\| \hat \lambda_n(U, V) - \lambda(U, V) \right\|_{UV,2},
    \end{align*}
    where $||a(U,V)||_{UV,2} = \sqrt{\int_0^1\int_0^1 |a(u,v)|^2 f_{UV}(u,v)\mathrm{d}u\mathrm{d}v}$.
    Then, by Assumption \ref{as:misc}(ii), we can easily find that $\left\| \ol{P}_{UV}(U,V)^\top R_n\right\|_{UV,2} = O_P(\sqrt{K_{UV}/n})$ and $\left\|\ol{P}_{UV}(U,V)^\top \Gamma_{TS}\bm{F}_{3,n}\right\|_{UV,2} = O_P(\tau^*_n + b_\mu + n^{-\nu})$.
    Thus,
    \begin{align*}
        \int_0^1 \int_0^1 \left| \hat \lambda_n(u, v) - \lambda(u, v) \right|\mathrm{d}u\mathrm{d}v 
        & =  O_P\left(\sqrt{K_{UV}/n} + \tau_n^* + b_\mu + n^{-\nu}\right) +  O_P\left(\zeta_{\dagger}n^{-\nu} + \sqrt{K_{UV}/n}\right) \cdot || \hat \theta_{n,g\lambda} - \theta^*_{g\lambda} ||  \\
        & = O_P\left(\sqrt{K_{UV}/n} + \tau_n^* + b_\mu + n^{-\nu}\right),
    \end{align*}
    by Theorem \ref{thm:1st_stage}(ii).
\end{proof}


\begin{flushleft}
    \textbf{Proof of Theorem \ref{thm:2nd_stage}}
\end{flushleft}

    (i) Using the decomposition in the proof of Lemma \ref{lem:IAE_lambda}, we can easily see that
    \begin{align*}
        \hat \lambda_n(u, v) - \lambda(u, v)
        & = O_P\left(\sqrt{K_{UV}}(\tau_n^* + b_\mu + n^{-\nu})\right) + o_P\left(\sqrt{K_{UV}/n}\right) + \ol{P}_{UV}(u,v)^\top R_n 
    \end{align*}
    uniformly in $(u, v) \in [0, 1]^2$.

    Here, we construct a partition $\mathcal{U}_n \coloneqq \{0, u_1,  u_2, \ldots, 1\}$ of $[0,1]$ in a way that, for any $(u, v) \in [0,1]^2$, there exists a point $(u^*, v^*) \in \mathcal{U}_n^2$ satisfying $||(u, v) - (u^*, v^*)|| = O(K_{UV}^{-\xi_1})$, where $\xi_1$ is as given in Assumption \ref{as:basis}(iv).
	Then, we have
	\begin{align*}
		\sup_{(u,v) \in [0,1]^2}\left| \ol{P}_{UV}(u,v)^\top R_n \right|
		&\le  \max_{(u, v) \in \mathcal{U}_n^2} \left| \ol{P}_{UV}(u,v)^\top R_n  \right|  +  \sup_{(u, v) \in [0,1]^2}\left| \left\{ P_{UV}(u, v) - P_{UV}(u^*, v^*) \right\}^\top  R_n \right|  \\
		& \le \max_{(u, v) \in \mathcal{U}_n^2} \left| \ol{P}_{UV}(u,v)^\top R_n \right|  +  O(1) \cdot \left\| R_n \right\| \\
		& = \max_{(u, v) \in \mathcal{U}_n^2} \left| \ol{P}_{UV}(u,v)^\top R_n \right| +  O_P\left(\sqrt{K_{UV}/n}\right),
	\end{align*}
	where the last equality is due to Markov's inequality.
    To derive the bound on the first term on the right-hand side, decompose $\tilde\epsilon_i = e_{1,i} + e_{2,i}$, where
	\begin{align*}
		e_{1,i} & \coloneqq \tilde\epsilon_i \mathbf{1}\{|\tilde\epsilon_i| \le M_n \} - \bbE[\tilde\epsilon_i \mathbf{1}\{|\tilde\epsilon_i| \le M_n \} \mid \mathcal{F}_N], \\
		e_{2,i} & \coloneqq \tilde\epsilon_i \mathbf{1}\{|\tilde\epsilon_i| > M_n \} - \bbE[\tilde\epsilon_i \mathbf{1}\{|\tilde\epsilon_i| > M_n \} \mid \mathcal{F}_N],
	\end{align*}
	and $M_n$ is a sequence of positive constants tending to infinity.
	Let $E_{1,n} = (e_{1,1}, \ldots , e_{1,n})^\top$ and $E_{2,n} = (e_{2,1}, \ldots , e_{2,n})^\top$.
	Further, let 
	\begin{align*}
	    q_i(u,v) 
	    & \coloneqq \ol{P}_{UV}(u,v)^\top \Gamma_{TS}\mathbb{S}_{g\lambda}\left[ \bm{\Pi}_n^\top \bm{\Pi}_n/n + \tau_n\bm{D}\right]^{-1} \Pi_i/n
	\end{align*} so that
	\begin{align*}
			\sum_{i=1}^n q_i(u,v)e_{1,i} = \ol{P}_{UV}(u,v)^\top \Gamma_{TS}\mathbb{S}_{g\lambda}\left[ \bm{\Pi}_n^\top \bm{\Pi}_n/n + \tau_n\bm{D}\right]^{-1} \bm{\Pi}_n^\top E_{1,n} /n.
	\end{align*}
	Note that $\bbE[q_i(u,v)e_{1,i} \mid \mathcal{F}_N] = 0$.
	Furthermore, it is straightforward to see that there exists a positive constant $c_1 > 0$ such that
	\begin{align*}
		|q_i(u,v)| 
		& \le || \ol{P}_{UV}(u,v) || \cdot  || \Gamma_{TS}\mathbb{S}_{g\lambda}\left[ \bm{\Pi}_n^\top \bm{\Pi}_n / n+ \tau_n\bm{D}\right]^{-1} \Pi_i || /n  
        \le  c_1 K_{UV}\sqrt{K_{TS}}/n 
	\end{align*} 
	uniformly in $u$ with probability approaching one.
	Moreover, uniformly in $u$,
	\begin{align*}
		\sum_{i=1}^n q_i(u, v)^2 
		\le c_2 K_{UV}/n 
	\end{align*} 
    for some $c_2 > 0$ wih probability approaching one.
	Therefore, for all $(u, v) \in \mathcal{U}_n^2$, we have $|q_i(u, v) e_{1,i}| \le c'_1 K_{UV} \sqrt{K_{TS}} M_n/n$ and $\sum_{i=1}^n \bbE[(q_i(u, v) e_{1,i})^2 \mid \mathcal{F}_N] = \sum_{i=1}^n q_i(u, v)^2  \bbE[e_{1,i}^2 \mid \mathcal{F}_N] \le c_2' K_{UV} /n$.
	Then, by Bernstein's inequality, we have
	\begin{align*}
		\Pr\left( \left| \sum_{i=1}^n q_i(u, v) e_{1,i} \right| > t \:\: \biggl| \:\:  \mathcal{F}_N\right) 
		& \le c \exp\left(- \frac{1}{2}\frac{t^2}{c_2' K_{UV} /n + t c'_1 K_{UV} \sqrt{K_{TS}} M_n /(3n)}\right) \\
		& \le c \exp\left(- \frac{t^2}{ c'(K_{UV} /n) ( 1 + t \sqrt{K_{TS}} M_n)}\right). 	
	\end{align*}
	Hence,
	\begin{align*}
		& \Pr\left(\max_{(u, v) \in \mathcal{U}_n^2}\left| \ol{P}_{UV}(u,v)^\top \Gamma_{TS}\mathbb{S}_{g\lambda}\left[ \bm{\Pi}_n^\top \bm{\Pi}_n/n + \tau_n\bm{D}\right]^{-1} \bm{\Pi}_n^\top E_{1,n} /n \right| > t \:\: \biggl| \:\:  \mathcal{F}_N \right) \\
		& \le c |\mathcal{U}_n|^2 \exp\left(- \frac{t^2}{ c'(K_{UV} /n) ( 1 + t \sqrt{K_{TS}} M_n)}\right) 
        = c \exp\left\{2 \ln |\mathcal{U}_n| - \frac{t^2}{ c'(K_{UV} /n) ( 1 + t \sqrt{K_{TS}} M_n)}\right\}.
	\end{align*}
	Then, setting $t = C \sqrt{(\kappa_n \ln \kappa_n) / n)}$ for a large constant $C > 0$, provided that $M_n$ grows sufficiently slowly so that $t \sqrt{\kappa_n} M_n = o(1)$, we have
	\begin{align*}
	    2 \ln |\mathcal{U}_n| - \frac{ C^2 (\kappa_n \ln \kappa_n)/n }{c' (K_{UV}/n) [1 + t \sqrt{K_{TS}} M_n)] }
		\asymp \ln\left(\frac{ |\mathcal{U}_n|^2 }{\kappa_n^{C^2} } \right).
	\end{align*}
	Here, recall that the length of each sub-interval in $\mathcal{U}_n$ is of order $O(K_{UV}^{-\xi_1})$, and thus the cardinality of $\mathcal{U}_n$ grows at rate $K_{UV}^{\xi_1}$.
	Thus, for sufficiently large $C$, we have $|\mathcal{U}_n|^2/\kappa_n^{C^2}  \to 0$, implying that 
	\begin{align*}
		\max_{(u, v) \in \mathcal{U}_n^2}\left| \ol{P}_{UV}(u,v)^\top \Gamma_{TS}\mathbb{S}_{g\lambda}\left[ \bm{\Pi}_n^\top \bm{\Pi}_n/n + \tau_n\bm{D}\right]^{-1} \bm{\Pi}_n^\top E_{1,n} /n \right| = O_P\left(\sqrt{(\kappa_n \ln \kappa_n) / n)}\right).
	\end{align*}
	
	Next, by Markov's inequality and Assumption \ref{as:error}, it holds that
	\begin{align*}
		& \Pr\left(\max_{(u,v) \in \mathcal{U}_n^2}\left| \ol{P}_{UV}(u,v)^\top \Gamma_{TS}\mathbb{S}_{g\lambda}\left[ \bm{\Pi}_n^\top \bm{\Pi}_n/n + \tau_n\bm{D}\right]^{-1} \bm{\Pi}_n^\top E_{2,n} /n \right| > t  \:\: \biggl| \:\:  \mathcal{F}_N \right) \\
		& \le \Pr\left(c \sqrt{K_{UV}}\left\| \Gamma_{TS}\mathbb{S}_{g\lambda}\left[ \bm{\Pi}_n^\top \bm{\Pi}_n/n + \tau_n\bm{D}\right]^{-1} \bm{\Pi}_n^\top E_{2,n} /n \right\| > t  \:\: \biggl| \:\:  \mathcal{F}_N \right) \\
		& \leq \Pr\left( \frac{c' \sqrt{K_{TS}} K_{UV}}{n} \sum_{i = 1}^n |e_{2,i}| > t  \:\: \biggl| \:\:  \mathcal{F}_N \right)\\
		& \leq \frac{2 c' \sqrt{K_{TS}} K_{UV}}{t n} \sum_{i=1}^n \bbE[|\tilde \epsilon_i| \cdot \mathbf{1}\{|\tilde \epsilon_i| > M_n \} \mid \mathcal{F}_N] \\
		& \leq \frac{2 c' \sqrt{K_{TS}} K_{UV}}{t n M_n^3} \sum_{i=1}^n \bbE[|\tilde \epsilon_i|^4 \cdot \mathbf{1}\{|\tilde \epsilon_i| > M_n \} \mid \mathcal{F}_N] = O\left( \frac{\sqrt{K_{TS}} K_{UV}}{t M_n^3} \right).
	\end{align*}
	Again, setting $t = C \sqrt{(\kappa_n \ln \kappa_n) / n)}$, if $\sqrt{n}\kappa_n / \sqrt{\ln \kappa_n} = O(M_n^3)$,
	\begin{align*}
		\frac{\sqrt{K_{TS}} K_{UV}}{t M_n^3}
		& \asymp \frac{1}{C}\frac{\sqrt{n}\kappa_n}{\sqrt{\ln \kappa_n} M_n^3} \to 0,
	\end{align*}
	as $C \to \infty$, which implies 
	\begin{align*}
		\max_{(u,v) \in \mathcal{U}_n^2}\left|\ol{P}_{UV}(u,v)^\top \Gamma_{TS}\mathbb{S}_{g\lambda}\left[ \bm{\Pi}_n^\top \bm{\Pi}_n/n + \tau_n\bm{D}\right]^{-1} \bm{\Pi}_n^\top E_{2,n} /n \right| =  O_P\left(\sqrt{(\kappa_n \ln \kappa_n) / n)}\right).
	\end{align*}
	It should be noted that the two requirements on $M_n$, $t \sqrt{\kappa_n} M_n = o(1)$ and $\sqrt{n} \kappa_n / \sqrt{\ln \kappa_n} = O(M_n^3)$, can be satisfied simultaneously: that is, when we set specifically $c^3 \sqrt{n} \kappa_n / \sqrt{\ln \kappa_n} = M_n^3$,
	\begin{align*}
        M_n = \frac{ c n^{1/6} \kappa_n^{1/3}}{ (\ln \kappa_n)^{1/6}} 
        \Rightarrow \sqrt{\frac{\kappa_n^2 \ln \kappa_n}{n}} M_n = \frac{c \kappa_n^{4/3} (\ln \kappa_n)^{2/6}}{n^{2/6}} \to 0
	\end{align*}
	under Assumption \ref{as:conv_rate}(ii).
	Thus, we have
	\begin{align*}
		\max_{(u, v) \in \mathcal{U}_n^2} \left| \ol{P}_{UV}(u,v)^\top R_n \right| = O_P\left(\sqrt{(\kappa_n \ln \kappa_n) / n)}\right).
	\end{align*}
	By combining these results, the proof is completed.
	
    \bigskip

    (ii) The proof is almost the same as that of (i).
    Decompose
    \begin{align*}
        \hat g_n(t,s) - g(t,s) 
        & = [P_{TS}(t,s) \otimes \hat{P}^*_{UV}]^\top \hat \theta_{n, g\lambda} - \int_0^1 \int_0^1 g(t,s) \lambda(u,v) \omega(u,v) \mathrm{d}u\mathrm{d}v \\
        & = [P_{TS}(t,s) \otimes \ol{P}^*_{UV}]^\top[\hat \theta_{n, g\lambda} - \theta_{g\lambda}^*] + [P_{TS}(t,s) \otimes (\hat{P}^*_{UV} -  \ol{P}^*_{UV})]^\top [\hat \theta_{n, g\lambda} - \theta_{g\lambda}^*] \\
        & \quad + (P_{TS}(t,s)^\top \theta_g^*) \cdot (\hat{P}^*_{UV} -  \ol{P}^*_{UV})^\top \theta_\lambda^* \\
        & \quad + \int_0^1 \int_0^1 \left( \ol{P}(t,s,u,v)^\top \theta_{g\lambda}^* - g(t,s)\lambda(u,v)\right) \omega(u,v) \mathrm{d}u\mathrm{d}v \\    
        & = [P_{TS}(t,s) \otimes \ol{P}^*_{UV}]^\top[\hat \theta_{n, g\lambda} - \theta_{g\lambda}^*] + [P_{TS}(t,s) \otimes (\hat{P}^*_{UV} -  \ol{P}^*_{UV})]^\top [\hat \theta_{n, g\lambda} - \theta_{g\lambda}^*] \\
        & \quad + (P_{TS}(t,s)^\top \theta_g^*) \cdot (\hat{P}^*_{UV} -  \ol{P}^*_{UV})^\top \theta_\lambda^* + O(K_{TS}^{-\mu_g} + K_{UV}^{-\mu_\lambda}).
    \end{align*}
    Observe that
    \begin{align*}
        & \left\| \hat P_{UV}(u,v)\hat \omega_n(u, v) - \ol{P}_{UV}(u,v) \omega(u, v) \right\|\\
        & \le  \left\| \hat P_{UV}(u,v) - \ol{P}_{UV}(u,v) \right\| \cdot |\hat \omega_n(u, v)| +  \left\| \ol{P}_{UV}(u,v) \right\| \cdot |\hat \omega_n(u, v) - \omega(u, v)| \\
        & = O_P\left(\zeta_{\dagger}n^{-\nu} + \sqrt{K_{UV}/n}\right) + O_P\left(\sqrt{K_{UV}}\right) \cdot |\hat \lambda_n(u,v) - \lambda(u, v)| \\
        & \quad + O_P\left(\sqrt{K_{UV}}\right) \cdot O_P\left(\sqrt{K_{UV}/n} + \tau_n^* + b_\mu + n^{-\nu}\right)
    \end{align*}
    uniformly in $(u,v)$, where the last equality follows Theorem \ref{thm:2nd_stage}(i), Lemma \ref{lem:IAE_lambda}, and
    \begin{align*}
        |\hat \omega_n(u, v) - \omega(u, v)|
        & \le  \frac{\left| \hat \lambda_n (u, v) - \lambda (u, v) \right| }{\int_0^1\int_0^1 \hat \lambda_n (u', v')^2 \mathrm{d}u'\mathrm{d}v'} + \frac{\left| \lambda(u, v) \int_0^1\int_0^1 [\lambda(u', v')^2 - \hat \lambda_n(u', v')^2] \mathrm{d}u'\mathrm{d}v' \right|}{ \int_0^1\int_0^1 \hat \lambda_n(u', v')^2 \mathrm{d}u'\mathrm{d}v'\int_0^1\int_0^1 \lambda(u', v')^2 \mathrm{d}u'\mathrm{d}v'} \\
        & \le  \frac{\left| \hat \lambda_n (u, v) - \lambda (u, v) \right| }{\int_0^1\int_0^1 \hat \lambda_n (u', v')^2 \mathrm{d}u'\mathrm{d}v'} + \frac{ |\lambda(u, v)| \int_0^1\int_0^1 |\lambda(u', v') - \hat \lambda_n(u', v')| \cdot |\lambda(u', v') + \hat \lambda_n(u', v')| \mathrm{d}u'\mathrm{d}v' }{ \int_0^1\int_0^1 \hat \lambda_n(u', v')^2 \mathrm{d}u'\mathrm{d}v'\int_0^1\int_0^1 \lambda(u', v')^2 \mathrm{d}u'\mathrm{d}v'} \\
        & = O_P(1)\cdot \left| \hat \lambda_n (u, v) - \lambda (u, v) \right| + O_P(1) \cdot \int_0^1\int_0^1 |\lambda(u, v) - \hat \lambda_n(u, v)| \mathrm{d}u \mathrm{d}v.
    \end{align*}
    Thus, 
    \begin{align*}
        \left\| P_{TS}(t,s) \otimes (\hat{P}^*_{UV} -  \ol{P}^*_{UV}) \right\|
        & = \left\| P_{TS}(t,s) \otimes \int_0^1 \int_0^1 \left(\hat P_{UV}(u,v)\hat \omega_n(u, v) - \ol{P}_{UV}(u,v) \omega(u, v)\right)\mathrm{d}u\mathrm{d}v \right\| \\
        & \le O(\sqrt{K_{TS}}) \int_0^1 \int_0^1 \left\| \hat P_{UV}(u,v)\hat \omega_n(u, v) - \ol{P}_{UV}(u,v) \omega(u, v) \right\| \mathrm{d}u\mathrm{d}v\\
        & \le O_P\left(\zeta_{\dagger}\sqrt{\kappa_n}n^{-\nu} + \kappa_n/\sqrt{n}\right) + O_P\left(\kappa_n\right) \int_0^1 \int_0^1 |\hat \lambda_n(u,v) - \lambda(u, v)| \mathrm{d}u\mathrm{d}v \\
        & \quad + O_P\left(\sqrt{\kappa_n^3/n} + \kappa_n(\tau_n^* + b_\mu + n^{-\nu})\right)\\
        & = O_P\left(\zeta_{\dagger}\sqrt{\kappa_n}n^{-\nu} + \kappa_n/\sqrt{n}\right) + O_P\left(\sqrt{\kappa_n^3/n} + \kappa_n(\tau_n^* + b_\mu + n^{-\nu})\right)
    \end{align*} 
    by Lemma \ref{lem:IAE_lambda}, implying that    
    \begin{align*}
        \left| [P_{TS}(t,s) \otimes (\hat{P}^*_{UV} - \ol{P}^*_{UV})]^\top [\hat \theta_{n, g\lambda} - \theta_{g\lambda}^*] \right|
        & \le  O_P\left(\zeta_{\dagger}\sqrt{\kappa_n}n^{-\nu} + \kappa_n/\sqrt{n}\right)  \cdot || \hat \theta_{n, g\lambda} - \theta_{g\lambda}^*|| \\
        & \quad + O_P\left(\sqrt{\kappa_n^3/n} + \kappa_n(\tau_n^* + b_\mu + n^{-\nu})\right) \cdot ||\hat \theta_{n, g\lambda} - \theta_{g\lambda}^*|| \\
        & = O_P\left(\sqrt{\kappa_n/n} + \tau_n^* + b_\mu + n^{-\nu}\right)
    \end{align*}
    by Assumption \ref{as:misc}(iii) and Theorem \ref{thm:1st_stage}(ii).
    Similarly, by \eqref{eq:P_diff_theta}, we have
    \begin{align*}
        & \left| [\hat P_{UV}(u,v)\hat \omega_n(u, v) - \ol{P}_{UV}(u,v) \omega(u, v)]^\top \theta_\lambda^* \right| \\
        & \le \left| [\hat P_{UV}(u,v) - \ol{P}_{UV}(u,v)]^\top \theta_\lambda^* \right| \cdot |\hat\omega_n(u, v)|  + (\ol{P}_{UV}(u,v)^\top \theta_\lambda^*) \cdot |\hat\omega_n(u, v) - \omega(u, v)| \\
        & = O_P(n^{-\nu} + K_{UV}^{-\mu_\lambda}) + O_P(1) \cdot |\hat \lambda_n(u,v) - \lambda(u, v)| + O_P\left(\sqrt{\kappa_n/n} + \tau_n^* + b_\mu + n^{-\nu}\right)
    \end{align*}
    uniformly in $(u,v)$, implying that 
    \begin{align*}
        \left| (P_{TS}(t,s)^\top \theta_g^*) \cdot (\hat{P}^*_{UV} -  \ol{P}^*_{UV})^\top \theta_\lambda^* \right|
        & = O(1) \cdot \int_0^1 \int_0^1 |\hat \lambda_n(u,v) - \lambda(u, v)| \mathrm{d}u\mathrm{d}v + O_P\left(\sqrt{\kappa_n/n} + \tau_n^* + b_\mu + n^{-\nu}\right)\\
        & = O_P\left(\sqrt{\kappa_n/n} + \tau_n^* + b_\mu + n^{-\nu}\right)
    \end{align*}
    by Lemma \ref{lem:IAE_lambda}.

    Next, we define a $K_{TS} \times K_{TSUV}$ block matrix $\Gamma_{UV}$ as follows:
        \begin{align*}
            \Gamma_{UV} \coloneqq \left(
            \begin{array}{cccc}
                \ol{P}^{*\top}_{UV} & \bm{0}_{1 \times K_{UV}} & \cdots & \bm{0}_{1 \times K_{UV}}\\
                \bm{0}_{1 \times K_{UV}} & \ol{P}^{*\top}_{UV} &        & \vdots \\
                \vdots                   &                  & \ddots & \vdots \\
                \bm{0}_{1 \times K_{UV}} & \cdots & \cdots &  \ol{P}^{*\top}_{UV}
            \end{array}
            \right).
        \end{align*}
    Then, we can re-write $[P_{TS}(t,s) \otimes \ol{P}^*_{UV}]^\top[\hat \theta_{n, g\lambda} - \theta_{g\lambda}^*] = P_{TS}(t,s)^\top \Gamma_{UV} [\hat \theta_{n, g\lambda} - \theta_{g\lambda}^*]$.
    By the same argument as above, we have
    \begin{align*}
        P_{TS}(t,s)^\top \Gamma_{UV} [\hat \theta_{n, g\lambda} - \theta_{g\lambda}^*]
        & = O_P\left(\sqrt{K_{TS}}(\tau_n^* + b_\mu + n^{-\nu})\right) + o_P\left(\sqrt{K_{TS}/n}\right) + P_{TS}(t,s)^\top R_n,
    \end{align*}
    uniformly in $(t,s)$, where (with an abuse of notation) $ R_n = \Gamma_{UV}\mathbb{S}_{g\lambda}\left[ \bm{\Pi}_n^\top \bm{\Pi}_n / n + \tau_n\bm{D}\right]^{-1}\bm{\Pi}_n^\top \tilde{\mathcal{E}}_n/n$.
    
    Similarly as before, we construct a partition $\mathcal{T}_n \mathcal{S}_n$ of $\mathcal{TS}$ such that, for any $(t,s) \in \mathcal{TS}$, there exists a point $(t^*, s^*) \in \mathcal{T}_n \mathcal{S}_n$ satisfying $||(t,s) - (t^*,s^*)|| = O(K_{TS}^{-\xi_2})$, where $\xi_2$ is as given in Assumption \ref{as:basis}(v).
	Then, we have
	\begin{align*}
		\sup_{(t,s) \in \mathcal{TS}}\left| P_{TS}(t,s)^\top R_n \right|
		& \le \max_{(t,s) \in \mathcal{T}_n \mathcal{S}_n} \left| P_{TS}(t,s)^\top R_n \right| +  O_P\left(\sqrt{K_{TS}/n}\right).
	\end{align*}
    Letting
	\begin{align*}
	    q_{i}(t,s) 
	    & \coloneqq  P_{TS}(t,s)^\top \Gamma_{UV}\mathbb{S}_{g\lambda}\left[ \bm{\Pi}_n^\top \bm{\Pi}_n / n + \tau_n\bm{D}\right]^{-1} \Pi_i/n,
	\end{align*} 
    we have $|q_i(t,s)| \le c_1 K_{TS}\sqrt{K_{UV}}/n$ and  $\sum_{i=1}^n q_i(t,s)^2 \le c_2 K_{TS}/n$ with probability approaching one.
	Then, by Bernstein's inequality, we have
	\begin{align*}
		& \Pr\left(\max_{(t,s) \in \mathcal{T}_n \mathcal{S}_n}\left| P_{TS}(t,s)^\top \Gamma_{UV}\mathbb{S}_{g\lambda}\left[ \bm{\Pi}_n^\top \bm{\Pi}_n/n + \tau_n\bm{D}\right]^{-1} \bm{\Pi}_n^\top E_{1,n} /n \right| > u \:\: \biggl| \:\:  \mathcal{F}_N \right) \\
		& \le c \exp\left\{\ln |\mathcal{T}_n \mathcal{S}_n| - \frac{u^2}{ c(K_{TS} /n) ( 1 + u \sqrt{K_{UV}} M_n)}\right\}.
	\end{align*}
	Then, setting $u = C \sqrt{(\kappa_n \ln \kappa_n) / n)}$ for a large constant $C > 0$, provided that $M_n$ grows sufficiently slowly so that $u \sqrt{\kappa_n} M_n = o(1)$, we have
	\begin{align*}
		\max_{(t,s) \in \mathcal{T}_n \mathcal{S}_n}\left| P_{TS}(t,s)^\top \Gamma_{UV}\mathbb{S}_{g\lambda}\left[ \bm{\Pi}_n^\top \bm{\Pi}_n/n + \tau_n\bm{D}\right]^{-1} \bm{\Pi}_n^\top E_{1,n} /n \right| = O_P\left(\sqrt{(\kappa_n \ln \kappa_n) / n)}\right).
	\end{align*}
	Next, by Markov's inequality, it holds that
	\begin{align*}
		& \Pr\left(\max_{(t,s) \in \mathcal{T}_n \mathcal{S}_n}\left|P_{TS}(t,s)^\top \Gamma_{UV}\mathbb{S}_{g\lambda}\left[ \bm{\Pi}_n^\top \bm{\Pi}_n/n + \tau_n\bm{D}\right]^{-1} \bm{\Pi}_n^\top E_{2,n} /n \right| > u  \:\: \biggl| \:\:  \mathcal{F}_N \right) \\
		& \le \Pr\left( \frac{c K_{TS} \sqrt{K_{UV}}}{n} \sum_{i = 1}^n |e_{2,i}| > u  \:\: \biggl| \:\:  \mathcal{F}_N \right) 
        = O\left( \frac{K_{TS} \sqrt{K_{UV}}}{u M_n^3} \right).
	\end{align*}
	Again, setting $u = C \sqrt{(\kappa_n \ln \kappa_n) / n)}$ for a large $C$, if $ \sqrt{n}\kappa_n / \sqrt{\ln \kappa_n} = O(M_n^3)$,
	\begin{align*}
		\max_{(t,s) \in \mathcal{T}_n \mathcal{S}_n}\left|P_{TS}(t,s)^\top \Gamma_{UV}\mathbb{S}_{g\lambda}\left[ \bm{\Pi}_n^\top \bm{\Pi}_n/n + \tau_n\bm{D}\right]^{-1} \bm{\Pi}_n^\top E_{2,n} /n \right| =  O_P\left(\sqrt{(\kappa_n \ln \kappa_n) / n)}\right).
	\end{align*}
	Thus, we have $\max_{(t,s) \in \mathcal{T}_n \mathcal{S}_n} \left| P_{TS}(t,s)^\top R_n \right| = O_P\left(\sqrt{(\kappa_n \ln \kappa_n) / n)}\right)$.
	By combining these results, the proof is completed.
    \qed


\bigskip

In what follows, we derive the asymptotic distribution of $\sqrt{n h_T h_S}(\tilde \beta_n(t, s) - \beta(t,s))$.
Noting that $X_i^\top \beta(T_i, S_i) = X_i^\top \gamma_i(t,s) + X_i(t,s)^\top \delta(t, s)$,
\begin{align*}
    & \sqrt{n h_T h_S}(\tilde \beta_n(t, s) - \beta(t,s)) = \mathbb{S}_{dx} \left[ \frac{1}{n}\sum_{i=1}^n X_i(t,s) X_i(t,s)^\top W_{i}(t,s) \right]^{-1} \\
    & \qquad\qquad \times \frac{\sqrt{h_T h_S}}{\sqrt{n}}\sum_{i=1}^n X_i(t,s) [X_i^\top \gamma_i(t,s) + g(T_i, S_i)\lambda(U_i, V_i) + \tilde \epsilon_i - \hat g_n(T_i, S_i) \hat \lambda_n(\hat U_i, \hat V_i)] W_{i}(t,s) \\
    & \eqqcolon \bm{H}_{1,n} + \bm{H}_{2,n} + \bm{H}_{3,n}, 
\end{align*}
where 
\begin{align*}
    \bm{H}_{1,n}
    & \coloneqq \mathbb{S}_{dx} \left[ \frac{1}{n}\sum_{i=1}^n X_i(t,s) X_i(t,s)^\top W_{i}(t,s) \right]^{-1} \frac{\sqrt{h_T h_S}}{\sqrt{n}}\sum_{i=1}^n X_i(t,s) X_i^\top \gamma_i(t,s)W_{i}(t,s) \\
    \bm{H}_{2,n}
    & \coloneqq \mathbb{S}_{dx} \left[ \frac{1}{n}\sum_{i=1}^n X_i(t,s) X_i(t,s)^\top W_{i}(t,s) \right]^{-1} \frac{\sqrt{h_T h_S}}{\sqrt{n}}\sum_{i=1}^n X_i(t,s) [g(T_i, S_i)\lambda(U_i, V_i) - \hat g_n(T_i, S_i) \hat \lambda_n(\hat U_i, \hat V_i)]W_{i}(t,s) \\
    \bm{H}_{3,n}
    & \coloneqq \mathbb{S}_{dx} \left[ \frac{1}{n}\sum_{i=1}^n X_i(t,s) X_i(t,s)^\top W_{i}(t,s) \right]^{-1} \frac{\sqrt{h_T h_S}}{\sqrt{n}}\sum_{i=1}^n X_i(t,s) \tilde \epsilon_i W_{i}(t,s).    
\end{align*}


\begin{lemma}\label{lem:kernel_LLN}
    Suppose that Assumptions \ref{as:D_space}, \ref{as:mixing}, \ref{as:omega_f}--\ref{as:band}(i), and \ref{as:alpha}(i) hold.
    Then,
\begin{align*}
    \frac{1}{n}\sum_{i = 1}^n X_i(t,s) X_i(t,s)^\top W_{i}(t,s) =
    \left(\begin{array}{ccc}
        \ol{\Omega}_{1}(t,s) & \mathbf{0}_{dx \times dx} & \mathbf{0}_{dx \times dx} \\
        \mathbf{0}_{dx \times dx} & \ol{\Omega}_{1}(t,s) \varphi_2^1 & \mathbf{0}_{dx \times dx} \\
        \mathbf{0}_{dx \times dx} & \mathbf{0}_{dx \times dx} & \ol{\Omega}_{1}(t,s) \varphi_2^1
    \end{array}\right)
    + o_P(1).
\end{align*}
\end{lemma}

\begin{proof}

For $(a,b) \in \{0,1,2\}$ such that $a + b \le 2$, by Assumption \ref{as:mixing}(ii), \ref{as:omega_f}, and \ref{as:kernel}, the bounded convergence theorem gives
\begin{align*}
    & \frac{1}{n}\sum_{i = 1}^n \bbE\left[  X_i X_i^\top W_{i}(t,s) \left( \frac{T_i - t}{h_T} \right)^a  \left( \frac{S_i - s}{h_T} \right)^b  \right] \\
    & = \frac{1}{n h_T h_S}\sum_{i = 1}^n \int_{\mathcal{TS}} \Omega_{1,i}(t',s') W\left(\frac{t' - t}{h_T} \right) W\left(\frac{s' - s}{h_S} \right) \left( \frac{t' - t}{h_T} \right)^a  \left( \frac{s' - s}{h_T} \right)^b f_i(t',s')\mathrm{d}t'\mathrm{d}s' \\
    & = \frac{1}{n}\sum_{i = 1}^n \int\int \Omega_{1,i}(t + h_T \phi_1, s + h_S \phi_2) W\left(\phi_1 \right) W\left(\phi_2 \right) \phi_1^a \phi_2^b f_i(t + h_T \phi_1, s + h_S \phi_2)\mathrm{d}\phi_1\mathrm{d}\phi_2 \to \ol{\Omega}_{1}(t,s) \varphi_a^1 \varphi_b^1
\end{align*}
as $n \to \infty$.
Note that by Assumption \ref{as:kernel}, $\varphi_0^1 = 1$ and $\varphi_1^1 = 0$ hold.

Next, for $1 \le l,k \le dx$, letting $W_{i,t,a} = W((T_i -t)/h_T)((T_i - t)/h_T)^a$ and $W_{i,s,b} = W((S_i -s)/h_S)((S_i - s)/h_S)^b$ for simplicity,
\begin{align*}
    \Var\left(\frac{1}{n h_T h_S}\sum_{i = 1}^n X_{l,i} X_{k,i} W_{i,t,a} W_{i,s,b}\right)
    & = \frac{1}{n^2 h_T^2 h_S^2}\sum_{i = 1}^n \Var\left(X_{l,i} X_{k,i} W_{i,t,a} W_{i,s,b}\right) \\
    & \quad + \frac{1}{n^2 h_T^2 h_S^2}\sum_{i = 1}^n\sum_{j \neq i}^n \Cov\left(X_{l,i} X_{k,i} W_{i,t,a} W_{i,s,b}, X_{l,j} X_{k,j} W_{j,t,a} W_{j,s,b}\right).
\end{align*}
For the first term,
\begin{align*}
    & \frac{1}{n^2 h_T^2 h_S^2}\sum_{i = 1}^n \Var\left(X_{l,i} X_{k,i} W_{i,t,a} W_{i,s,b}\right) \\ 
    & \le \frac{1}{n^2 h_T^2 h_S^2}\sum_{i = 1}^n \bbE( (X_{l,i} X_{k,i})^2 W_{i,t,a}^2 W_{i,s,b}^2 ) \\
    & \le \frac{c}{n^2 h_T^2 h_S^2}\sum_{i = 1}^n \int_{\mathcal{TS}}W\left(\frac{t' - t}{h_T} \right)^2\left| \frac{t' - t}{h_T} \right|^{2a} W\left(\frac{s' - s}{h_S} \right)^2 \left| \frac{s' - s}{h_S} \right|^{2b} f_i(t', s')\mathrm{d}t'\mathrm{d}s' \\
    & = \frac{c}{n^2 h_T h_S}\sum_{i = 1}^n  \int\int W\left(\phi_1 \right)^2 |\phi_1|^{2a} W\left(\phi_2 \right)^2 |\phi_2|^{2b} f_i(t + h_T\phi_1, s + h_S\phi_2)\mathrm{d}\phi_1\mathrm{d}\phi_2 = O((nh_T h_S)^{-1}).
\end{align*}
For the second term, note that $\{X_{l,i} X_{k,i} W_{i,t,a} W_{i,s,b}\}$ is a uniformly bounded $\alpha$-mixing process with the mixing coefficient $\alpha^\dagger$ in \eqref{eq:mixing2}.
Moreover, 
\begin{align*}
    \left\| X_{l,i} X_{k,i} W_{i,t,a} W_{i,s,b} \right\|_4^4
    & \le c \int_{\mathcal{TS}} \left| W\left(\frac{t' - t}{h_T} \right) \right|^4 \left| W\left(\frac{s' - s}{h_S} \right) \right|^4  \left| \frac{t' - t}{h_T} \right|^{4a}  \left| \frac{s' - s}{h_T} \right|^{4b} f_i(t',s')\mathrm{d}t'\mathrm{d}s' \\
    & = O(h_T h_S).
\end{align*}
Thus, by Davydov's inequality,
\begin{align}\label{eq:cov_davydov}
    \begin{split}
    & \sum_{j \neq i}^n \left| \Cov\left(X_{l,i} X_{k,i} W_{i,t,a} W_{i,s,b}, X_{l,j} X_{k,j} W_{j,t,a} W_{j,s,b}\right) \right|\\
    & = \sum_{r = 1}^\infty \sum_{1 \le j \le n: \: \Delta(i, j) \in [r, r + 1)}\left|  \Cov\left(X_{l,i} X_{k,i} W_{i,t,a} W_{i,s,b}, X_{l,j} X_{k,j} W_{j,t,a} W_{j,s,b}\right)\right| \\
    & \le c \sum_{r = 1}^\infty \sum_{1 \le j \le n: \: \Delta(i,j) \in [r, r + 1)} \alpha^\dagger(1,1,r)^{1/2}\left\| X_{l,i} X_{k,i} W_{i,t,a} W_{i,s,b} \right\|_4 \left\| X_{l,j} X_{k,j} W_{j,t,a} W_{j,s,b} \right\|_4 \\
    & \le O(\sqrt{h_T h_S}) \sum_{r = 1}^\infty r^{d-1}\alpha^\dagger(1,1,r)^{1/2} \\
    & \le O(\sqrt{h_T h_S}) \sum_{r = 1}^{2 \ol{\Delta}} r^{d-1} + O(\sqrt{h_T h_S}) \sum_{r = 1}^\infty (r + 2\ol{\Delta})^{d-1} \hat\alpha(r)^{1/2} = O(\sqrt{h_T h_S}),
    \end{split}
\end{align}
by Assumption \ref{as:alpha}(i).\footnote{
    Comparing the general terms of the sequences $\{r^{e + d - 1}\hat\alpha(r)^{1/2}\}$ and $\{(r + 2\ol{\Delta})^{d - 1}\hat\alpha(r)^{1/2} \}$, we have 
    \begin{align*}
        \frac{r^{e + d - 1}\hat\alpha(r)^{1/2}}{(r + 2\ol{\Delta})^{d - 1}\hat\alpha(r)^{1/2}}
        & = \left( \frac{r}{r + 2\ol{\Delta}} \right)^{e + d - 1}(r + 2\ol{\Delta})^e \to \infty
    \end{align*}
    as $r \to \infty$.
} 
Hence, $\Var((n h_T h_S)^{-1}\sum_{i = 1}^n X_{l,i} X_{k,i} W_{i,t,a} W_{i,s,b}) = O((n h_T^{3/2} h_S^{3/2})^{-1}) = O(n^{-1/2})$ by Assumption \ref{as:band}(i), for all $a$ and $b$.
Then, the result follows from Chebyshev's inequality.
\end{proof}


\begin{lemma}\label{lem:bias}
    Suppose that Assumptions \ref{as:D_space}, \ref{as:mixing}, \ref{as:series_approximation}, \ref{as:omega_f}--\ref{as:band}(i), and \ref{as:alpha}(i) hold.
    Then,
    \begin{align*}
        \bm{H}_{1,n} = \sqrt{n h_T h_S}\left( \frac{\varphi_2^1}{2}[h_T^2 \ddot \beta_{TT}(t,s) + h_S^2 \ddot \beta_{SS}(t,s) ] + o_P(h_T^2 + h_S^2)\right) + o_P(1).
    \end{align*}
\end{lemma}

\begin{proof}

    In view of Assumption \ref{as:omega_f}(ii) and Lemma \ref{lem:kernel_LLN}, it suffices to show that 
    \begin{align*}
        & \frac{\sqrt{h_T h_S}}{\sqrt{n}}\sum_{i=1}^n X_i(t,s) X_i^\top \gamma_i(t,s)W_{i}(t,s) \\ 
        & \quad = \sqrt{n h_T h_S} \left[ \ol{\Omega}_{1,n}(t,s)\left( \frac{\varphi_2^1}{2}[h_T^2 \ddot \beta_{TT}(t,s) + h_S^2 \ddot \beta_{SS}(t,s) ] \right) + o_P(h_T^2 + h_S^2)\right] + o_P(1).
    \end{align*}
    For $1 \le l \le dx$ and $i$ such that $|T_i - t| \le c_W h_T$ and $|S_i - s| \le c_W h_S$, by Assumption \ref{as:series_approximation},
    \begin{align*}
        \beta_l(T_i, S_i)
        & = \beta_l(t,s) + \dot \beta_{l,T}(t,s) (T_i - t) + \dot \beta_{l,S}(t,s) (S_i - s)\\
        & \quad + \frac{1}{2} \left\{ (T_i - t)^2 \ddot \beta_{l,TT}(t,s) + 2(T_i - t)(S_i - s) \ddot \beta_{l,TS}(t,s) + (S_i-s)^2 \ddot \beta_{l,SS}(t,s) \right\} + o(h_T^2 + h_S^2).
    \end{align*}
    Therefore,
    \begin{align*}
        \gamma_i(t,s) = \frac{1}{2} \left\{ (T_i - t)^2 \ddot \beta_{TT}(t,s) + 2(T_i - t)(S_i - s) \ddot \beta_{TS}(t,s) + (S_i-s)^2 \ddot \beta_{SS}(t,s) \right\} + o(h_T^2 + h_S^2),
    \end{align*}
    and, by Assumption \ref{as:band}(i),
    \begin{align*}
        & \frac{\sqrt{h_T h_S}}{\sqrt{n}}\sum_{i=1}^n X_i(t,s) X_i^\top \gamma_i(t,s)W_{i}(t,s) \\
        & = \frac{h_T^2}{2\sqrt{n h_T h_S}} \sum_{i = 1}^n W\left(\frac{T_i - t}{h_T}\right)W\left(\frac{S_i - s}{h_S}\right) \left(\frac{T_i - t}{h_T}\right)^2 \left(\begin{array}{c} 1 \\ (T_i - t)/h_T \\ (S_i - s)/h_S \end{array}\right) \otimes X_i X_i^\top \ddot \beta_{TT}(t,s) \\
        & \quad + \frac{h_T h_S}{\sqrt{n h_T h_S}} \sum_{i = 1}^n W\left(\frac{T_i - t}{h_T}\right)W\left(\frac{S_i - s}{h_S}\right) \left(\frac{T_i - t}{h_T}\right) \left(\frac{S_i - s}{h_S}\right) \left(\begin{array}{c} 1 \\ (T_i - t)/h_T \\ (S_i - s)/h_S \end{array}\right) \otimes X_i X_i^\top \ddot \beta_{TS}(t,s) \\
        & \quad + \frac{h_S^2}{2\sqrt{n h_T h_S}} \sum_{i = 1}^n W\left(\frac{T_i - t}{h_T}\right)W\left(\frac{S_i - s}{h_S}\right) \left(\frac{S_i - s}{h_S}\right)^2 \left(\begin{array}{c} 1 \\ (T_i - t)/h_T \\ (S_i - s)/h_S \end{array}\right) \otimes X_i X_i^\top \ddot \beta_{SS}(t,s) + o_P(1) \\
        & \eqqcolon \bm{H}_{11,n} + \bm{H}_{12,n} + \bm{H}_{13,n} + o_P(1), \text{ say}.
    \end{align*}
    For $\bm{H}_{11,n}$, by Assumption \ref{as:omega_f}(i),
    \begin{align*}
        & \bbE \bm{H}_{11,n} \\
        & = \frac{h_T^2}{2\sqrt{n h_T h_S}} \sum_{i = 1}^n \int_{\mathcal{TS}}W\left(\frac{t' - t}{h_T}\right)W\left(\frac{s' - s}{h_S}\right)\left(\frac{t' - t}{h_T}\right)^2 \left(\begin{array}{c} 1 \\ (t' - t)/h_T \\ (s' - s)/h_S \end{array}\right) \otimes \Omega_{1,i}(t',s')\ddot \beta_{TT}(t,s) f_i(t',s') \mathrm{d}t' \mathrm{d}s' \\
        & = \frac{h_T^2\sqrt{n h_T h_S}}{2n} \sum_{i = 1}^n \int\int W\left(\phi_1\right)W\left(\phi_2\right)\left(\begin{array}{c}  \phi_1^2 \\ \phi_1^3 \\ \phi_1^2 \phi_2 \end{array}\right) \otimes \Omega_{1,i}(t + h_T \phi_1, s + h_S \phi_2) f_i(t + h_T \phi_1, s + h_S \phi_2) \ddot \beta_{TT}(t,s) \mathrm{d}\phi_1 \mathrm{d}\phi_2 \\
        & = \frac{h_T^2\sqrt{n h_T h_S}}{2} \left(\begin{array}{c}  \ol{\Omega}_{1,n}(t,s) \ddot \beta_{TT}(t,s) \varphi_2^1  \\ \mathbf{0}_{dx \times dx} \\ \mathbf{0}_{dx \times dx} \end{array}\right) + h_T^2 \sqrt{n h_T h_S} O(h_T + h_S).
    \end{align*}
    Similarly, we can find that $\bbE \bm{H}_{12,n} = h_T h_S \sqrt{n h_T h_S} O(h_T + h_S)$, and
    \begin{align*}
        \bbE \bm{H}_{13,n} = \frac{h_S^2\sqrt{n h_T h_S}}{2} \left(\begin{array}{c}  \ol{\Omega}_{1,n}(t,s) \ddot \beta_{SS}(t,s)\varphi_2^1  \\ \mathbf{0}_{dx \times dx} \\ \mathbf{0}_{dx \times dx} \end{array}\right) + h_S^2 \sqrt{n h_T h_S} O(h_T + h_S). 
    \end{align*}
    Further, by the same argument as in the proof of Lemma \ref{lem:kernel_LLN}, we can easily show that $\Var\left(\bm{H}_{11,n}\right)$, $\Var\left(\bm{H}_{12,n}\right)$, and $\Var\left( \bm{H}_{13,n}  \right)$ are all $o(1)$.
    With Assumption \ref{as:band}(i), the lemma follows from Chebyshev's inequality. 
\end{proof}


\begin{lemma}\label{lem:matrix_LLN2}
    Suppose that Assumptions \ref{as:D_space}, \ref{as:mixing}, \ref{as:basis}(i)--(ii), \ref{as:eigen}, \ref{as:1st_stage}, \ref{as:conv_rate}(i), \ref{as:kernel}--\ref{as:eigen2}, and \ref{as:alpha}(i) hold.
    Then, 
\begin{description}
    \item[(i)] $\left\|\bm{P}_{n,UV}^\top \bm{I}_n(t,s) \bm{P}_{n,UV} /n - \bbE[ \bm{P}_{n,UV}^\top \bm{I}_n(t,s) \bm{P}_{n,UV}/n] \right\| = O_P(K_{UV}/n^{1/4}) = o_P(1)$; 
    \item[(ii)] $\left\| \hat{\bm{P}}_{n,UV}^\top \bm{I}_n(t,s) \hat{\bm{P}}_{n,UV} /n -\bbE[ \bm{P}_{n,UV}^\top \bm{I}_n(t,s) \bm{P}_{n,UV}/n] \right\| = O_P(\zeta_\dagger n^{(1/6)-\nu}) + O_P(K_{UV}/n^{1/4}) = o_P(1)$, where  $\hat{\bm{P}}_{n,UV} \coloneqq (\ol{P}_{UV}(\hat U_1, \hat V_1), \ldots, \ol{P}_{UV}(\hat U_n, \hat V_n))^\top$;
    \item[(iii)] $\left\| \bm{P}_{n,TS}^\top \bm{I}_n(t,s) \bm{P}_{n,TS} /n - \bbE[ \bm{P}_{n,TS}^\top \bm{I}_n(t,s) \bm{P}_{n,TS}/n] \right\| = O_P(K_{TS}/n^{1/4}) = o_P(1)$.
\end{description}
\end{lemma}

\begin{proof}
    (i) 
    Letting $\ol{P}_{i,UV}^{(a)}$ be the $a$-th element of $\ol{P}_{UV}(U_i, V_i)$, for $1 \le l, k \le K_{UV}$, observe that
    \begin{align*}
        & \bbE \left\|  \bm{P}_{n,UV}^\top \bm{I}_n(t,s) \bm{P}_{n,UV}/n - \bbE[ \bm{P}_{n,UV}^\top \bm{I}_n(t,s) \bm{P}_{n,UV}/n] \right\|^2 \\
    	& = \frac{1}{n^2 h_T^2 h_S^2} \sum_{l = 1}^{K_{UV}} \sum_{k = 1}^{K_{UV}} \bbE\left\{\sum_{i=1}^n \left( \ol{P}_{i,UV}^{(l)} \ol{P}_{i,UV}^{(k)} \mathbf{1}_i(t,s) - \bbE[\ol{P}_{i,UV}^{(l)} \ol{P}_{i,UV}^{(k)} \mathbf{1}_i(t,s)] \right) \right\}^2 \\
		& = \frac{1}{n^2 h_T^2 h_S^2} \sum_{l =1}^{K_{UV}} \sum_{k = 1}^{K_{UV}} \sum_{i=1}^n \mathrm{Var}\left( \ol{P}_{i,UV}^{(l)} \ol{P}_{i,UV}^{(k)} \mathbf{1}_i(t,s) \right) + \frac{1}{n^2 h_T^2 h_S^2} \sum_{l = 1}^{K_{UV}}\sum_{k = 1}^{K_{UV}} \sum_{i=1}^n \sum_{j \neq i}^n \mathrm{Cov}\left(\ol{P}_{i,UV}^{(l)} \ol{P}_{i,UV}^{(k)} \mathbf{1}_i(t,s), \ol{P}_{j,UV}^{(l)} \ol{P}_{j,UV}^{(k)} \mathbf{1}_j(t,s) \right) \\
		& \le  \frac{1}{n^2 h_T^2 h_S^2} \sum_{i=1}^n  \bbE \left[  \ol{P}_{UV}(U_i, V_i)^\top \ol{P}_{UV}(U_i, V_i) \ol{P}_{UV}(U_i, V_i)^\top \ol{P}_{UV}(U_i, V_i) \mathbf{1}_i(t,s) \right] \\
        & \quad  + \frac{1}{n^2 h_T^2 h_S^2} \sum_{l = 1}^{K_{UV}}\sum_{k = 1}^{K_{UV}} \sum_{i=1}^n \sum_{j \neq i}^n \mathrm{Cov}\left(\ol{P}_{i,UV}^{(l)} \ol{P}_{i,UV}^{(k)} \mathbf{1}_i(t,s), \ol{P}_{j,UV}^{(l)} \ol{P}_{j,UV}^{(k)} \mathbf{1}_j(t,s) \right).
    \end{align*}
    Further, for the first term on the right-hand side,
    \begin{align*}
        & \bbE \left[ \ol{P}_{UV}(U_i, V_i)^\top \ol{P}_{UV}(U_i, V_i) \ol{P}_{UV}(U_i, V_i)^\top \ol{P}_{UV}(U_i, V_i) \mathbf{1}_i(t,s) \right] \\
        & = O(K_{UV}) \cdot \int_{\mathcal{TS}} \bbE \left[  \ol{P}_{UV}(U_i, V_i)^\top \ol{P}_{UV}(U_i, V_i) \mid T_i = t', S_i = s' \right] \mathbf{1}\{(t' - t)/h_T, (s' - s)/h_S \in [-c_W, c_W]\}f_i(t',s')\text{d}t'\text{d}s' \\
        & = O(K_{UV} h_T h_S) \int_{-c_W}^{c_W} \int_{-c_W}^{c_W} \bbE \left[ \ol{P}_{UV}(U_i, V_i)^\top \ol{P}_{UV}(U_i, V_i) \mid T_i = t + h_T \phi_1, S_i = s + h_S \phi_2\right]f_i(t + h_T \phi_1,s + h_S \phi_2)\text{d}\phi_1\text{d}\phi_2 \\
        & = O(K_{UV}^2h_T h_S).
    \end{align*}
    For the second term, noting that $||\ol{P}_{i,UV}^{(l)} \ol{P}_{i,UV}^{(k)} \mathbf{1}_i(t,s)||_4^4 = O(h_T h_S)$, the same calculation as in \eqref{eq:cov_davydov} yields $\sum_{j \neq i}^n |\mathrm{Cov}(\ol{P}_{i,UV}^{(l)} \ol{P}_{i,UV}^{(k)} \mathbf{1}_i(t,s), \ol{P}_{j,UV}^{(l)} \ol{P}_{j,UV}^{(k)} \mathbf{1}_j(t,s) )| = O(\sqrt{h_T h_S})$.
    Hence, the left-hand side of the above is $O(K_{UV}^2/(n h_T^{3/2} h_S^{3/2})) = O(K_{UV}^2/\sqrt{n})$, and the result follows from Markov's inequality.
    \bigskip

    (ii)
    By result (i) and the triangle inequality,
    \begin{align*}
        & \left\|\hat{\bm{P}}_{n,UV}^\top \bm{I}_n(t,s) \hat{\bm{P}}_{n,UV} /n - \bbE[ \bm{P}_{n,UV}^\top \bm{I}_n(t,s) \bm{P}_{n,UV}/n]  \right\| \\
        & \qquad \le \left\| \hat{\bm{P}}_{n,UV}^\top \bm{I}_n(t,s) \hat{\bm{P}}_{n,UV}/n -  \bm{P}_{n,UV}^\top \bm{I}_n(t,s) \bm{P}_{n,UV}/n \right\| + O_P(K_{UV}/n^{1/4}).
    \end{align*}
    Similar to \eqref{eq:Pihat_decomp},
	\begin{align*}
		& \left\| \hat{\bm{P}}_{n,UV}^\top \bm{I}_n(t,s) \hat{\bm{P}}_{n,UV}/n - \bm{P}_{n,UV}^\top \bm{I}_n(t,s) \bm{P}_{n,UV}/n \right\| \\
		& \le \left\| \left( \hat{\bm{P}}_{n,UV}^\top - \bm{P}_{n,UV}^\top \right) \bm{I}_n(t,s) \left( \hat{\bm{P}}_{n,UV} - \bm{P}_{n,UV} \right) /n \right\| + 2 \left\| \bm{P}_{n,UV}^\top  \bm{I}_n(t,s) \left( \hat{\bm{P}}_{n,UV} - \bm{P}_{n,UV} \right) /n \right\|.
	\end{align*}
    Since $||\mathbf{1}_i(t,s)[\ol{P}_{UV}(\hat U_i, \hat V_i) - \ol{P}_{UV}(U_i, V_i)] || = O_P(\zeta_\dagger n^{-\nu})$ uniformly in $i$, we have 
    \begin{align*}
        \left\| \left( \hat{\bm{P}}_{n,UV}^\top - \bm{P}_{n,UV}^\top \right) \bm{I}_n(t,s) \left( \hat{\bm{P}}_{n,UV} - \bm{P}_{n,UV} \right) /n \right\| 
        & = O_P(\zeta_\dagger^2 n^{-2\nu}/(h_T h_S)) \\
        & = O_P(\zeta_\dagger^2 n^{(1/3)-2\nu})
    \end{align*}
    and 
    \begin{align*}
        & \left\| \bm{P}_{n,UV}^\top  \bm{I}_n(t,s) \left( \hat{\bm{P}}_{n,UV} - \bm{P}_{n,UV} \right) /n \right\|^2 \\
        & = \mathrm{tr}\left\{ \left( \hat{\bm{P}}_{n,UV} - \bm{P}_{n,UV} \right)^\top \bm{I}_n(t,s) \bm{P}_{n,UV} \bm{P}_{n,UV}^\top  \bm{I}_n(t,s) \left( \hat{\bm{P}}_{n,UV} - \bm{P}_{n,UV} \right) /n^2\right\} \\
        & \le O_P(1) \cdot \mathrm{tr}\left\{ \left( \hat{\bm{P}}_{n,UV} - \bm{P}_{n,UV} \right)^\top \bm{I}_n(t,s) \left( \hat{\bm{P}}_{n,UV} - \bm{P}_{n,UV} \right) /n\right\} = O_P(\zeta_\dagger^2 n^{(1/3)-2\nu}),
    \end{align*}
    where we have used result (i) and Assumption \ref{as:eigen2} to derive the above inequality.
    This completes the proof.
    \bigskip

    (iii) The proof is analogous to that of (i) and is omitted.
\end{proof}


\begin{lemma}\label{lem:gl_diff}
    Suppose that Assumptions \ref{as:D_space}--\ref{as:eigen2}, and \ref{as:alpha}(i) hold.
    Then, $\bm{H}_{2,n} = o_P(1)$.
\end{lemma}

\begin{proof}
    In view of Assumption \ref{as:omega_f}(ii) and Lemma \ref{lem:kernel_LLN}, it suffices to show that 
    \begin{align*}
        \frac{\sqrt{h_T h_S}}{\sqrt{n}}\sum_{i=1}^n W_{i}(t,s) X_i(t,s) [\hat g_n(T_i, S_i) \hat \lambda_n(\hat U_i, \hat V_i) - g(T_i, S_i)\lambda(U_i, V_i) ] = o_P(1).
    \end{align*}
    Decompose
    \begin{align*}
        \frac{1}{n}\sum_{i=1}^n W_{i}(t,s) X_i(t,s) [\hat g_n(T_i, S_i) \hat \lambda_n(\hat U_i, \hat V_i) - g(T_i, S_i)\lambda(U_i, V_i) ] \eqqcolon \bm{H}_{21,n} + \bm{H}_{22,n} + \bm{H}_{23,n}, 
    \end{align*}
    where
    \begin{align*}
        \bm{H}_{21,n}
        & \coloneqq \frac{1}{n} \sum_{i = 1}^n W_{i}(t,s) X_i(t,s) \hat \lambda_n(\hat U_i, \hat V_i) [\hat g_n(T_i, S_i)  - g(T_i, S_i)]\\
        \bm{H}_{22,n}
        & \coloneqq \frac{1}{n} \sum_{i = 1}^n W_{i}(t,s) X_i(t,s) g(T_i, S_i) [\hat \lambda_n(\hat U_i, \hat V_i)  - \lambda(\hat U_i, \hat V_i)]\\
        \bm{H}_{23,n}
        & \coloneqq \frac{1}{n} \sum_{i = 1}^n W_{i}(t,s) X_i(t,s) g(T_i, S_i) [\lambda(\hat U_i, \hat V_i)  - \lambda(U_i, V_i)].
    \end{align*}
    First, it is easy to see that $\bm{H}_{23,n} = O_P(n^{-\nu})$.
    We next consider $\bm{H}_{22,n}$.
    As in the proof of Lemma \ref{lem:IAE_lambda} and Theorem \ref{thm:2nd_stage}(i), we have
    \begin{align*}
        \hat \lambda_n(u,v) - \lambda(u,v) 
        & = \ol{P}_{UV}(u, v)^\top \Gamma_{TS} \bm{F}_{3,n} + O_P\left( \sqrt{(\kappa_n \ln \kappa_n)/n}\right) + O_P(\tau_n^* + b_\mu + n^{-\nu})
    \end{align*}
    uniformly in $(u,v) \in [0,1]^2$.
    Note that $W_{i}(t,s) = W_{i}(t,s)\mathbf{1}_i(t,s)$ by Assumption \ref{as:kernel}.
    By Cauchy-Schwarz inequality, Assumptions \ref{as:basis}(iii) and \ref{as:eigen2}, and Lemma \ref{lem:matrix_LLN2}(ii),
    \begin{align*}
        & \left\| \frac{1}{n}\sum_{i = 1}^n  W_{i}(t,s) X_i(t,s) g(T_i, S_i) \ol{P}_{UV}(\hat U_i, \hat V_i)^\top \Gamma_{TS} \bm{F}_{3,n} \right\|^2 \\
        & \le \frac{h_T h_S}{n} \sum_{i = 1}^n || W_{i}(t,s) X_i(t,s) g(T_i, S_i)||^2 \cdot \frac{1}{n h_T h_S} \sum_{i = 1}^n (\mathbf{1}_i(t,s) \ol{P}_{UV}(\hat U_i, \hat V_i)^\top \Gamma_{TS} \bm{F}_{3,n})^2 \\
        & = O_P(1)\cdot \frac{1}{n h_T h_S} \sum_{i = 1}^n (\mathbf{1}_i(t,s) \ol{P}_{UV}(\hat U_i, \hat V_i)^\top \Gamma_{TS} \bm{F}_{3,n})^2\\
        & = O_P(1)\cdot  \bm{F}_{3,n}^\top \Gamma_{TS}^\top \left(\hat{\bm{P}}_{n,UV}^\top \bm{I}_n(t,s) \hat{\bm{P}}_{n,UV} /n \right)\Gamma_{TS} \bm{F}_{3,n} = O_P(1)\cdot  ||\bm{F}_{3,n}||^2 = O_P((\tau_n^*)^2+ b_\mu^2 + n^{-2\nu}).
    \end{align*}
    This proves that $\bm{H}_{22,n} = O_P( \sqrt{(\kappa_n \ln \kappa_n)/n} + \tau_n^* + b_\mu + n^{-\nu})$.

    For $\bm{H}_{21,n}$, note that $\hat \lambda_n$ is uniformly bounded in view of Theorem \ref{thm:2nd_stage}(i).
    As in the proof of Theorem \ref{thm:2nd_stage}(ii), we have
    \begin{align*}
        \hat g_n(t,s) - g(t,s) 
        & = P_{TS}(t, s)^\top \Gamma_{UV} \bm{F}_{3,n} + O_P( \sqrt{(\kappa_n \ln \kappa_n)/n}) + O_P(\tau_n^* + b_\mu + n^{-\nu}).
    \end{align*}
    Then, by the same argument as above with Assumption \ref{as:eigen2} and Lemma \ref{lem:matrix_LLN2}(iii), we have $\bm{H}_{21,n} = O_P( \sqrt{(\kappa_n \ln \kappa_n)/n} + \tau_n^* + b_\mu + n^{-\nu})$.
    Finally, the result follows from Assumptions \ref{as:conv_rate}(ii) and \ref{as:band}(ii).
\end{proof}

\begin{lemma}\label{lem:beta_normal}
    Suppose that Assumptions \ref{as:D_space}--\ref{as:alpha} hold.
    Then, we have
    \begin{align*}
        & \sqrt{n h_T h_S}\left(\tilde \beta_n(t, s) - \beta(t,s) - \frac{\varphi_2^1}{2}[h_T^2 \ddot \beta_{TT}(t,s) + h_S^2 \ddot \beta_{SS}(t,s) ]\right)\\
        & \quad  \overset{d}{\to} N\left(\mathbf{0}_{dx \times 1}, (\varphi_0^2)^2 [\ol{\Omega}_1(t,s)]^{-1} \ol{\Omega}_2(t,s) [\ol{\Omega}_1(t,s)]^{-1}\right). 
    \end{align*}
\end{lemma}

\begin{proof}
    In view of Lemmas \ref{lem:kernel_LLN}, \ref{lem:bias}, and \ref{lem:gl_diff}, we only need to prove the asymptotic normality of $\bm{H}_{3,n}$.
    Let $\bm{c}$ be an arbitrary $3d_x \times 1$ vector such that $||\bm{c}|| = 1$.
    Also let 
    \begin{align*}
        \chi_i(t,s) \coloneqq \frac{1}{\sqrt{n h_T h_S}} \bm{c}^\top X_i(t,s) W\left( \frac{T_i - t}{h_T}\right) W\left( \frac{S_i - s}{h_S}\right) \tilde \epsilon_i
    \end{align*}
    so that $\sqrt{(h_T h_S)/n}\sum_{i = 1}^n  \bm{c}^\top X_i(t,s) \tilde \epsilon_i W_{i}(t,s)
    = \sum_{i = 1}^n \chi_i(t,s)$.
    Then, by Slutsky's theorem and the Cramer-Wold device, our task is further reduced to showing the asymptotic normality: 
    \begin{align*}
        \sum_{i = 1}^n \chi_i(t,s) \overset{d}{\to} N\left(0, \lim_{n \to \infty}\bm{c}^\top\ol{\Phi}_{2,n}(t,s)\bm{c}\right),
    \end{align*}
    where
    \begin{align*}
        \ol{\Phi}_{2,n}(t,s) \coloneqq \left(\begin{array}{ccc} 
            \ol{\Omega}_{2,n}(t,s) (\varphi_0^2)^2 & \bm{0}_{d_x \times d_x} & \bm{0}_{d_x \times d_x} \\
            \bm{0}_{d_x \times d_x} & \ol{\Omega}_{2,n}(t,s) \varphi_2^2 \varphi_0^2 & \bm{0}_{d_x \times d_x} \\
            \bm{0}_{d_x \times d_x} & \bm{0}_{d_x \times d_x} & \ol{\Omega}_{2,n}(t,s) \varphi_2^2 \varphi_0^2
            \end{array} \right).
    \end{align*}

    First, note that $\{\chi_i\}$ is an $\alpha$-mixing process with the mixing coefficient $\alpha^\dagger$.
    Moreover, by the law of iterated expectations, $\bbE\left[\sum_{i = 1}^n \chi_i(t,s)\right] = 0$.
    Next, we demonstrate $\Var\left(\sum_{i = 1}^n \chi_i(t,s)\right) = \bm{c}^\top\ol{\Phi}_{2,n}(t,s)\bm{c} + o(1)$ by showing $\sum_{i = 1}^n \bbE[\chi_i(t,s)^2] = \bm{c}^\top\ol{\Phi}_{2,n}(t,s)\bm{c} + o(1)$ and $\sum_{i = 1}^n\sum_{j \neq i}^n \bbE[\chi_i(t,s) \chi_j(t,s)] = o(1)$.
    Observe that
    \begin{align*}
        & \sum_{i = 1}^n \bbE[\chi_i(t,s)^2] \\
        & = \frac{1}{n h_T h_S} \sum_{i = 1}^n \bbE\left[  \bm{c}^\top \left[
            \left(\begin{array}{ccc} 
            1  &  \frac{T_i - t}{h_T} & \frac{S_i - s}{h_S} \\
            \frac{T_i - t}{h_T} &  \left(\frac{T_i - t}{h_T}\right)^2 & \left(\frac{T_i - t}{h_T}\right)\left(\frac{S_i - s}{h_S}\right) \\
            \frac{S_i - s}{h_S} & \left(\frac{T_i - t}{h_T}\right)\left(\frac{S_i - s}{h_S}\right) & \left(\frac{S_i - s}{h_S}\right)^2
            \end{array} \right) \otimes X_i X_i^\top \tilde \epsilon_i^2 \right] \bm{c} 
             W\left( \frac{T_i - t}{h_T}\right)^2 W\left( \frac{S_i - s}{h_S}\right)^2 \right] \\
        & = \frac{1}{n} \sum_{i = 1}^n \int \int \bm{c}^\top \left[
            \left(\begin{array}{ccc} 
            1  & \phi_1 & \phi_2 \\
            \phi_1 & \phi_1^2 & \phi_1 \phi_2 \\
            \phi_2 & \phi_1 \phi_2 & \phi_2^2
            \end{array} \right) \otimes \Omega_{2,i}(t + h_T \phi_1, s + h_S \phi_S) \right] \bm{c} W\left(\phi_1\right)^2 W\left(\phi_2\right)^2 f_i(t + h_T \phi_1, s + h_S \phi_2) \mathrm{d}\phi_1\mathrm{d}\phi_2 \\
        & = \bm{c}^\top\ol{\Phi}_{2,n}(t,s)\bm{c} + o(1).
    \end{align*}

    Here, by Assumptions \ref{as:mixing}(ii), \ref{as:ij}(i), and (ii), we have
    \begin{align*}
        \bbE[\chi_i(t,s) \chi_j(t,s)]
        & = \frac{1}{n h_T h_S} \int_{(\mathcal{TS})^2} \bbE[\bm{c}^\top X_i(t,s) X_j(t,s)^\top\bm{c} \tilde \epsilon_i \tilde \epsilon_j \mid T_i = t_i, S_i = s_i, T_j = t_j, S_j = s_j]   \\
        & \quad \times W\left( \frac{t_i - t}{h_T}\right) W\left( \frac{s_i - s}{h_S}\right) W\left( \frac{t_j - t}{h_T}\right) W\left( \frac{s_j - s}{h_S}\right)f_{i,j}(t_i, s_i, t_j, s_j)\mathrm{d}t_i\mathrm{d}s_i\mathrm{d}t_j\mathrm{d}s_j \\
        & \le \frac{c h_T h_S}{n} \int\int\int\int W\left( \phi_{1,i}\right) W\left( \phi_{1,j}\right) W\left(\phi_{2,i} \right) W\left( \phi_{2,j} \right)\mathrm{d}\phi_{1,i}\mathrm{d}\phi_{1,j}\mathrm{d}\phi_{2,i}\mathrm{d}\phi_{2,j}  = O((h_T h_S)/n).
    \end{align*}
    Further, observe that
    \begin{align}\label{eq:chi_qnorm}
        \begin{split}
        \left\| \chi_i(t,s) \right\|_4^4
        & = \frac{1}{n^2 h_T^2 h_S^2}\int_{\mathcal{TS}} \bbE[|\bm{c}^\top X_i(t',s')|^4 |\tilde \epsilon_i|^4 \mid T_i = t', S_i = s']W\left( \frac{T_i - t'}{h_T}\right)^4 W\left( \frac{S_i - s'}{h_S}\right)^4 f_i(t',s')\mathrm{d}t'\mathrm{d}s' \\
        & \le \frac{c}{n^2 h_T h_S} \int\int W\left( \phi_1 \right)^4 W\left( \phi_2 \right)^4 f_i(t + h_T \phi_1, s + h_S \phi_2)\mathrm{d}\phi_1\mathrm{d}\phi_2 = O((n^2 h_T h_S)^{-1})
        \end{split}
    \end{align}
    by Assumption \ref{as:error}.
    Then, by Davydov's inequality,
    \begin{align*}
        & \sum_{j \neq i}^n \left| \bbE[\chi_i(t,s) \chi_j(t,s)] \right| \\
        & = \sum_{r = 1}^\infty \sum_{1 \le j \le n: \: \Delta(i,j) \in [r, r + 1)}\left| \bbE[\chi_i(t,s) \chi_j(t,s)] \right|\\
        & = \sum_{r = 1}^{\delta_n} \sum_{1 \le j \le n: \: \Delta(i,j) \in [r, r + 1)}\left| \bbE[\chi_i(t,s) \chi_j(t,s)] \right| + \sum_{r = \delta_n + 1}^\infty \sum_{1 \le j \le n: \: \Delta(i,j) \in [r, r + 1)]}\left| \bbE[\chi_i(t,s) \chi_j(t,s)] \right|\\
        & \le \frac{c \delta_n^d h_T h_S}{n} + c \sum_{r = \delta_n + 1}^\infty \sum_{1 \le j \le n: \: \Delta(i,j) \in [r, r + 1)} \alpha^\dagger(1,1,r)^{1/2} \left\| \chi_i(t,s) \right\|_4 \left\| \chi_j(t,s) \right\|_4\\
        & \le \frac{c \delta_n^d h_T h_S}{n} + \frac{c'}{n \sqrt{h_T h_S}} \sum_{r = \delta_n + 1}^\infty \sum_{1 \le j \le n: \: \Delta(i,j) \in [r, r + 1)} \alpha^\dagger(1,1,r)^{1/2},
    \end{align*}
    where $\delta_n$ is a given sequence of integers such that $\delta_n \to \infty$ and $\delta_n^d h_T h_S = o(1)$.
    Moreover, for sufficiently large $n$ such that $\delta_n \ge 2\ol{\Delta}$,
    \begin{align*}
        \sum_{r = \delta_n + 1}^\infty \sum_{1 \le j \le n: \: \Delta(i,j) \in [r, r + 1)} \alpha^\dagger(1,1,r)^{1/2}
        & = \sum_{r = \delta_n + 1}^\infty \sum_{1 \le j \le n: \: \Delta(i,j) \in [r, r + 1)} \alpha(C \ol{\Delta}^d, C\ol{\Delta}^d,r - 2 \ol{\Delta})^{1/2} \\
        & \le \sum_{r = \delta_n + 1}^\infty \sum_{1 \le j \le n: \: \Delta(i,j) \in [r, r + 1)} [(2C \ol{\Delta}^d)^\vartheta \hat \alpha(r)]^{1/2} \\
        & \le c \sum_{r = \delta_n + 1}^\infty r^{d-1} \hat \alpha(r)^{1/2}  \\
        & \le \frac{c}{\delta_n^e} \sum_{r = \delta_n + 1}^\infty r^{e + d - 1} \hat \alpha(r)^{1/2},
    \end{align*}
    where $e$ is as given in Assumption \ref{as:alpha}(i).
    Now, we choose $\delta_n$ so that $(h_T h_S)^{-1/2}\delta_n^{-e} = O(1)$.
    Note that this is not inconsistent with $\delta_n^d h_T h_S = o(1)$ since $(h_T h_S)^{- d/(2e)}h_T h_S \to 0$ for $e > d/2$.
    Then, since $\sum_{r = \delta_n + 1}^\infty r^{e + d - 1} \hat \alpha(r)^{1/2} = o(1)$ as $n \to \infty$ under Assumption \ref{as:alpha}(i),
    \begin{align*}
        \sum_{j \neq i}^n \left| \bbE[\chi_i(t,s) \chi_j(t,s)] \right|
        & \le \frac{c \delta_n^d h_T h_S}{n} + \frac{c}{n \sqrt{h_T h_S}\delta_n^e} \sum_{r = \delta_n + 1}^\infty r^{e + d - 1} \hat \alpha(r)^{1/2} = o(n^{-1}).
    \end{align*}
    Combining these results, we have $\Var\left(\sum_{i = 1}^n \chi_i(t,s)\right) = \bm{c}^\top\ol{\Phi}_{2,n}(t,s)\bm{c} + o(1)$, as desired.

    Finally, we check the conditions for the central limit theorem of \cite{jenish2009central} (JP, hereinafter) are satisfied under the assumptions made here.
    First, Assumption 1 in JP is satisfied by Assumption \ref{as:D_space} noting the equivalence between the Euclidean metric and the maximum metric in $\mathbb{R}^d$.
    Assumptions 3(b) and 3(c) in JP are satisfied by Assumption \ref{as:mixing}(i) and \ref{as:alpha}(ii), respectively.
    Next, as shown in \eqref{eq:chi_qnorm}, $\bbE|\chi_i(t,s)|^4 = o(1)$ uniformly in $i$.
    This implies that $\lim_{k \to \infty} \sup_{1 \le i \le n} \bbE\left[ |\chi_i(t,s)|^\ell \mathbf{1}\{|\chi_i(t,s)| > k\} \right] = 0$ for some $3 \le \ell < 4$, where $\ell$ is as given in Assumption \ref{as:alpha}(iii).
    With this and Assumption \ref{as:alpha}(iii), the conditions (3) and (4) in Corollary 1 of JP are satisfied.
    This completes the proof.
\end{proof}

\begin{lemma}\label{lem:cov_estimation}
    Suppose that the assumptions in Theorem \ref{thm:catr} hold.
    If the following assumption additionally holds: $\sup_{i \in \mathcal{D}_N}\bbE[|\tilde \epsilon_i|^{8} \mid T_i, S_i] < \infty$, then we have 
    \begin{align*}
        \left\| \hat \Omega_{2,n}(t,s) - (\varphi_0^2)^2 \ol{\Omega}_2(t,s)\right\| = o_P(1).
    \end{align*}
\end{lemma}

\begin{proof}
    Let us define $\tilde \Omega_{2,n}(t,s) \coloneqq (h_T h_S /n) \sum_{i=1}^n X_i X_i^\top \tilde \epsilon_i^2 W_{i}(t,s)^2$.
    By the triangle inequality,
    \begin{align}\label{eq:decomp_omega}
        \left\| \hat \Omega_{2,n}(t,s) - (\varphi_0^2)^2 \ol{\Omega}_2(t,s) \right\| \le \left\| \hat \Omega_{2,n}(t,s) - \tilde \Omega_{2,n}(t,s) \right\| + \left\| \tilde \Omega_{2,n}(t,s) - (\varphi_0^2)^2 \ol{\Omega}_2(t,s) \right\|.
    \end{align}
    
    Observe that
    \begin{align*}
        \hat \epsilon_i(t,s)^2 
        & = \left( X_i^\top[\beta(T_i, S_i) - \tilde \beta_n(t, s)] + g(T_i,S_i) \lambda(U_i,V_i) - \hat g_n(T_i, S_i) \hat \lambda_n(\hat U_i, \hat V_i) +  \tilde \epsilon_i \right)^2.
    \end{align*}
    In view of Theorem \ref{thm:2nd_stage}, we have $\max_{1 \le i \le n}|g(T_i,S_i) \lambda(U_i,V_i) - \hat g_n(T_i, S_i) \hat \lambda_n(\hat U_i, \hat V_i)| = o_P(1)$.
    Further, by Lemma \ref{lem:beta_normal}, $|X_i^\top[\beta(T_i, S_i) - \tilde \beta_n(t, s)]| \le c ||\beta(T_i, S_i) - \tilde \beta_n(t, s)|| = O(h_T + h_S) + o_P(1)$ for any $i$ satisfying $|T_i - t| \le c_W h_T$ and $|S_i - s| \le c_W h_S$.
    Using these results, it can be easily shown that the first term of \eqref{eq:decomp_omega} is $o_P(1)$.

    Next, it is straightforward to see that
    \begin{align*}
        \bbE[\tilde \Omega_{2,n}(t,s)]
        & = \frac{h_T h_S}{n}\sum_{i = 1}^n \bbE\left[  X_i X_i^\top \tilde \epsilon_i^2 W_{i}(t,s)^2 \right] \\
        & = \frac{1}{n h_T h_S}\sum_{i = 1}^n \int_{\mathcal{TS}} \Omega_{2,i}(t',s') W\left(\frac{t' - t}{h_T} \right)^2 W\left(\frac{s' - s}{h_S} \right)^2 f_i(t',s')\mathrm{d}t'\mathrm{d}s' 
        \to (\varphi_0^2)^2\ol{\Omega}_2(t,s).
    \end{align*}
    For $1 \le l,k \le dx$, 
    \begin{align*}
        \Var \left(\frac{1}{n h_T h_S}\sum_{i = 1}^n X_{l,i} X_{k,i} \tilde \epsilon_i^2 W_{i,t,0}^2 W_{i,s,0}^2  \right)
        & = \frac{1}{n^2 h_T^2 h_S^2}\sum_{i = 1}^n \Var\left(X_{l,i} X_{k,i} \tilde \epsilon_i^2 W_{i,t,0}^2 W_{i,s,0}^2 \right) \\
        & \quad + \frac{1}{n^2 h_T^2 h_S^2}\sum_{i = 1}^n\sum_{j \neq i}^n \Cov\left(X_{l,i} X_{k,i} \tilde \epsilon_i^2 W_{i,t,0}^2 W_{i,s,0}^2, X_{l,j} X_{k,j} \tilde \epsilon_j^2 W_{j,t,0}^2 W_{j,s,0}^2\right).
    \end{align*}
    By similar calculation as in the proof of Lemma \ref{lem:kernel_LLN}, we have $\frac{1}{n^2 h_T^2 h_S^2}\sum_{i = 1}^n \Var\left(X_{l,i} X_{k,i} \tilde \epsilon_i^2 W_{i,t,0}^2 W_{i,s,0}^2\right) = O((nh_T h_S)^{-1})$.
    Moreover, since
    \begin{align*}
         \left\| X_{l,i} X_{k,i} \tilde \epsilon_i^2 W_{i,t,0}^2 W_{i,s,0}^2 \right\|_4^4
        & = \int_{\mathcal{TS}} \bbE[| X_{l,i} X_{k,i} |^4 |\tilde \epsilon_i|^8 \mid T_i = t', S_i = s']K\left( \frac{T_i - t'}{h_T}\right)^8 K\left( \frac{S_i - s'}{h_S}\right)^8 f_i(t',s')\mathrm{d}t'\mathrm{d}s' \\
        & = O(h_T h_S)
    \end{align*}
    by the additional assumption made here, $\sum_{j \neq i}^n \left| \Cov\left(X_{l,i} X_{k,i} \tilde \epsilon_i^2 W_{i,t,0}^2 W_{i,s,0}^2, X_{l,j} X_{k,j} \tilde \epsilon_j^2 W_{j,t,0}^2 W_{j,s,0}^2\right) \right| = O(\sqrt{h_T h_S})$ by the same argument as in \eqref{eq:cov_davydov}.
    Hence, $\Var((n h_T h_S)^{-1}\sum_{i = 1}^n X_{l,i} X_{k,i} \tilde \epsilon_i^2 W_{i,t,0}^2 W_{i,s,0}^2) = O((n h_T^{3/2} h_S^{3/2})^{-1}) = O(n^{-1/2})$ by Assumption \ref{as:band}(i).
    Thus, by Chebyshev's inequality, the second term of \eqref{eq:decomp_omega} is $o_P(1)$.
\end{proof}


\section{Monte Carlo Experiments}\label{sec:MC}

In this appendix, we investigate the finite sample performance of the proposed estimator.
We consider the following data-generating process (DGP):
\begin{align*}
    Y_i 
    & = \beta_1(T_i, S_i) + X_i \beta_2(T_i, S_i) + g(T_i, S_i) \epsilon_i \\
    T_i
    & = X_i \gamma_1 + Z_i \gamma_2 + \eta(U_i),
\end{align*}
where $S_i = n_i^{-1} \sum_{j \in \mathcal{P}_i} T_j$, $X_i \overset{i.i.d.}{\sim} N(0,1)$, $Z_i \overset{i.i.d.}{\sim} N(0,1)$, $U_i \overset{i.i.d.}{\sim} \text{Uniform}[0,1]$, $\epsilon_i = \varrho \lambda(U_i, V_i) + \upsilon_i$ with $\upsilon_i \overset{i.i.d.}{\sim} N(0, 0.7^2)$, and
\begin{align*}
    \beta_1(t,s) 
    & = 0.4\sin(1.5 t) + 0.2s - 0.2 \sqrt{(t^2 + s^2)/2} \\
    \beta_2(t,s) 
    & = 0.2t - 0.4\cos(1.5 s) + 0.2 \sqrt{(t^2 + s^2)/2} \\
    \lambda(u,v) 
    & = 0.3u + 0.3v + \Phi(u)\Phi(v)[1 + 0.5\Phi(u)\Phi(v)] - c_\lambda \\
    g(t,s)
    & = 0.3 t + 0.3 s + 0.2 \cos((t + s)/2) \\
    \eta(u)
    & = \sqrt{0.5 + 2u} + \Phi(u - 0.5),
\end{align*}
with $\Phi$ being the standard normal distribution function.
We set $(\gamma_1, \gamma_2) = (1,1)$.
Here, $\varrho$ indicates the magnitude of endogeneity, which is chosen from $\varrho \in \{0.3, 1, 2\}$, and $c_\lambda$ is a constant introduced to ensure $\bbE [\lambda(U,V)] = 0$, which is computed by numerical integration for each DGP.
The network $\bm{A}_N$ is created as a ring-shape network where each individual interacts with $k$-nearest neighbors ($k/2$ neighbors in front and another $k/2$ behind).
By this design, $\{(U_i,V_i)\}$ are identically jointly distributed, so that we can use the whole sample for estimation. 
We consider two values for $k$: $k \in \{2, 6\}$, and, for each setup, we generate data for two sample sizes: $n \in \{1000, 2000\}$.

The detailed computation procedure is as follows.
(\textbf{Step. 0})
We first estimate the treatment equation using a CQR method with $L = 199$ quantile values with equal step sizes ($u_1 = 0.005$, $u_2 = 0.010$, \ldots, $u_{199} = 0.995$).
To perform the CQR, we use the \textbf{R} package \textbf{cqrReg}.
We found that 199 is sufficient for this numerical exercise.
(\textbf{Step. 1})
The next step is to estimate the approximated series regression model in \eqref{eq:series_approx}.
For the basis function, we use the B-spline basis with two internal knots for all $(P_T(t), P_S(s), P_U(u), P_V(v))$ and all sample sizes.
The locations of the knots are determined based on the empirical quantiles.
The penalty matrix $\bm{D}$ is created based on the integrated second derivatives of the univariate splines; that is, for example, the first $K_{TS} \times K_{TS}$ block of $\bm{D}$ is given by the summation of $\int_\mathcal{TS}[\ddot{P}_T(t)\otimes P_S(s)][\ddot{P}_T(t)\otimes P_S(s)]^\top\mathrm{d}t\mathrm{d}s$ and $\int_\mathcal{TS}[P_T(t)\otimes \ddot{P}_S(s)][P_T(t)\otimes \ddot{P}_S(s)]^\top\mathrm{d}t\mathrm{d}s$, where $\ddot{P}_T(t) \coloneqq \partial^2 P_T(t)/(\partial t)^2$, and $\ddot{P}_S(s)$ is similar.
The penalty tuning parameter $\tau_n$ is chosen from $\tau_n \in \{1/n,5/n,10/n\}$.
In this stage, we compute not only the penalized least squares estimator in \eqref{eq:series_pen} but also an unpenalized one.
The latter will be used to estimate $\ddot \beta_{TT}(t,s)$ and $\ddot \beta_{SS}(t,s)$, which are necessary to compute the optimal $C(t,s,x)$ as given in \eqref{eq:optimal_band}.
(\textbf{Step. 2})
Once the series estimators are obtained, we next compute the estimators of $\lambda$ and $g$ based on \eqref{eq:lambda_est} and \eqref{eq:g_est}, respectively.
Here, to speed up the computation, we use the numerical integration method based on 10000 Halton sequences to approximate the integrations involved in these estimators.
(\textbf{Step. 3})
Finally, we run the local linear regression \eqref{eq:ll_reg} with Epanechnikov kernel to obtain the final estimator of $\mathrm{CATR}(t,s,x)$.
For the choice of the bandwidths, we follow the discussion in Remark \ref{rem:opt_band}, where the unknown quantities in \eqref{eq:optimal_band} are replaced by their empirical counterparts.

In this experiment, we compare the performances of three different estimators including our proposed CATR estimator.
For the other two estimators, the one is an infeasible oracle estimator that is obtained by replacing $\hat g_n(T_i, S_i)\hat \lambda_n(\hat U_i, \hat V_i)$ in \eqref{eq:ll_reg} with its true value $g(T_i, S_i)\lambda(U_i, V_i)$.
The other is an inconsistent estimator obtained from a misspecified model, which ignores the endogeneity of $(T_i, S_i)$; that is, we set $\hat g_n(T_i, S_i)\hat \lambda_n(\hat U_i, \hat V_i) = 0$ in \eqref{eq:ll_reg}.
Using these estimators, we estimate $\mathrm{CATR}(t,s,x)$ at 100 evaluation points where $(t,s,x)$ is chosen from $t \in \{t_1, t_2, \ldots, t_{50}\}$ with equal intervals, $s \in \{s_1, s_2\}$, and $x = 1$.
Here, we set $t_1 = 0.47$ and $t_{50} = 2.90$, where these numbers correspond approximately to the 20\% and 80\% quantile values of $\{T_i\}$.
Similarly, we set $s_1 \in \{ 0.84, 1.20\}$ and $s_2 = 1.70$, where $(0.84, 1.20, 1.70)$ correspond approximately to the 20\% quantile when $k = 2$, 20\% quantile when $k = 6$, and the median of $\{S_i\}$, respectively (note that the median does not vary with $k$ by symmetricity).
The number of Monte Carlo repetitions for each scenario is set to 1000, and we evaluate the estimators in terms of the AABIAS (average absolute bias) and ARMSE (average root mean squared error):
\begin{align*}
    \mathrm{AABIAS}(s,x)
    & =  \frac{1}{50}\sum_{j = 1}^{50}\left| \frac{1}{1000}\sum_{l = 1}^{1000} \hat{\mathrm{CATR}}^{(l)}_n(t_j, s, x) - \mathrm{CATR}(t_j, s, x)\right| \\
    \mathrm{ARMSE}(s,x)
    & =  \frac{1}{50}\sum_{j = 1}^{50}\left( \frac{1}{1000}\sum_{l = 1}^{1000} \left[ \hat{\mathrm{CATR}}^{(l)}_n(t_j, s, x) - \mathrm{CATR}(t_j, s, x)\right]^2 \right)^{1/2},
\end{align*}
where the superscript $(l)$ means that it is obtained from the $l$-th replicated dataset.

The results of the simulation analysis when $s$ is at the 20\% quantile and those when $s$ is at the median are respectively reported in Tables \ref{table:MC_1} and \ref{table:MC_2}.
Overall, we can observe that the proposed estimator works satisfactorily well for all scenarios considered here.
In particular, our estimator has only negligible differences in performance compared to the infeasible oracle estimator, which is consistent with the theory.
For the choice of penalty parameter $\tau_n$, in terms of AABIAS, the estimator with a greater penalty tends to be slightly less precise, as expected, while there is almost no difference in terms of ARMSE.
When comparing the results when $k = 2$ and those when $k = 6$, we can observe that the ARMSE value tends to increase in $k$ when $s$ is at the 20\% quantile, while we can observe the opposite tendency once $s$ is at the median.
This phenomenon would be understandable since the distribution of $\{S_i\}$ becomes more concentrated around the mean (i.e., the median, by symmetricity) as $k$ increases.
Finally, as expected, as the magnitude of the endogeneity increases, the inconsistent estimator exhibits larger bias, which clearly demonstrates the importance of correctly accounting for the endogeneity.

\begin{table}[!h]
    \begin{center}
    \caption{Simulation results ($s$: 20\% quantile)}
    \small\begin{tabular}{cccc|ccccc}
            \hline \hline
     & $k$ & $n$ & $\varrho$ & $\hat{\mathrm{CATR}}_n^{\text{inf}}$ & $\hat{\mathrm{CATR}}_n^1$ & $\hat{\mathrm{CATR}}_n^5$ & $\hat{\mathrm{CATR}}_n^{10}$ & $\hat{\mathrm{CATR}}_n^{\text{ex}}$ \\
        \hline
        AABIAS & 2 & 1000 & 0.3 & 0.033  & 0.032  & 0.033  & 0.034  & 0.045  \\
     & 2 & 1000 & 1 & 0.033  & 0.035  & 0.036  & 0.037  & 0.078  \\
     & 2 & 1000 & 2 & 0.033  & 0.038  & 0.040  & 0.042  & 0.134  \\
     & 2 & 2000 & 0.3 & 0.032  & 0.031  & 0.031  & 0.031  & 0.043  \\
     & 2 & 2000 & 1 & 0.032  & 0.033  & 0.033  & 0.034  & 0.077  \\
     & 2 & 2000 & 2 & 0.032  & 0.037  & 0.037  & 0.039  & 0.133  \\
     & 6 & 1000 & 0.3 & 0.041  & 0.040  & 0.043  & 0.044  & 0.056  \\
     & 6 & 1000 & 1 & 0.041  & 0.042  & 0.047  & 0.050  & 0.093  \\
     & 6 & 1000 & 2 & 0.041  & 0.046  & 0.054  & 0.059  & 0.149  \\
     & 6 & 2000 & 0.3 & 0.038  & 0.036  & 0.037  & 0.039  & 0.052  \\
     & 6 & 2000 & 1 & 0.038  & 0.038  & 0.040  & 0.043  & 0.087  \\
     & 6 & 2000 & 2 & 0.038  & 0.040  & 0.044  & 0.048  & 0.142  \\
         \hline
    ARMSE & 2 & 1000 & 0.3 & 0.117  & 0.117  & 0.117  & 0.117  & 0.121  \\
     & 2 & 1000 & 1 & 0.117  & 0.119  & 0.119  & 0.119  & 0.145  \\
     & 2 & 1000 & 2 & 0.117  & 0.122  & 0.122  & 0.123  & 0.196  \\
     & 2 & 2000 & 0.3 & 0.085  & 0.086  & 0.086  & 0.086  & 0.091  \\
     & 2 & 2000 & 1 & 0.085  & 0.087  & 0.087  & 0.088  & 0.117  \\
     & 2 & 2000 & 2 & 0.085  & 0.091  & 0.091  & 0.092  & 0.169  \\
     & 6 & 1000 & 0.3 & 0.127  & 0.126  & 0.127  & 0.128  & 0.133  \\
     & 6 & 1000 & 1 & 0.127  & 0.128  & 0.129  & 0.131  & 0.158  \\
     & 6 & 1000 & 2 & 0.127  & 0.132  & 0.135  & 0.137  & 0.209  \\
     & 6 & 2000 & 0.3 & 0.094  & 0.094  & 0.094  & 0.095  & 0.102  \\
     & 6 & 2000 & 1 & 0.094  & 0.095  & 0.096  & 0.097  & 0.128  \\
     & 6 & 2000 & 2 & 0.094  & 0.098  & 0.099  & 0.101  & 0.178  \\
     \hline \hline
    \end{tabular}
    \label{table:MC_1}
    \end{center}
    \small
    Notes: $\hat{\mathrm{CATR}}_n^{\text{inf}}$: infeasible oracle CATR estimator; $\hat{\mathrm{CATR}}_n^a$: proposed CATR estimator with penalty $\tau_n = a/n$;  $\hat{\mathrm{CATR}}_n^{\text{ex}}$: inconsistent CATR estimator under a misspecified model that treats $(T_i, S_i)$ as exogenous.
    \normalsize
\end{table}

\begin{table}[!h]
    \begin{center}
    \caption{Simulation results ($s$: median)}
    \small\begin{tabular}{cccc|ccccc}
            \hline \hline
     & $k$ & $n$ & $\varrho$ & $\hat{\mathrm{CATR}}_n^{\text{inf}}$ & $\hat{\mathrm{CATR}}_n^1$ & $\hat{\mathrm{CATR}}_n^5$ & $\hat{\mathrm{CATR}}_n^{10}$ & $\hat{\mathrm{CATR}}_n^{\text{ex}}$ \\
        \hline
AABIAS & 2 & 1000 & 0.3 & 0.055  & 0.055  & 0.056  & 0.057  & 0.071  \\
 & 2 & 1000 & 1 & 0.055  & 0.058  & 0.059  & 0.060  & 0.108  \\
 & 2 & 1000 & 2 & 0.055  & 0.062  & 0.063  & 0.064  & 0.162  \\
 & 2 & 2000 & 0.3 & 0.044  & 0.044  & 0.044  & 0.045  & 0.060  \\
 & 2 & 2000 & 1 & 0.044  & 0.047  & 0.047  & 0.048  & 0.098  \\
 & 2 & 2000 & 2 & 0.044  & 0.051  & 0.051  & 0.052  & 0.153  \\
 & 6 & 1000 & 0.3 & 0.043  & 0.041  & 0.043  & 0.044  & 0.058  \\
 & 6 & 1000 & 1 & 0.043  & 0.042  & 0.044  & 0.046  & 0.094  \\
 & 6 & 1000 & 2 & 0.043  & 0.043  & 0.046  & 0.048  & 0.147  \\
 & 6 & 2000 & 0.3 & 0.041  & 0.039  & 0.039  & 0.040  & 0.056  \\
 & 6 & 2000 & 1 & 0.041  & 0.040  & 0.040  & 0.042  & 0.092  \\
 & 6 & 2000 & 2 & 0.041  & 0.041  & 0.042  & 0.044  & 0.144  \\
 \hline
ARMSE & 2 & 1000 & 0.3 & 0.132  & 0.132  & 0.133  & 0.133  & 0.140  \\
 & 2 & 1000 & 1 & 0.132  & 0.134  & 0.135  & 0.135  & 0.168  \\
 & 2 & 1000 & 2 & 0.132  & 0.140  & 0.140  & 0.141  & 0.218  \\
 & 2 & 2000 & 0.3 & 0.095  & 0.096  & 0.096  & 0.096  & 0.105  \\
 & 2 & 2000 & 1 & 0.095  & 0.098  & 0.098  & 0.098  & 0.134  \\
 & 2 & 2000 & 2 & 0.095  & 0.102  & 0.102  & 0.103  & 0.186  \\
 & 6 & 1000 & 0.3 & 0.116  & 0.116  & 0.116  & 0.116  & 0.123  \\
 & 6 & 1000 & 1 & 0.116  & 0.117  & 0.117  & 0.118  & 0.149  \\
 & 6 & 1000 & 2 & 0.116  & 0.120  & 0.121  & 0.122  & 0.197  \\
 & 6 & 2000 & 0.3 & 0.090  & 0.090  & 0.090  & 0.090  & 0.098  \\
 & 6 & 2000 & 1 & 0.090  & 0.091  & 0.091  & 0.092  & 0.126  \\
 & 6 & 2000 & 2 & 0.090  & 0.094  & 0.094  & 0.095  & 0.176  \\
 \hline \hline
\end{tabular}
\label{table:MC_2}
\end{center}
\small
Notes: $\hat{\mathrm{CATR}}_n^{\text{inf}}$: infeasible oracle CATR estimator; $\hat{\mathrm{CATR}}_n^a$: proposed CATR estimator with penalty $\tau_n = a/n$;  $\hat{\mathrm{CATR}}_n^{\text{ex}}$: inconsistent CATR estimator under a misspecified model that treats $(T_i, S_i)$ as exogenous.
\end{table}

\clearpage

\section{Supplementary Information for the Analysis in Section \ref{sec:empiric}}\label{app:empiric}

\subsection{Supplementary Statistical Information} \label{app:empir_table}

\begin{table}[!h]
    \footnotesize
    \begin{center}
    \caption{Descriptive Statistics}
    \begin{tabular}{l|cccc|cccc}
    \hline\hline
    Sample & \multicolumn{4}{c|}{All data: $N=1867$}  & \multicolumn{4}{c}{$k = 2$: $n = 1773$} \\
    & \multicolumn{4}{c|}{(used for treatment equation estimation)}  & \multicolumn{4}{c}{(used for CATR estimation)} \\
    Variable & Mean & Std. Dev. & Min. & Max. & Mean & Std. Dev. & Min. & Max. \\
    \hline
    \textit{crime} & \multicolumn{4}{c|}{(not used)} & 1.175  & 0.916  & 0.000  & 16.591  \\
    \textit{unempl} & 5.656  & 2.156  & 0.000  & 22.416  & 5.697  & 2.131  & 0.952  & 22.416  \\
    $\ol{\textit{unempl}}$ & \multicolumn{4}{c|}{(not used)}  & 5.721  & 1.954  & 1.510  & 18.031 \\
    \textit{density} & 1.006  & 1.921  & -3.961  & 6.597  & 1.036  & 1.938  & -3.961  & 6.597\\
    \textit{sales} & 0.288  & 0.824  & -2.704  & 6.928  & 0.305  & 0.823  & -2.704  & 6.928  \\
    \textit{elderly} & 0.246  & 0.070  & 0.085  & 0.534  & 0.245  & 0.070  & 0.091  & 0.534  \\
    \textit{single} & 0.243  & 0.083  & 0.069  & 0.687  & 0.240  & 0.081  & 0.069  & 0.637   \\
    \textit{childcare} & 0.734  & 0.512  & 0.000  & 5.587  & \multicolumn{4}{c}{(not used)} \\
\hline
\end{tabular}
\label{table:descstat}
\end{center}
\end{table}

\begin{table}[!h]
    \footnotesize
    \begin{center}
    \caption{Results of the Treatment Equation Estimation}
    \begin{tabular}{l|ccc|ccc}
    \hline\hline
    & \multicolumn{3}{c|}{Ordinary least squares regression} & \multicolumn{3}{c}{Composite quantile regression ($L = 399$)} \\
    & Coef. & Std. Err. & $t$-value & Coef. & Std. Err. & $t$-value \\
\hline
    Intercept & 3.637 & 0.290 & 12.523 &  &  &  \\
    \textit{density} & 0.591  & 0.039  & 15.058  & 0.496  & 0.026  & 19.013  \\
    \textit{sales} & -0.488  & 0.073  & -6.652  & -0.211  & 0.061  & -3.477  \\
    \textit{elderly} & 5.315  & 1.037  & 5.128  & 4.905  & 0.529  & 9.273  \\
    \textit{single} & 1.743  & 0.611  & 2.854  & -0.347  & 0.480  & -0.723  \\
    \textit{childcare} & -0.224  & 0.106  & -2.127  & -0.425  & 0.086  & -4.922  \\
\hline
\end{tabular}
\label{table:cqr}
\end{center}
\end{table}

\begin{figure}[!ht]
	\centering
	\includegraphics[bb = 0 0 548 343, width = 10cm]{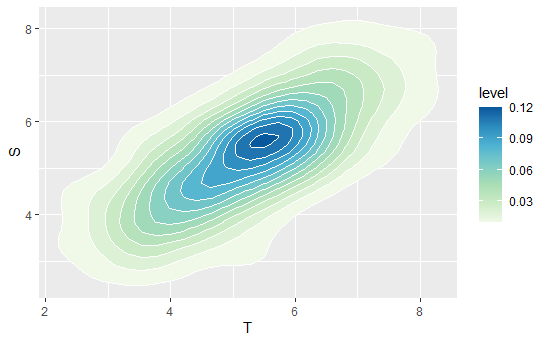}
	\caption{Joint density of $\textit{unempl}$ ($T$) and $\ol{\textit{unempl}}$ ($S$)}
	\label{fig:joint_density}
\end{figure}

\subsection{Robustness Checks} \label{app:empir_robust}

In Figure \ref{fig:catr_all}, we provide the CATR estimation results more detailed than Figure \ref{fig:CATR} in the main body, where we evaluate $\ol{\textit{unempl}}$ ($S$) at 0.2, 0.3, \ldots, 0.8 empirical quantiles, from the top row to the bottom in this order.
(Note that, different from the figure in the manuscript, the left column corresponds to rural areas, the middle is when $x$ is at the sample median, and the right is urban areas.)
In each panel, the dashed lines indicate the 95\% confidence interval.
These results clearly show that in average and urban cities, the CATR value likely increases as the unemployment rate of surrounding cities increases, whereas this tendency is not prominent in rural cities.

Figure \ref{fig:robustWmat} presents the CATR estimates based on five different interaction matrices.
The first panel (a) corresponds to Figure \ref{fig:CATR} reported in the manuscript.
Panel (b) uses the same interaction matrix definition as the one in panel (a), except that the maximum number of interacting cities is set to $k = 3$.
The other three panels are the results obtained when the interaction matrix is determined through the geographical distance between city centers (not through the adjacency).
Specifically, each element of $\bm{A}_N$ is given by $A_{i,j} = 1$ if $j$ is within the $k$-nearest-neighbor range from $i$.
We report the results when we set $k = 2$, $5$, and $10$ in panels (c), (d), and (e), respectively.
These results may corroborate our finding that, in general, the spillover effect is important especially in non-rural areas.

\begin{figure}[!ht]
	\centering
	\includegraphics[bb = 0 0 708 980, width = 14cm]{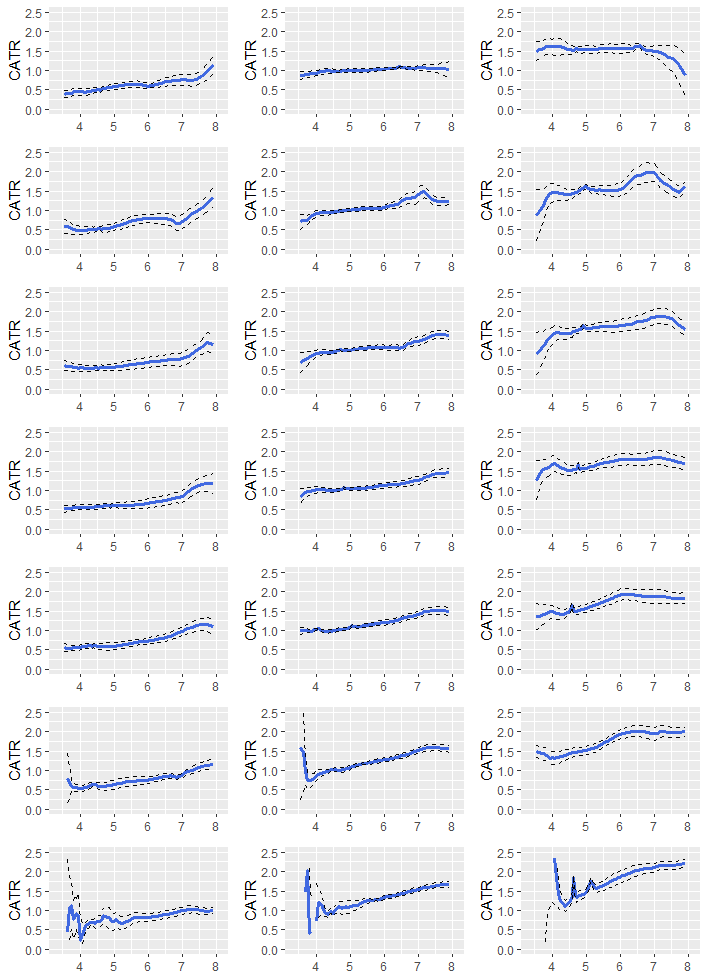}
	\caption{CATR estimation results at $S = $ 0.2, 0.3, \ldots, 0.8 quantiles.}
    \footnotesize
    (Left: $x = \text{median} \mid \textit{density} \le \text{20\% value}$, middle: $x = \text{median} \mid \text{all sample}$, right: $x = \text{median} \mid \textit{density} \ge \text{80\% value}$)
	\label{fig:catr_all}
\end{figure}

\begin{figure}[!ht]
        \begin{minipage}[b]{0.9\linewidth}
         \centering
         \includegraphics[bb = 0 0 1065 334, width = 12cm]{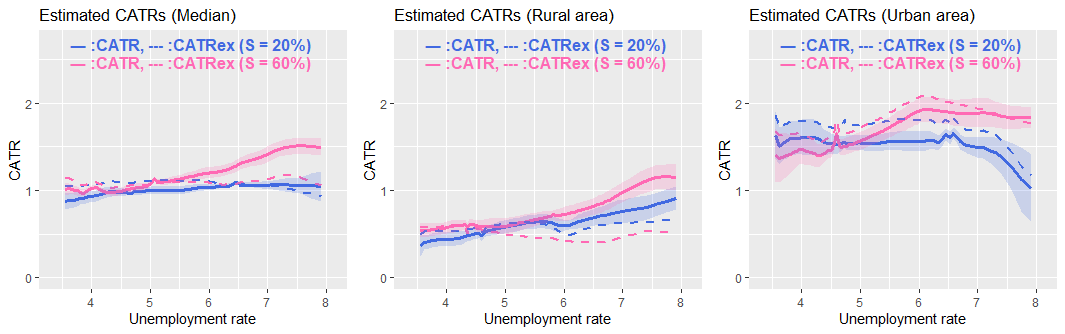}
         \subcaption{Adjacency $\bm{A}_N$ with $k = 2$}\label{W2neigh}
        \end{minipage}\\
        \begin{minipage}[b]{0.9\linewidth}
        \centering
        \includegraphics[bb = 0 0 1065 334, width = 12cm]{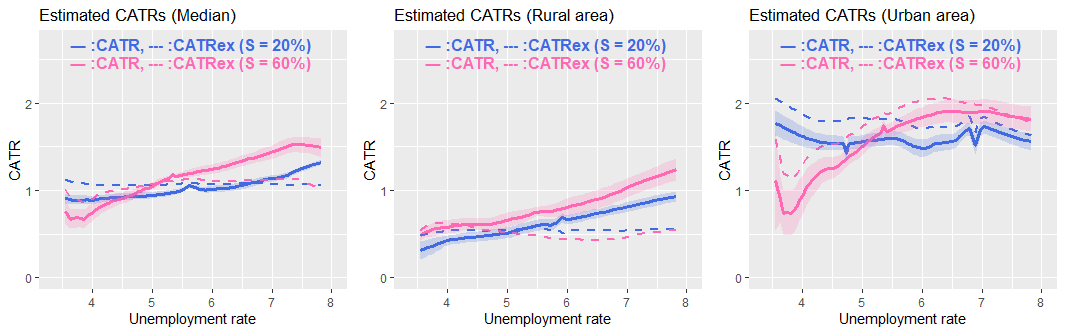}
        \subcaption{Adjacency $\bm{A}_N$ with $k = 3$}\label{W3neigh}
        \end{minipage}\\
        \begin{minipage}[b]{0.9\linewidth}
        \centering
        \includegraphics[bb = 0 0 1065 334, width = 12cm]{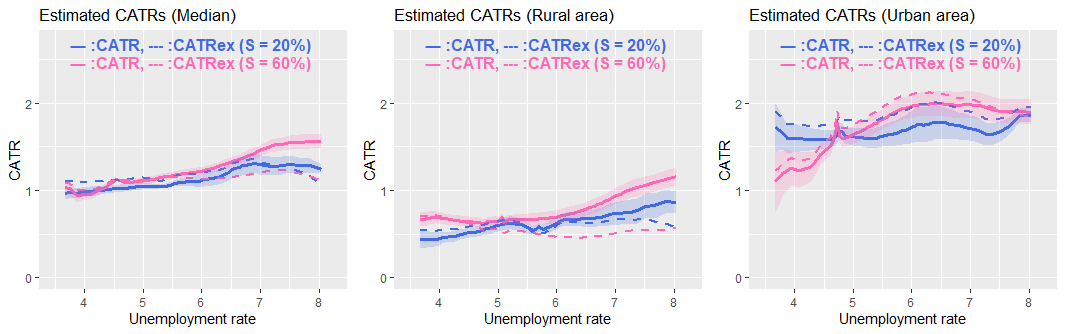}
        \subcaption{Distance $\bm{A}_N$ with $k = 2$}\label{Wd2}
        \end{minipage}\\
        \begin{minipage}[b]{0.9\linewidth}
        \centering
        \includegraphics[bb = 0 0 1065 334, width = 12cm]{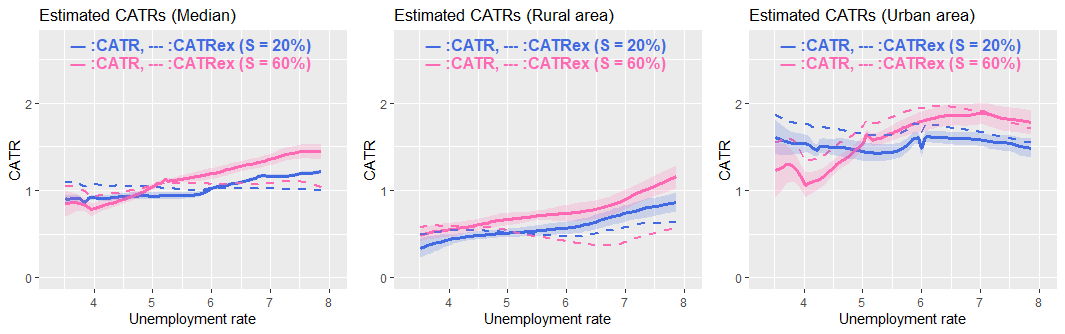}
        \subcaption{Distance $\bm{A}_N$ with $k = 5$}\label{Wd5}
        \end{minipage}\\
        \begin{minipage}[b]{0.9\linewidth}
        \centering
        \includegraphics[bb = 0 0 1065 334, width = 12cm]{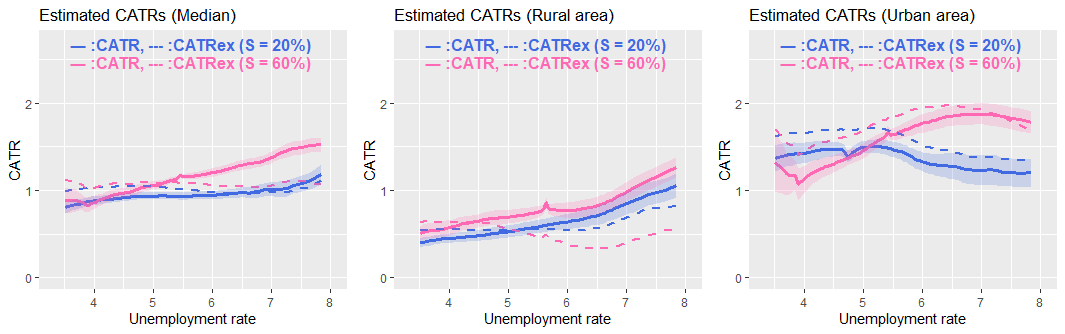}
        \subcaption{Distance $\bm{A}_N$ with $k = 10$}\label{Wd10}
        \end{minipage}\\
	\caption{CATR estimation with different neighborhood definitions}
	\label{fig:robustWmat}
\end{figure}

}
\end{document}